\documentclass{article}

\usepackage{graphicx}

\usepackage[hmargin=1in,vmargin=1in]{geometry}
\usepackage{hyperref}
\usepackage{hchang}

\usepackage{tcolorbox}
\usepackage{url}
\usepackage{graphicx}
\usepackage{caption} 
\usepackage{bbm}
\usepackage{float}
\usepackage{framed}
\usepackage{enumerate}
\usepackage{color}
\usepackage{comment}
\usepackage{thmtools}
\usepackage{thm-restate}
\usepackage{graphicx}
\usepackage{subcaption}

\usepackage{cleveref}
\usepackage{subfiles}

\usepackage{algorithm}
\usepackage[noend]{algpseudocode}
\usepackage{pdfpages}
\usepackage{subcaption}
\usepackage{marvosym}
\usepackage{xspace}
\usepackage[framemethod=TikZ]{mdframed}

\newtheorem{theorem}{Theorem}
\newtheorem{problem}{Problem}

\newtheorem{lemma}{Lemma}
\newtheorem{claim}{Claim}
\newtheorem{definition}{Definition}

\newtheorem{remark}{Remark}
\newtheorem{observation}[theorem]{\textbf{Observation}}
\newcommand{\bdry}{\partial}
\newcommand{\subseq}{\sqsubseteq}
\newcommand{\congest}{CONGEST\xspace}

\newcommand{\cutgraph}{\textrm{\CutLeft}}
\newcommand{\cut}{\cutgraph}

\mdfdefinestyle{lpbox}{
    linecolor=black,
    linewidth=1pt,
    roundcorner=5pt,
    innertopmargin=10pt,
    innerbottommargin=10pt,
    innerleftmargin=10pt,
    innerrightmargin=10pt,
    backgroundcolor=gray!5 
}

\newcommand{\eps}{\epsilon}
\newcommand{\OO}{\Tilde{O}}

\newif\ifshowcomments
\showcommentstrue 

\title{Planar Embedding of Okamura-Seymour Quasimetrics in Polynomial Time with an Application to Distributed SSSP}
\author{
  Hung Le%
  \thanks{University of Massachusetts, Amherst, USA. \texttt{hungle@cs.umass.edu}}
  \and%
  Hector Tierno%
  \thanks{University of Massachusetts, Amherst, USA. \texttt{htierno@umass.edu}}
  \and%
  Shuang Yang 
  \thanks{University of Massachusetts, Amherst, USA. \texttt{shuangyang@umass.edu}}
}

\date{}

\begin{document}

\maketitle

\begin{abstract}
    A quasi-metric $(T,\delta_T)$ is an \EMPH{Okamura-Seymour quasimetric} if there exists an edge-weighted planar embedded directed graph $G = (V,E,w)$ such that $T$ is  a set of terminals on the outerface of $G$ and $\delta_G(t,t') = \delta_T(t,t')$ for every pair $(t,t')\in T\times T$. If   $(T,\delta_T)$ is an Okamura-Seymour quasimetric, then $G$ is a \EMPH{planar embedding} of $(T,\delta_T)$.

    In a recent pioneering work, Chen and Tan~\cite{ChenTan2025} gave a polynomial-time algorithm to \EMPH{test} if a given quasi-metric $(T,\delta_T)$ is an Okamura-Seymour quasimetric. A key step in their proof is existential, which suffices for an efficient testing algorithm but does not imply an efficient embedding algorithm. Our paper closes this gap by giving an algorithmic implementation of their existential step via linear programming. As a result, we obtain the first polynomial-time algorithm for finding a planar embedding of any given Okamura-Seymour quasimetric $(T,\delta_T)$. 

    As an application, we show how to use our planar embedding of Okamura-Seymour quasimetrics to compute a $(1+\eps)$-approximate single-source shortest path (SSSP) in planar directed graphs in the distributed \congest model in $\Tilde{O}(D)$ rounds for any fixed $\eps\in (0,1)$, nearly matching a simple lower bound of $\Omega(D)$ and resolving a fundamental problem in this area. The best-known algorithm for this problem has round complexity $\Tilde{O}(D^2)$~\cite{li2019planar}. 
\end{abstract}

\section{Introduction}

\paragraph{Planar and minor-free metrics} are the shortest path metrics of planar and minor-free graphs. Research on the geometry of these families of metrics over the last few decades has produced deep results with a wide range of applications, including the KPR theorem~\cite{KPR93}, the $\ell_1$ embedding conjecture~\cite{gnrs99}, embedding into metrics of small treewidth~\cite{FKS19,FL22,CLPP23,CCCLPP25}, and the tree cover theorem~\cite{CCLMST23a,CCLMST24}, to name a few.  These results assume that the input metric is given indirectly, via the graph whose shortest path metric represents the input metric. In many practical scenarios, the input metric is given explicitly, by a metric of pairwise distances (for example, computed from features of objects as in machine learning applications), and we are interested in whether the metric comes from a special family so that we can exploit its structures for downstream applications. Such a \EMPH{metric recognition}, notably recognizing tree metrics~\cite{Zaretskii1965,buneman1971,Gordon1987,Bandelt1990,Batagelj1990}, has been studied since 70s. A basic open problem asked by  Linial and Rabinovic~\cite{matouvsek2002open} and mentioned by others~\cite{CC08,KR08}  is \EMPH{planar testing of metric spaces}: ``Is NP-complete or polynomial to decide if a finite metric space is planar?'' In applications, one is more interested in constructing a planar embedding of a metric space, but testing planarity is often viewed as the first step toward that goal. In sharp contrast, for graphs, Hopcroft and Tarjan~\cite{HT74} gave the first linear-time algorithm to construct a planar embedding; their algorithm opened the door to an entire area of algorithm design for planar graphs.  

There has been progress towards recognizing and embedding subclasses of planar metrics, including tree metrics~\cite{Zaretskii1965,buneman1971,Bandelt1990,Batagelj1990}, ultrametrics~\cite{Gordon1987}, cactus metrics~\cite{HHMM20}, and  Okamura-Seymour metrics~\cite{hurkens1988tidy,CO20}. Perhaps the most significant result is the recognition of Okamura-Seymour metrics, whose points are vertices on the outerface of an edge-weighted planar embedded graph $G$, and their distances are the graph distances in $G$.    Okamura-Seymour metrics have a flow-cut gap of exactly 1~\cite{OS81} and play a central role in other problems such as metric compression~\cite{krauthgamer2013mimicking,abboud2018near,li2019planar}, face cover and $\ell_1$-embedding~\cite{krauthgamer2019flow,filtser2020face,kumar2025approximate}. 

\paragraph{Planar and minor-free quasimetrics} are directed analogues of planar and minor-free metrics. We say that $(X,\delta_X)$ is a \EMPH{quasimetric} if $X$ is a set of points, and the distance function $\delta_X: X\times X \rightarrow \real_{\geq 0}\cup \{+\infty\}$ satisfies (i) $\delta_X(x,x) = 0$ for every $x\in X$ and (ii) $\delta_X(x,z) \leq \delta_X(x,y) +\delta_X(y,z)$. Note that the difference between a metric and a quasi-metric is the symmetry: for a quasi-metric $\delta_X$, it is possible that $\delta_X(x,y) \not= \delta_X(y,x)$. Naturally, quasi-metrics are exactly the shortest path metrics of  edge-weighted directed graphs, a.k.a, \EMPH{digraphs}. Digraphs are central in graph algorithms, as they can model one-way relationships between objects that an (undirected) graph cannot, such as one-way streets in road networks. However,  their geometry, namely quasimetrics, is extremely difficult to understand. For example, the directed analogue of the KPR theorem~\cite{KPR93}, called \EMPH{quasi-partitions}~\cite{KS21}, remains wide open: the best Lipschitz constant\footnote{We refer readers to \cite{KS21} for precise definitions of quasipartitions and Lipschitz constant of quasimetrics.} for planar quasimetrics is $O(\log ^2n )$ shown by Kawarabayashi and Anastasios Sidiropoulos~\cite{KS21}. It remains an outstanding open problem if $O(1)$ is achievable, which, if true, would match the undirected counterpart.  Despite only a few known results on the geometry of planar/minor-free quasimetrics~\cite{memoli2018quasimetric,KS21,le2023vc,karczmarz2025}, their algorithmic applications are already substantial, including directed cut problems~\cite{KS21}, computing diameter~\cite{le2023vc}, and exact distance oracles~\cite{karczmarz2025}. 

Unsurprisingly, constructing a planar embedding of quasimetrics is much harder than constructing a planar embedding of metrics, given the difficulty of understanding quasimetrics. Even the very fundamental case of constructing a planar embedding of Okamura-Seymour quasimetrics is not fully understood. We say that a quasi-metric $(T,\delta_T)$ is an \EMPH{Okamura-Seymour quasimetric} if there exists an edge-weighted planar embedded digraph $G = (V,E,w)$ such that (i) $T\subseteq V$ and is  a set of terminals on the outerface of $G$ and (ii) $\delta_G(t,t') = \delta_T(t,t')$ for every pair $(t,t')\in T\times T$. We call $G$ a \EMPH{planar embedding} of $(T,\delta_T)$. In this work, we study the problem of constructing a planar embedding of Okamura-Seymour quasimetrics.

\begin{problem}[Okamura-Seymour Quasimetric Planar Embedding Problem]\label{problem:main}Given  Okamura-Seymour quasimetric $(T,\delta_T)$, construct a planar embedding of $(T,\delta_T)$ in polynomial time. 
\end{problem}
That is, construct an edge-weighted, planar embedded digraph $G$ and the placement of $T$ on the outerface of the embedding of $G$ such that $(G,T)$ realizes $(T,\delta_T)$ in polynomial time. 

In a recent pioneering work, Chen and Tan~\cite{ChenTan2025} gave an efficient solution for the testing problem:  given a quasimetric $(T,\delta_T)$, test if  $(T,\delta_T)$ is an  Okamura-Seymour quasimetric.  It is often the case that a solution for the testing problem can be readily turned into a solution for the embedding problem. (The planarity testing algorithm by Hocroft and Tarjan~\cite{HT74} is a famous example.) However, for Okamura-Seymour quasimetrics, this is not the case: a key step (see \Cref{subsec:ideas}) in the planar embedding framework of Chen and Tan's proof is existential, which, interestingly, suffices for the testing problem. 

Our main result in this work is a polynomial-time algorithm for constructing a planar embedding of any given Okamura-Seymour quasimetric, fully resolving \Cref{problem:main}.

\begin{theorem}\label{thm:main} Let $(T,\delta_T)$ be an Okamura-Seymour quasimetric where the distance between any two terminals is an integer in $\{0,1,\ldots, W\}$. We can construct a planar embedding of $(T,\delta_T)$ in time $\poly(|T|,\log(W))$.
\end{theorem}

Here $\poly(\cdot)$ is a polynomial function of (one or more) arguments. The exponent of the polynomial in \Cref{thm:main} is rather large, of about $15$. A key bottleneck is that we will be writing an (exponentially) large linear program and using the ellipsoid method to solve it.  We make no attempt to optimize the running time. 

Our technique for proving \Cref{thm:main} can also be extended to solve a more general problem: embedding arbitrary quasimetrics into Okamura-Seymour quasimetrics with minimum distortion following a canonical ordering. Specifically, given a (not necessarily Okamura-Seymour) quasi-metric $(T,\delta_T)$, and a permutation $\sigma_T: T\rightarrow T$, an \EMPH{Okamura-Seymour embedding of $(T,\delta_T)$ following $\sigma_T$} is an Okamura-Seymour instance $(G,T)$ such that (i) for very pair of terminals $(t,t')\in T$, $\delta_G(t,t')\geq \delta_T(t,t')$, and (ii) the clockwise ordering of $T$ along the outerface of $G$ is exactly $\langle \sigma_T[1], \sigma_T[2],\ldots, \sigma_T[|T|] \rangle$. As usual, the \EMPH{distortion} of the Okamura-Seymour embedding is defined as $\max_{t\not= t' \in T}\frac{\delta_G(t,t')}{\delta_T(t,t')}$.  We show that finding an Okamura-Seymour embedding with minimum distortion following a canonical ordering can be solved in polynomial time.

\begin{restatable}{theorem}{OSEmbedding}\label{thm:os-embedding} Given a quasi-metric $(T,\delta_T)$, and a permutation $\sigma_T: T\rightarrow T$, there is a polynomial time algorithm to find an Okamura-Seymour embedding $(G, T)$ of $(T,\delta_T)$ following $\sigma_T$ such that the distortion is minimum.
\end{restatable}

If we are not given the ordering $\sigma_T$ as part of the input, then the problem becomes to find an Okamura-Seymour embedding of  $(T,\delta_T)$ with minimum distortion. We show in \Cref{thm:embedding-hardness} below that this problem is NP-hard, even for the undirected case (the proof is in \Cref{subsec:hardness}). 

\begin{restatable}{theorem}{OSEmbeddingNPhardness}\label{thm:embedding-hardness} 
Given an $n$-point finite metric $(X,d)$, and any constant $1<D<\frac{1+\sqrt{5}}{2}$, it is NP-hard to decide  whether there exists an Okamura-Seymour metric $(X,\rho)$ such that
    \[
    d(x,y)\le \rho(x,y)\le D \cdot d(x,y)
    \qquad \forall x,y\in X.
    \]
\end{restatable}

The NP-hardness in \Cref{thm:embedding-hardness} is perhaps not surprising, given that analogous problems for trees metrics~\cite{ABFPT98} and line metrics~\cite{badoiu2005,badoiuSTOC05} are NP-hard. (All of these problems have polynomial time algorithms when the distortion is exactly $1$.) However, what we found surprising is that if one can find a canonical ordering, one can bypass the NP-hardness by  \Cref{thm:os-embedding}. This is exactly what we will do in an algorithmic application of our embedding results discussed in \Cref{subsec:SSSP}. Specifically, we manage to get the canonical ordering $\sigma_T$ in this specific setting and get a polynomial-time algorithm.

\subsection{An application: distributed approximate SSSP in planar digraphs in $\Tilde{O}(D)$ round}\label{subsec:SSSP}

The distributed model we study is the \congest model where the communication graph is a (possibly edge-weighted and directed) graph $G = (V,E)$ with a hop diameter\footnote{The hop diameter is the diameter of the unweighted and undirected version $G$.} $D$. In every (synchronous) round, $O(\log n) = \Tilde{O}(1)$ bits can be sent along each edge in any direction, i.e., the communication graph is bidirectional (c.f.~\cite{Parter20}). (We use $\Tilde{O}$ notation to suppress polylog factors.) The local computation time at each node, though not explicitly specified in the model, should be polynomial. 

A recent line of work in the distributed \congest model that has attracted a lot of attention focuses on solving fundamental problems on planar graphs. The starting point is  the pioneering work of  Ghaffari and Haeupler~\cite{GH16} who gave a distributed algorithm for planar embedding in $\OO(D)$ rounds, matching a trivial lower bound $\Omega(D)$ which holds for a vast majority of distributed problems in planar graphs. The central goal of this line of research is to design distributed \congest algorithms for planar (di)graph problems in $\OO(D)$ rounds. 

\paragraph{Undirected planar graphs.} Over the past decade, there has been significant progress on undirected planar graphs where $\OO(D)$-round algorithms are known for many basic problems, including MST~\cite{GH16b}, global mincut~\cite{GH16,GZ22}, approximate single-source shoretst path SSSP~\cite{ZGYHS22,RGHZL22}, DFS~\cite{GP17,JMR25},  cycle separators~\cite{GP17,JMR25,ADW26}, weighted girth~\cite{ADPW25}, and recently maxflow~\cite{ADPW25} (in slightly larger $D \cdot n^{o(1)}$ rounds), to name a few.  Along with these results are deep technical innovations, for example, congestion shortcuts~\cite{GH16}, minor aggregation~\cite{ZGYHS22,RGHZL22}, and metric compression~\cite{li2019planar}.

\paragraph{Planar digraphs.} Solving distributed problems on digraphs is much harder, and the progress has been slower. For many problems, $\Tilde{O}(D^2)$  rounds remain the state of the art, including maxflow~\cite{ADPW25} and global mincut~\cite{ADPW25}. Reachability and strongly connected components (SCCs) in planar digraphs are the only known nontrivial problems which admit $\OO(D)$ rounds~\cite{Parter20}, to the best of our knowledge. 

Here, we study approximate SSSP in planar digraphs, a fundamental problem in the distributed \congest model. We give the first algorithm for solving this problem in $\OO(D)$ rounds. 

\begin{restatable}{theorem}{Distributed}\label{thm:distributed} Let $G=(V,E,w)$ be an edge-weighted, planar digraph where every edge has an integer weight in $\{0,1,\ldots, n^{O(1)}\}$. Let $s$ be any vertex in $V$, $D$ be the hop diameter of $G$, and $\epsilon\in (0,1)$ be any parameter. There is a deterministic algorithm in the \congest model that finds a $(1+\epsilon)$-SSSP rooted at $s$ and runs $\OO(D/\epsilon)$ rounds and polynomial local time.
\end{restatable}

Now we give a very brief overview of how we use \Cref{thm:os-embedding} to prove \Cref{thm:distributed}. Like other distributed algorithms in planar (di)graphs~\cite{li2019planar,Parter20,ADPW25}, our algorithm uses the bounded diameter decomposition (BDD)~\cite{li2019planar}.  Starting with the input graph $G$, we find a cycle balanced separator of length $O(D)$ to separate $G$ into $G^{in}$ and $G^{out}$, and then recurse on  $G^{in}$ and $G^{out}$. The resulting recursion tree, where each node is associated with $\OO(1)$ cycles of length $O(D)$ each, is the BDD. Li and Parter~\cite{li2019planar} gave an algorithm to compute BDD in $\OO(D)$ rounds, so we can assume that a  BDD is given.

Our algorithm is closely related to the algorithm for reachability by Parter~\cite{Parter20}, who solves this problem by designing a reachability labeling scheme. Here, each vertex is given a label of $\Tilde{O}(1)$ bits such that the reachability of two vertices can be inferred from their labels alone. Recursively, assume that one has already computed a reachability labeling scheme for vertices in $G^{in}$ and $G^{out}$, then one can broadcast the labels of the vertices in the cycle separator to a vertex $u$. Next, $u$ infers a connectivity preserver (a possibly non-planar graph preserving connectivity) of separator vertices, which it could use to compute its own reachability label in $G$. Parter's labeling scheme is designed so that the inference is possible. 

Naturally, for approximate SSSP computation, one could use the approximate distance labeling scheme by Thorup~\cite{Thorup04}, which only has size $\OO(1)$ bits per vertex.  Then one can compute a \EMPH{non-planar distance emulator} of separator vertices from these labels alone. However, when putting two non-planar distance preservers from $G^{in}$ and $G^{out}$ back together (by identifying the shared vertices on the separator), we get a non-planar graph. It is unclear how one can compute an approximate distance labeling scheme from this non-planar graph.  Now our \Cref{thm:os-embedding} comes to the rescue: we can use the labels from  $G^{in}$ and $G^{out}$, along with the ordering along the separator, to compute an approximate \EMPH{planar distance emulator} in polynomial time. Once we have a planar distance emulator, we can apply Thorup's labeling scheme~\cite{Thorup04} to get the approximate distance labels in $G$.   As  mentioned above, in this setting, we can actually get the ordering of the vertices along the separator to make the polynomial construction possible; \Cref{thm:embedding-hardness} implies that without the ordering, the problem is NP-hard. One can view our technique as an efficient construction of a compact approximate distance emulator for planar digraphs, the first of its kind. Existing compact distance emulators are known only for unweighted graphs~\cite{CO20,CKT22,LTZ25,CC25}.

\subsection{Technical ideas: Proof overview of \Cref{thm:main}}\label{subsec:ideas}

Our planar embedding of Okamura-Seymour quasimetrics, or \EMPH{OS quasimetrics} for short, is built directly on the ideas of Chen and Tan~\cite{ChenTan2025}. For completeness, we first review their ideas.

For a planar edge-weighted planar-embedded digraph $G = (V,E)$ and terminal set $T\subseteq V$ on the outerface of $G$, we call $(G,T)$ a \EMPH{directed Okamura-Seymour instance}. Let $\sigma_T$ be a sequence representing the \EMPH{clockwise ordering} of all terminals in $T$ on the outerface of $G$. If $s$ is a subsequence of $\sigma_T$, we write $s\subseq \sigma_T$.  By the triangle inequality, it is well known that the terminal distances in a directed Okamura-Seymour instance satisfy the (directed) \EMPH{Monge property}, namely: 
\begin{equation}\label{eq:monge}
\delta_G(x_1,y_2)  + \delta_G(x_2,y_1) \leq \delta_G(x_1,y_1) + \delta_G(x_2,y_2) \quad \forall ~ \langle x_1,y_2,y_1,x_2\rangle  \subseq \sigma_T  
\end{equation}

Thus, for a quasimetric $(T,\delta_T)$ to be an OS quasimetrics, it has to satisfy the Monge property w.r.t some ordering $\sigma_T$ of the terminals. Chen and Tan~\cite{ChenTan2025} show that this necessary condition is also sufficient.  Additionally, they (i) give a rather simple algorithm to find a valid ordering $\sigma_T$ and (ii) show that there exists an OS instance realizing $(T,\delta_T)$ that has polynomial size.

\begin{theorem}[Theorems 1 and 2 in~\cite{ChenTan2025}]\label{thm:Chen-Tan} A quasimetric $(T,\delta_T)$ is an Okamura-Seymour quasimetric if and only if there exists an ordering $\sigma_T$ such that $\delta_T$ satisfies the Monge property w.r.t. $\sigma_T$.  Furthermore, if $(T,\delta_T)$ is an Okamura-Seymour quasimetric, then:
\begin{enumerate}
    \item the ordering $\sigma_T$ can be found in time $\poly(|T|)$, and
    \item  there exists an Okamura-Seymour instance $(G,T)$ realizing $(T,\delta_T)$ has  $O(|T|^6)$ vertices.
\end{enumerate}
\end{theorem}

\Cref{thm:Chen-Tan} implies a simple algorithm for testing of a given quasi-metric $(T,\delta_T)$ is an OS quasimetric: first compute the ordering $\sigma_T$ and then test the Monge property in \Cref{eq:monge} w.r.t. $\sigma_T$  for every quadruple of terminals forming a subseuence of $\sigma_T$.  

As mentioned above, the difficult step of proving \Cref{thm:Chen-Tan} is to show that the Monge property w.r.t. $\sigma_T$ is sufficient. For this, Chen and Tan construct a planar digraph $G$ directly as follows:

\begin{enumerate}
    \item \textbf{path drawing:~} Draw a curve $\pi_{tt'}$ for each (ordered) pair of terminals $(t,t')\in T\times T$ representing the planar embedding of the shortest path $\pi_G(t,t')$. The set of curves $\{\pi_{tt'}\}_{(t,t')\in T\times T}$ has the property that any two curves cross at a \EMPH{finite} number of points, and each crossing point is \EMPH{proper}, i.e., it is not a touching point. The crossing points and the terminals are the vertices of $G$, and the segments connecting them are the edges of $G$.
    \item \textbf{edge-weight assignment:} Show that there exists a weight assignment to every edge such that  $\delta_G(t,t') = \delta_T(t,t')$ for every pair $(t,t')\in T\times T$.
\end{enumerate}

The first step is arguably more difficult because one has to draw paths in a way that there exists a weight assignment scheme to the edges such that the terminal distances are preserved. 
Next, we sketch the ideas of the path drawing  by Chen  and Tan~\cite{ChenTan2025}. 

\begin{figure}[ht]
    \centering
    \begin{subfigure}[b]{0.45\textwidth}
        \centering
        \includegraphics[page=1, trim=70pt 160pt 70pt 90pt, clip, width=\textwidth]{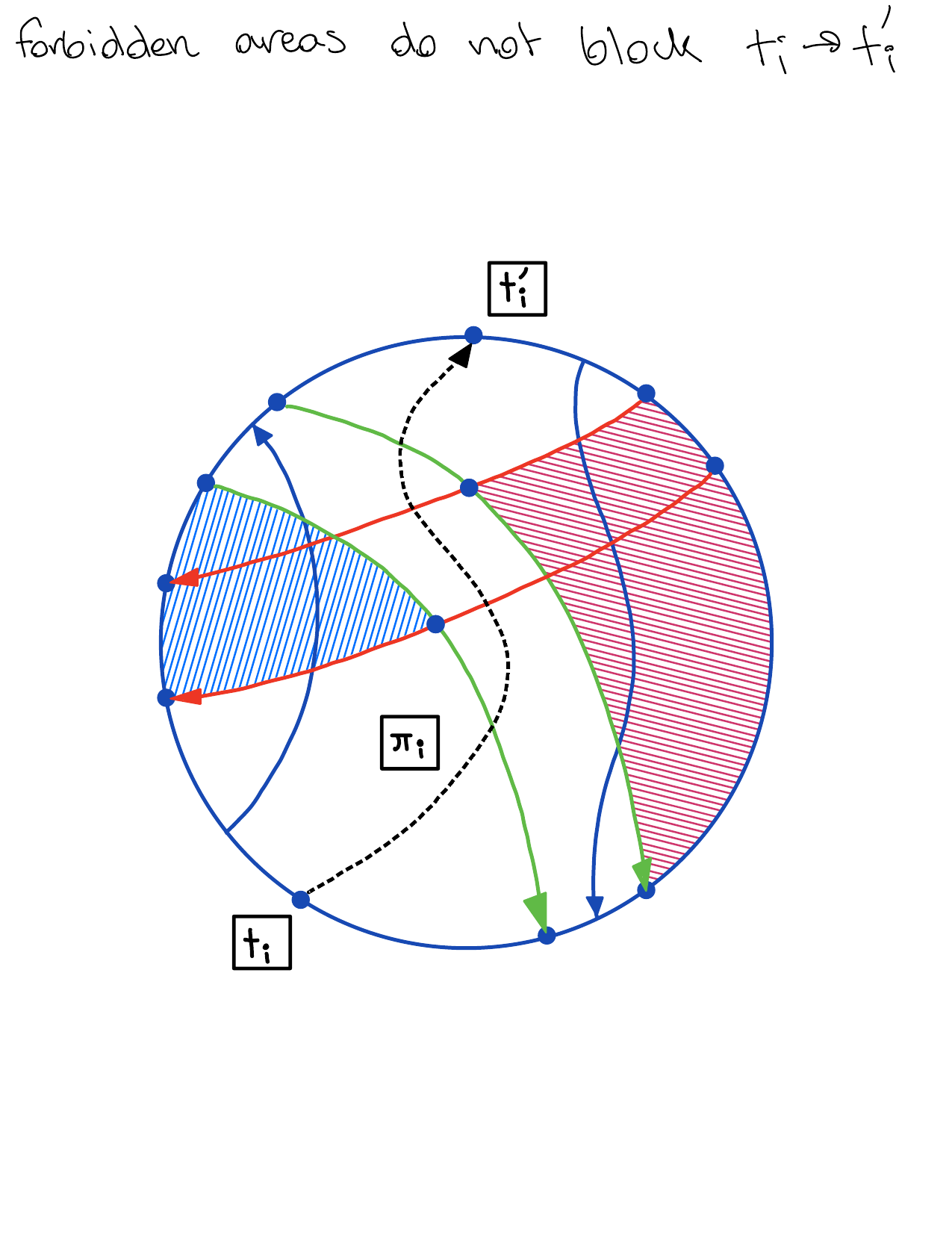}
        \caption{Forbidden Faces}
        \label{fig:ChenTanForbidden}
    \end{subfigure}
    \hfill 
    \begin{subfigure}[b]{0.45\textwidth}
        \centering
        \includegraphics[page=2, trim=70pt 230pt 70pt 90pt, clip, width=\textwidth]{figuresfinal/IntroAndUseCaseFinal.pdf}
        \caption{The Wall}
        \label{fig:ChenTanWall}
    \end{subfigure}
    \hfill
    \caption{An illustration depicting a) the forbidden faces (dashed areas) and a directed path from $t_i$ to $t'_i$ (black dotted arrow) to show that the forbidden faces do not prevent $t_i$ from reaching $t'_i$. In b) we fix the embedding of $\pi_i$ of a) (black solid line), and highlight the complete wall (bold boundary of the green region).}
    \label{fig:ChenTan}
\end{figure}

\paragraph{Chen-Tan path drawing.} Chen and Tan proposed a very natural greedy process: drawing one path at a time, in arbitrary order. That is, the terminal pairs are ordered arbitrarily, say $\{(t_1,t_1'),(t_2,t_2')\ldots\}$, and for each pair $(t_i,t'_i)$, draw a new path $\pi_i$ in the current (planar embedded) graph $G_{i-1}$, which contains paths drawn in previous steps. (The starting point graph $G_0$ has no edges.) 

The tricky part of this greedy process is that, when one tries to draw $\pi_i$ on top of $G_{i-1}$ to realize the distance $\delta_T(t_i,t_i')$, there are certain faces of $G_{i-1}$ that $\pi_{i}$ cannot pass through: if one draws $\pi_i$ going through such a face, there is no way to assign edge weights to the graph $G_i$ to realize the distances $\delta_{T}(t_1,t_1'),\ldots, \delta_T(t_i,t_i')$ (due to the Monge property). Such a face is called a \EMPH{forbidden face}. (One can also regard the outerface as a forbidden face since we will never draw paths going through the face.) They show that the union of all forbidden faces does not block $t_i$ from reaching $t_i'$ (see Figure~\ref{fig:ChenTanForbidden}). However, avoiding forbidden faces may not be sufficient:  there are many ways to draw a path $\pi_i$ from $t_i$ and $t_i'$ avoiding forbidden faces, and not all of them are feasible. Here, a \EMPH{feasible drawing} of $\pi_i$ means that it is possible to assign weights to the edges of $G_i$ to realize all the distances of the first $i$ pairs of $T\times T$, called a \EMPH{feasible weight assignment} to $G_i$. Their key contribution is the following lemma, which shows that a feasible drawing always exists. Thus,  we can draw $\pi_i$ and continue drawing curves until all the terminal distances are realized.

\begin{lemma}[Claim 6~\cite{ChenTan2025}, rephrased]\label{lm:Chen-Tan-keylemma}  There always exists a feasible drawing of $\pi_i$ on top of $G_{i-1}$.
\end{lemma}

However, the proof of \Cref{lm:Chen-Tan-keylemma} is existential and is by contradiction. At a very high level, suppose that for every possible drawing\footnote{There is a simple way to enumerate all drawings of $\pi_i$} of $\pi_i$, there is no way to assign the weights to realize distances if the first $i$ pairs, then they show that before $\pi_i$ is inserted, $G_{i-1}$ already has no feasible weight assignment (for the first $i-1$ pairs), which is a contradiction. In more detail, let us fix a drawing of $\pi_i$. The boundary of the sequence of faces of $G_{i-1}$ that $\pi_{i}$ goes through is called a \EMPH{complete wall}. A prefix of the complete wall that contains a prefix of $\pi_i$ (starting from $t_i$) is called a \EMPH{wall} (see Figure~\ref{fig:ChenTanWall}). 

For a wall $W$, Chen and Tan~\cite{ChenTan2025} defined a so-called \EMPH{$W$-good tuple}: roughly speaking, a $W$-good tuple is a witness for the fact that there is no feasible weight assignment. (The precise definition of a  $W$-good tuple is too complicated to recite here in full.) Every complete wall (containing the entire curve $\pi_i$) contains a $W$-good tuple since $G_i$ has no feasible weight assignment. Their idea is to do backward induction to show that every wall $W$ contains a $W$-good tuple. Specifically, if one grows $W$ by one edge on the left boundary or one edge on the right boundary, one gets two walls $W'$ and $W''$. By induction both $W'$ and $W''$ contain $W'$-good and $W''$-good tuples, respectively. Then they can show that $W$ contains a $W$-good tuple (Claim 9 in~\cite{ChenTan2025}). Working back to the case $W = \varnothing$, the empty wall also contains a $(W = \varnothing)$-good tuple. However,  when $W = \varnothing$, $\pi_i$ is an empty curve, so the graph is $G_{i-1}$. The existence of $\varnothing$-good tuple implies that $G_{i-1}$ has no feasible weight assignments for the first $i-1$ pairs, which is a contradiction. 

\begin{figure}[htbp]
    \centering
    \begin{subfigure}[b]{0.49\textwidth}
        \centering
        \includegraphics[page=6, trim=40pt 160pt 70pt 60pt, clip, width=.45\textwidth]{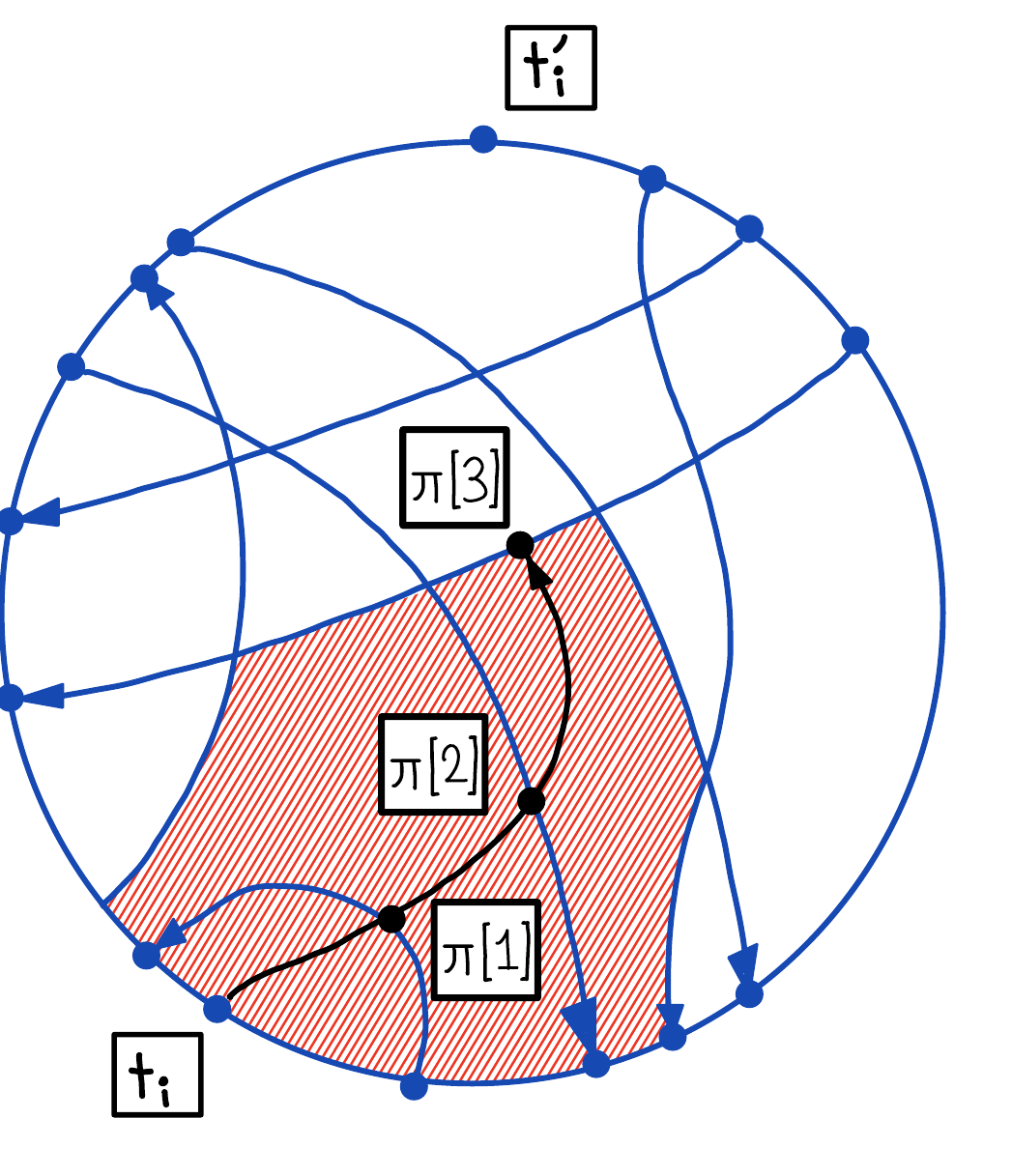}
        \caption{An Example Embedding of $\pi_i$}
        \label{fig:introsplittingsetup}
    \end{subfigure}
    \hfill 
    \begin{subfigure}[b]{0.49\textwidth}
        \centering
        \includegraphics[page=5, trim=40pt 160pt 70pt 60pt, clip, width=.45\textwidth]{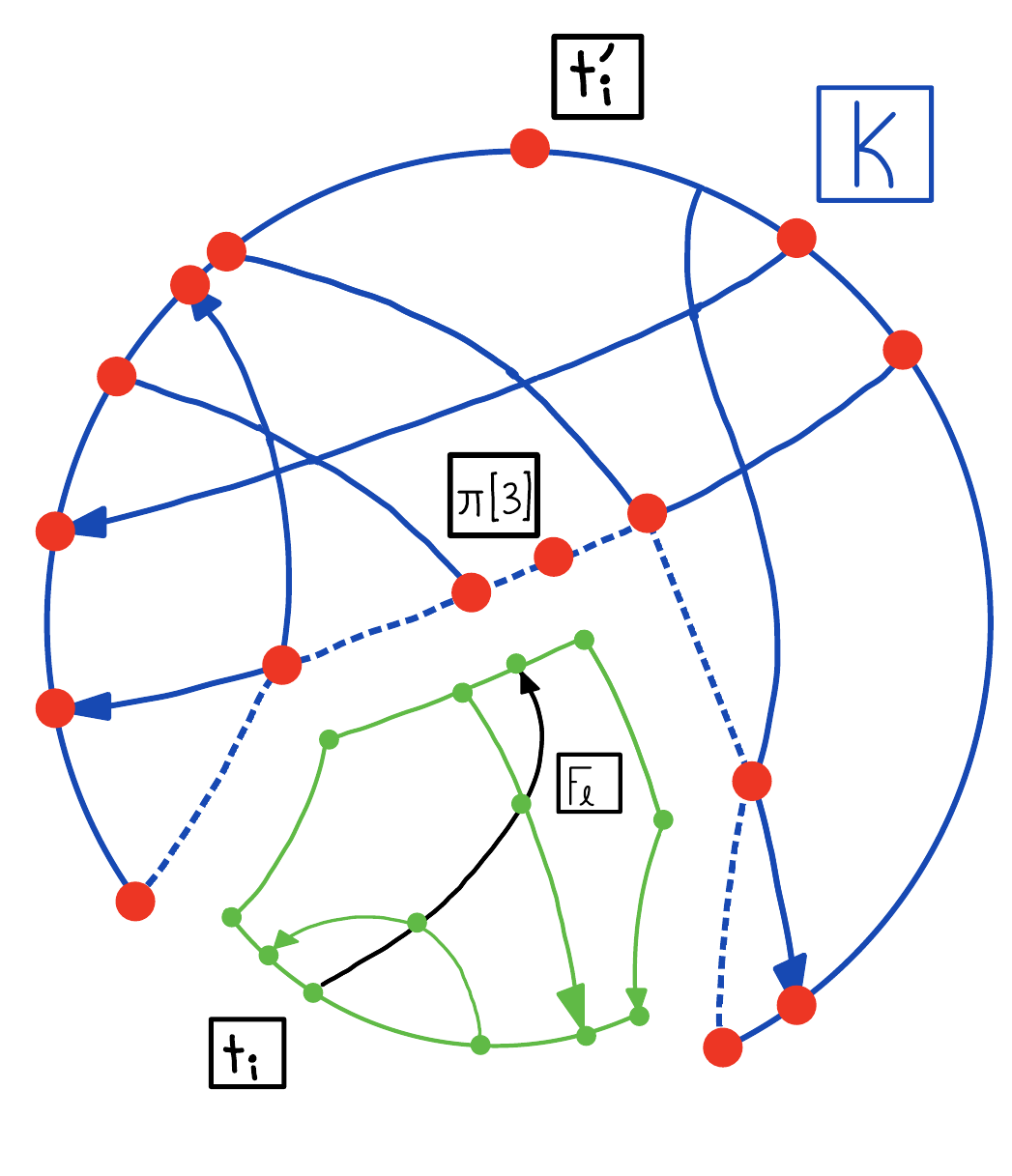}
        \caption{Cutting Yields $(K, T^{+})$}
        \label{fig:introsplittingccl}
    \end{subfigure}
    \hfill
    \caption{Consider an embedding of $\pi_i$ as in a) (black line connecting $t_i$, and $\pi[\ell] = \pi[3]$). The faces $F_0, F_1, ..., F_{\ell}$ are the red shaded region. Then in b), cutting out $F_0\cup\ldots \cup F_{\ell}$ (here two faces delimited by green boundaries) along the boundary we have another directed OS instance $(K,T^{+})$. $K$ is the collection of blue paths, and $T^{+}$ contains all terminals and vertices on the boundary of the faces (red nodes).}
    \label{fig:introsplitting}
\end{figure}

\paragraph{Our ideas.~} As a warm-up, we can turn the existential proof of \Cref{lm:Chen-Tan-keylemma} into an exponential time algorithm as follows: we enumerate all possible drawings of $\pi_i$, and for each drawing of $\pi_i$, solve a linear program (LP) to find the set of weights on the edges. Each drawing of $\pi_i$ is precisely characterized by a sequence of faces of $G_{i-1}$ that $\pi_i$ passes through. A subtle step in this algorithm is to guarantee that $G_{i}$ only has $\poly(|T|)$ faces by sparsifying $G_i$ using \Cref{thm:Chen-Tan} in case its number of faces is too large. Sparsifying $G_i$ essentially changes the greedy embedding process, so we have to ``restore'' (sparsified) $G_i$'s greedy structure (\Cref{cor:static_recognition_iter_i}). Once we make sure that $G_{i}$ has $O(\poly(|T|))$ faces, there are $2^{\poly(|T|)}$ different embeddings of $\pi_i$. Finally, to find the weight assignment, we observe that the LP can be solved in polynomial time using the Ellipsoid method, by providing a simple polynomial time separation oracle using a shortest path algorithm. Clearly, the enumeration step is the bottleneck in this embedding pipeline. 

To improve the running time to polynomial, instead of enumerating all embeddings of $\pi_i$,  we grow $\pi_i$ one step at a time, starting from $t_i$ until reaching $t_{i}'$. As mentioned above, $\pi_i$ can be represented by a sequence of faces that it goes through.  At any intermediate step $\ell$, we have a sequence of $\ell$ faces $(F_0,F_1,\ldots, F_\ell)$ where $F_0$ contains $t_i$ representing a prefix of length $\ell$ of $\pi_i$. (This sequence of faces constitutes a wall in Chen-Tan's terminology.) The invariant that our algorithm maintains is:

\begin{quote}
    \textbf{Valid Prefix Invariant:} For any $\ell\geq 0$,  there exists a feasible drawing of $\pi_i$ on $G_{i-1}$ such that $\pi_i$ goes through all the faces $F_0,\ldots, F_{\ell}$. 
\end{quote}

The next face $F_{\ell+ 1}$ that (a feasible) $\pi_i$ goes through is adjacent to $F_{\ell}$, i.e., $F_{\ell}$ and $F_{\ell+ 1}$  share an edge. For each edge $e$ on the boundary of $F_{\ell}$, except the edge between 
$F_{\ell}$ and $F_{\ell-1}$, let $F\not= F_{\ell}$ be the other face whose boundary contains $e$. Our \EMPH{key technical contribution  is a polynomial-time algorithm to test if $F$ is a candidate for $F_{\ell+1}$, i.e., deciding if there exists a feasible $\pi_i$ going through $(F_0,\ldots, F_{\ell},F)$}. By the valid prefix invariant, there will be a face $F$ adjacent to $F_{\ell}$ such that the test will pass. Therefore, we can set $F_{\ell+1} = F$ and the prefix invariant holds for $F_{\ell+1}$. The algorithm will stop whenever we reach $t_i'$, i.e., $F_{\ell+1}$ contains $t_i'$ on the boundary. Since the number of faces of $G_{i-1}$ is $O(\poly(|T|))$, the number of growing steps is polynomial, and the whole algorithm is in polynomial time.

Our testing algorithm is based on linear programming: we will write a linear program with an exponential number of constraints, and then give a polynomial-time separation oracle to solve it using the Ellipsoid method. Our LP is formulated by exploiting the Monge property. 
Let $\pi_i[0] = t_i$, and $\pi_i[j]$ be the crossing point of $\pi_i$ and the edge  shared between $F_{j-1}$ and $F_{j}$ for $1\leq j \leq \ell$. We draw $\pi_i$ passing through $F$ by picking a point  on $e$, denoted by $\pi[\ell]$, and drawing an edge from $\pi_i[\ell-1]$ to $\pi[\ell]$, which will be a new edge of $\pi_i$. Our problem now is reduced to deciding if there exists a feasible weight assignment to the first $\ell$ edges of $\pi_i$, namely $\{(\pi_i[j-1],\pi_i[j])\}_{j=1}^{\ell}$ such that there exists a feasible way to extend $\pi_i$ from $\pi_i[\ell]$ all the ways to $t'_{i}$.  

To illustrate the idea of our linear program for the above problem, let us assume (incorrectly) that we can assign weight to edges in $\pi_i$ without changing the weight of any edge, as well as the terminal distances, in $G_{i-1}$. (The edges intersecting $\pi_i$ at $\{\pi[j]\}_{1\leq j \leq \ell}$ will be split into two edges, but the total weight of the two splitting edges is the same as the original edge.) Let $x(\pi_i[j-1],\pi_i[j])$ be a variable denoting the weight assign to edge  $(\pi_i[j-1],\pi_i[j])$ for $j\in [1,\ell]$. We introduce another \EMPH{variable $d(\pi[\ell_i],t_i')$} to denote the total length of the suffix of $\pi_i$ from  $\pi_i[\ell]$ to $t'_i$. The first constraint of our LP is straightforward:

\begin{equation}\label{eq:length-constraint}
     \sum_{j=1}^{\ell}  x(\pi_i[j-1],\pi_i[j])  + d(\pi[\ell_i],t_i') = \delta_{T}(t_i,t_{i}')~,
\end{equation}
which simply means that the sum of the weights of the edges along $\pi_i$ must realize the distance $\delta_{T}(t_i,t_{i}')$. To formulate the second type of inequality, which encodes the fact that the weight of $d(\pi[\ell_i],t_i')$ can be realized in the future as well, we observe the following: If we cut out $F_0\cup\ldots \cup F_{\ell}$ along the boundary (see Figure~\ref{fig:introsplitting}), then we have another directed OS instance, denoted by $(K,T^{+})$ with the set $T^{+}$ of terminals containing all the terminals and vertices on the boundary of the faces. Thus, by \Cref{thm:Chen-Tan}, we can realize $d(\pi[\ell_i],t_i')$ if this distance, along with other distances between terminals in $T^{+}$, satisfies the Monge property in $K$ w.r.t. a clockwise ordering \EMPH{$\sigma_K$} on the boundary of $K$. More precisely:
\begin{equation}\label{eq:monge-constraint}
    \delta_{K}(u,t_i') + \delta_{K}(\pi_i[\ell],v) \leq \delta_{K}(u,v) + d(\pi[\ell_i],t_i') \quad \forall~ \langle u,\pi_i[\ell],v, t'_{i}\rangle \subseq \sigma_{K}
\end{equation}
As we assume that the weight assignment of edges in $\pi_i$ does not change the weight of existing edges, only $d(\pi[\ell_i],t_i')$ is a variable in \Cref{eq:monge-constraint}; all other terms $\delta_{K}(u,t_i'), \delta_{K}(\pi_i[\ell],v), \delta_{K}(u,v)$ are constants. Our linear program is to find a feasible non-negative assignment to variables $\{ x(\pi_i[j-1],\pi_i[j])\}_{j=1}^{\ell}$ and $d(\pi[\ell_i],t_i')$ satisfying the constraints in \Cref{eq:length-constraint} and \Cref{eq:monge-constraint}. There are only a polynomial number of constraints, so this problem is solvable in polynomial time. Of course, the assumption that the weight assignment of edges in $\pi_i$ does not change the weight of existing edges is not warranted.  

In general, adding the edges $\{(\pi_i[j-1],\pi_i[j])\}_{j=1}^{\ell}$ and assigning weights could change the structure of shortest paths in $G_{i-1}$ completely, as these edges could create unintentional shortcuts. Even vertices not reachable from each other in $G_{i-1}$ (infinite distances) now can be reachable via new edges (finite distances). In our linear program, we have to (i) reassign weights to all edges of $G_{i-1}$, and (ii) make sure that the new weights and the new edges (from $\pi_i$) still realize existing distances between terminal pairs $\{(t_j,t_j')\}_{j=1}^{i-1}$, and (iii) $d(\pi[\ell_i],t_i')$ is realizable in the cut out graph $K$. The Monge constraints to realize (iii) in \Cref{eq:monge-constraint} remain crucial, but now the terms $\delta_{K}(u,t_i'), \delta_{K}(\pi_i[\ell],v), \delta_{K}(u,v)$,  instead of being constants, are variables and constrained by linear combinations of other edge-weight variables encoding the distances. The number of constraints is exponential to prevent shortcuts in (ii). (See \hyperlink{validityLP}{\textsc{PrefixValidityLP}} for details.) Luckily, we are still able to design a polynomial time separation oracle to solve this exponential LP. (Very) roughly speaking, a key observation for deriving the separation oracle is that all the curves $\pi_{1},\ldots, \pi_{i}$ must remain shortest paths between their terminal endpoints, even after reassigning weights, and therefore, we can reduce the separation oracle to shortest path computation.

\section{An Exponential Time Algorithm}\label{sec:Exp-time}

The goal of this section is to introduce relevant terminology and the embedding framework of Chen and Tan~\cite{ChenTan2025}. Once the framework is set up, we will give a simple exponential-time algorithm for embedding an OS quasimetric on the plane. A technical subtlety of this section is handling a graph sparsification step, which will also be used in our polynomial-time algorithm. We believe the introductory materials in this section will make our polynomial-time algorithm easier to understand, as it is somewhat technically heavy.

\begin{theorem}\label{thm:exp} Let $(T,\delta_T)$ be an Okamura-Seymour quasimetric where the distance between any two terminals is an integer in $\{0,1,\ldots, W\}$. There exists an algorithm with running time $2^{\poly(|T|)}\poly(\log W)$  for constructing a planar embedding of $(T,\delta_T)$. 
\end{theorem}

 For a (di)graph $G$, we use $V(G)$ and $E(G)$ to denote its vertex set and edge set, respectively. 
 
 Imagine a greedy procedure that embeds an Okamura-Seymour quasimetric $(T,\delta_T)$ by drawing the shortest path between terminal pairs one at a time. The next path is drawn in a way that it cannot share edges with previous paths  (but it can share vertices, i.e., crossing points). The resulting  Okamura-Seymour instance at the end of this process, $(G,T)$, is called a nest, as formally defined below.

\begin{definition}[Nest~\cite{ChenTan2025}]  A directed Okamura-Seymour instance $(G, T)$ realizing quasi-metric $(T,\delta_T)$ is a \EMPH{nest}  if and only if $G$ is composed of a set $\mathcal{P}$ of directed paths subject to the following conditions:

    \begin{enumerate}
        \item every path $P \in \mathcal{P}$ begins and ends in $T$,
        \item $P_1, P_2 \in \mathcal{P}$ are edge-disjoint,
        \item $\E = \cup_{P \in \mathcal{P}}E(P)$.
    \end{enumerate}
\end{definition}

The nest is central in the greedy embedding framework by Chen and Tan~\cite{ChenTan2025} given below.

\begin{tcolorbox}[
 colback=white,          
  colbacktitle=gray,       
  coltitle=white,          
  colframe=gray,           
 title={\textsc{Chen-Tan Embedding Framework}}]
 \textbf{Input:} an Okamura-Seymour quasimetric $(T,\delta_T)$\\
 \textbf{Output:} A directed Okamura-Seymour instance $(G, T)$ realizing  $(T,\delta_T)$.\\
 
 \begin{enumerate}
     \item $\sigma\leftarrow$ a valid ordering of  $T$ from \Cref{thm:Chen-Tan}.
     \item $\{(t_1,t_1'),(t_2,t_2')\ldots,\}$ an (arbitrary) ordering of $|T|^2$ terminal (ordered) pairs in $T\times T$. 
     \item $G_0 = (T,\varnothing)$ be a planar embedded graph with no edge, where $T$ is ordered on the boundary of a disk according to $\sigma$.
     \item  For each pair of terminals $(t_i,t_{i'})$ in the ordering, where $i \in [1,|T^2|]$
     \begin{enumerate}
         \item $\Pi_{\leq i-1} \leftarrow \{\pi_j: j\leq i-1\}$.
         \item $(G_i, \pi_i)\leftarrow \textsc{EmbedNextPath}(G_{i-1}, t_i,t_i', \Pi_{\leq i-1})$. This procedure inserts  a curve $\pi_i$ from $t_i$ to $t'_{i}$ on top of the current embedding $G_{i-1}$ such that two following conditions are satisfied:
         \begin{itemize}
             \item \textbf{(Embedding condition)}  $\pi_{i}$ is a simple curve connecting $t_i$ to $t'_i$ such that for every previously inserted curve $\pi_{j}$ for $j\leq i-1$, the curves $\pi_{i}$ and $\pi_{j}$ intersect only via finitely \emph{proper} crossings (i.e., they do not merely touch tangentially and then separate). Moreover, we view each crossing as a vertex and each curve segment between a pair of consecutive crossings as a directed edge. 
             \item \textbf{(Weighting condition)} Let $G_i$ be the graph obtained by drawing $\pi_i$ on top of $G_{i-1}$. The exists an edge-weight function   $w_i:E(G_i)\to \mathbb{R}_{\ge 0}$ such that for every
  $1\le j\le i$, the curve $\pi_{j}$ is a directed shortest path from $t_j$ to $t'_j$ in $(G_{i},w_i)$, and its length equals to $\delta_T(t_j,t'_j)$.
         \end{itemize}
     \end{enumerate}
 \end{enumerate}
\end{tcolorbox}

In this framework, the procedure $\textsc{EmbedNextPath}(G_{i-1}, t_i,t_i', \Pi_{i-1})$ will largely determine the running time of the resulting algorithm. We refer to this step as \EMPH{iteration $i$}. We note that the embedding condition implies that $G_i$ is a nest after each iteration, and that there is no vertex of degree 2 except on the boundary. Chen and Tan~\cite{ChenTan2025} show that the framework will give a valid planar embedding of $(T,\delta_T)$, as formally stated in the following lemma.

\begin{lemma}[Claim 6~\cite{ChenTan2025}]\label{lm:iterative-correctness} If the embedding and weighting conditions of $\textsc{EmbedNextPath}(G_{i-1}, t_i,t_i', \Pi_{\leq i-1})$ are satisfied for every $i$, the final graph $G_{|T|^2}$ (when $i = |T|^2$) is a planar embedding of $(T,\delta_T)$. \\
Conversely, if  $(T,\delta_T)$ is an Okamura-Seymour quasimetric, then for every $i$, there exists $G_i$ satisfying both  the embedding and weighting conditions of $\textsc{EmbedNextPath}(G_{i-1}, t_i,t_i', \Pi_{\leq i-1})$. (That is, the Chen-Tan embedding framework is guaranteed to find a valid planar embedding.) 
\end{lemma}

For each iteration $i$, we define a  \EMPH{partial quasi-metric}, denoted by $\delta^{\leq i}_T$, as follows: for any ordered terminal pair $(t_j,t-J')\in T\times T$, if $j\leq i$, we set $\delta^{\leq i}_T(t_j,t'_j)= \delta_T(t_j,t'_j)$; otherwise, we set $\delta^{\leq i}_T(t_j,t'_j)=\ast$ to indicate that the value is undefined (unknown). Note that for $i$, $\delta^{\leq i}_T(t_i,t_i')=\delta_T(t_i,t_i')$.

We say that the instance $(G_i,T)$ output by  $\textsc{EmbedNextPath}(G_{i-1}, t_i,t_i')$ is a nest instance realizing the partial quasimetric $(T,\delta^{\leq i}_{T})$.

\begin{figure}[ht]
    \centering
    \begin{subfigure}[b]{0.45\textwidth}
        \centering
        \includegraphics[page=1, trim=70pt 160pt 70pt 170pt, clip, width=\textwidth]{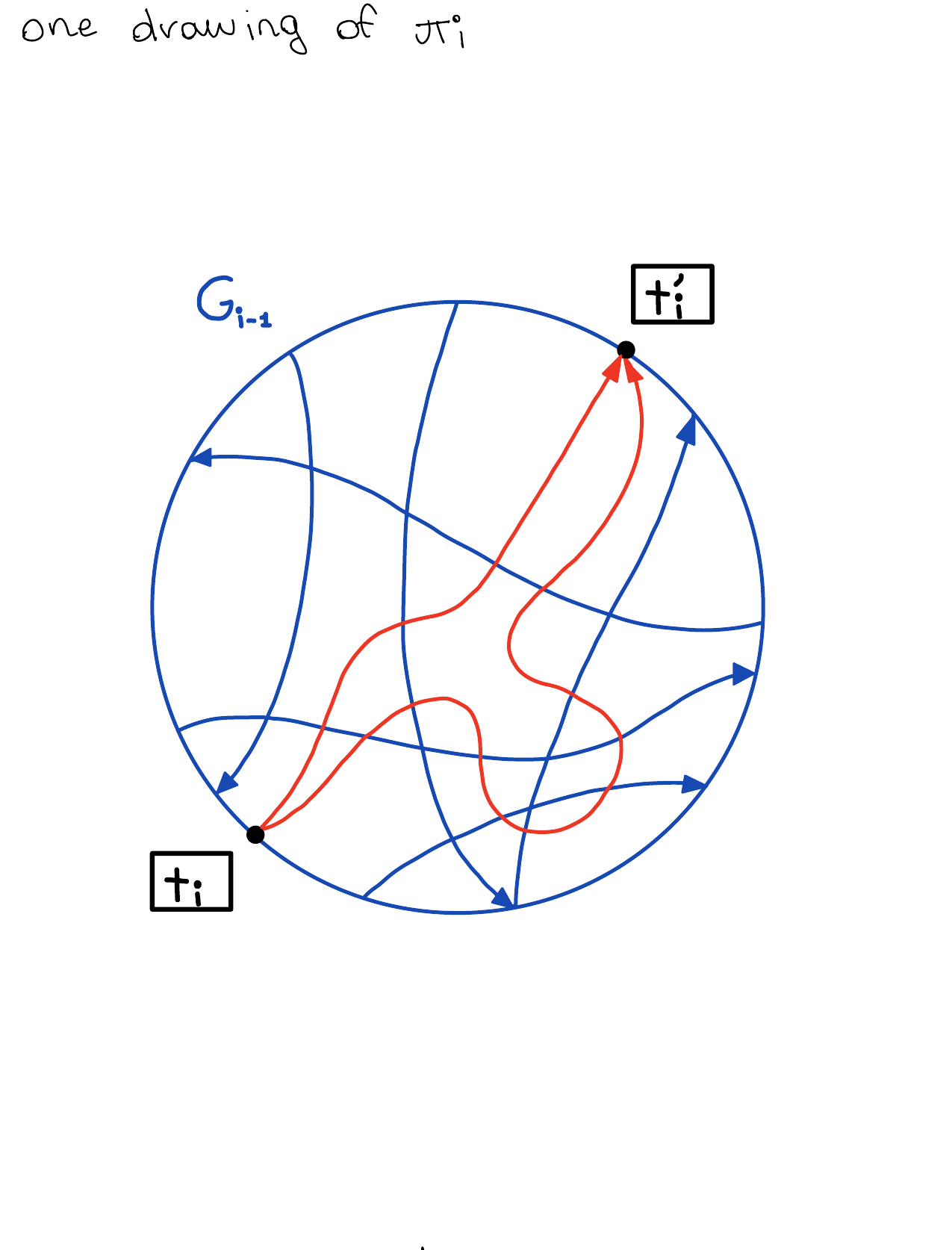}
        \caption{Two Possible Ways of Drawing $\pi_i$}
        \label{fig:simple}
    \end{subfigure}
    \hfill 
    \begin{subfigure}[b]{0.45\textwidth}
        \centering
        \includegraphics[page=2, trim=70pt 160pt 70pt 170pt, clip, width=\textwidth]{figuresfinal/SimpleFinal.pdf}
        \caption{The Corresponding Simple Dual Path}
        \label{fig:simpledual}
    \end{subfigure}
    \hfill
    \caption{a) An illustration of two different embeddings of $\pi_i$ in $G_{i-1}$ (red dipaths) and b) the corresponding simple dual path in $G_{i-1}^*$ of the proof of Claim~\ref{lm:iterative-correctness} for both embeddings (here drawn on top of $G_{i-1}$). In b), we omit the dual vertex corresponding to the infinite face and adjacent edges in $G_{i-1}^*$ for visual clarity.}
    \label{fig:simpledualcomparaison}
\end{figure}

\paragraph{Drawing and enumerating $\pi_i$.} The drawing of $\pi_i$ on top of $G_{i-1}$ is precisely characterized by a sequence of faces $\langle F_0,F_1,\ldots, F_{h}\rangle$ where (i) $F_0$ contains $t_i$ and $F_{h}$ contains $t_{i}'$ on their respective boundaries, (ii) $F_{\ell}$ shares exactly one edge with $F_{\ell+1}$ for all $\ell\in [0,h-1]$. $\pi_i$ is drawn as a curve passing through these faces, and intersecting only with edges shared between two consecutive faces~\cite{ChenTan2025}.  Condition (ii) means that this sequence of faces forms a path in the dual graph $G^*_{i-1}$. In the following claim, we show that enumerating all possible drawings of $\pi_i$ is equivalent to enumerating \EMPH{simple paths} in the (undirected) dual graphs $G^*_{i-1}$ where the endpoints correspond to faces containing $t_i$ and $t_i'$ on their boundaries; see \Cref{fig:simpledualcomparaison}. We refer to the sequence of faces  $\langle  F_0,F_1,\ldots, F_{h}\rangle$  as a \EMPH{candidate trajectory} of $\pi_i$. 

\begin{claim} \label{clm:simple_path}
    Enumerating all embeddings of $\pi_{i}$ is equivalent to enumerating all simple paths in the undirected dual graph $G^*_{i-1}$ where the endpoints correspond to faces containing $t_i$ and $t_i'$ on their boundaries.
\end{claim}
\begin{proof} Suppose that there is a valid embedding of $\pi$ such that the corresponding dual path, denoted by $\pi^*_i$, is non-simple. Let $G_i$ be obtained from $G_{i-1}$ by embedding $\pi_i$ and assigning a weight to every edge so that the distances between terminal pairs up to iteration $i$ are preserved. Since $\pi^*_i$ is non-simple, there exists a face $F'$  of $G_{i-1}$ such that $\pi^*_i$ visits at least twice. Let $\bdry F'$ be the boundary of $F'$. Let $u$ and $v$ be $\pi_i\cap \bdry F'$ be the first vertex and the last vertex of $\pi_i$, respectively, on the boundary of $F'$. Note that the subpath $\pi_i[u,v]$ of $\pi_i$ contains at least one vertex other than $u$ and $v$ by the definition of $F'$. Now we simplify $\pi_i$ by removing all internal vertices of $\pi_i[u,v]$  and adding a direct edge from $u$ to $v$ (inside $F'$) with weight to be the total weight of  $\pi_i[u,v]$. Clearly, the total weight of $\pi_i$ is preserved. Furthermore, simplifying $\pi_i$ does not create any new crossings and therefore cannot create any new shortcuts. By the weighting condition w.r.t $G_i$, simplifying $\pi_i$ does not change the weight of the shortest paths in $\Pi_{\leq i-1}$, which might have intersected $\pi_i$ before the simplification. Thus, $G_i$ after the simplification still satisfies both the embedding and the weighting conditions. By repeatedly simplifying $\pi$ until its dual path $\pi^*_i$ is simple, we get the claim.
\end{proof}

\paragraph{Checking feasibility of $\pi_i$.~} Given a candidate trajectory $\mathcal{F}$ of $\pi_i$ on top of $G_{i-1}$, we have to test if the weighting condition in the Chen-Tan framework can be satisfied after we draw $\pi_i$ following $\mathcal{F}$. If the test passes, then we accept the drawing of $\pi_i$, and set $G_i$ to be the resulting graph $G_{i-1}\cup \{\pi_i\}$. In this case, we call $\mathcal{F}$ a \EMPH{valid trajectory}. Otherwise (the test fails), we reject the current candidate trajectory of $\pi_i$, and try the next one. Our test is based on a simple linear program. Recall that $\Pi_{\leq i} = \cup_{1\leq j \leq i} \pi_j$.

\begin{lemma}\label{lm:check-feasibility} Let $\hat{G}_i$ be the graph obtained by drawing $\pi_i$ on top of $G_{i-1}$ following a given candidate trajectory $\{F_0,F_1,\ldots, F_{h}\}$. There is an algorithm , denoted by $\textsc{AssignWeight}(\hat{G}_i, \delta^{\leq i}_{T},\Pi_{\leq i})$, that either:
\begin{enumerate}
    \item returns a weight function $w_i$ that assigns a weight $w_i(e)$ to each edge $\hat{G}_i$ such that $(\hat{G}_i,T)$ realizes the partial metric $(T,\delta^{\leq i}_{T})$.
    \item or returns $\bot$ if such a weight function $w_i$ does not exist.
\end{enumerate}
The running time of the algorithm is $\poly(|V(\hat{G_i})|, \log(W))$.
\end{lemma}
\begin{proof} For every terminal pair $(t_j,t_j')$, we denote by $\mathcal{P}(t_j,t_j')$ the set of all simple directed paths from $t_j$ to $t'_j$ in $\hat{G}_i$. We check the existence of $w_i$ by deciding if the following LP is feasible:

\begin{align*}
\sum_{e\in E(\pi_{j})} x(e) &\le \delta^{\leq i}_{T}(t_j,t'_j) 
&& \forall\, 1\le j\le i, \\
\sum_{e\in E(P)} x(e) &\ge \delta^{\leq i}_{T}(t_j,t_j')
&& \forall\, 1\le j\le i,\ \forall\ P\in \mathcal{P}(t_j,t_j')\\
x(e) &\ge 0
&& \forall\, e\in E(\hat{G}_i).
\end{align*}

Here the variables are $\{x(e)\}_{e\in E(\hat{G}_i)}$, which assign a weight $x(e)$ to each edge $e$. The first and the second constraints together guarantee that $\pi_j$ must be a shortest path between its terminal endpoints, and its weight is exactly   $\delta^{\leq i}_{T}(t_j,t_j')$. 

Observe that the LP has an exponential number of constraints, since  $\mathcal{P}(t_j,t_j')$ could have an exponential number of paths. However, the LP admits a very simple polynomial-time separation oracle: given an assignment to the variables, denoted by $\{\Bar{x}(e)\}_{e\in E(\hat{G}_i)}$, we can find a violated constraint as follows. First, we compute the shortest distance between every two terminals $(t_j,t_j')$ in $\hat{G}_i$ w.r.t. the weight function $\Bar{x}$ on the edges, which can be done in polynomial time using Dijkstra's algorithm. We denote the distance by  $\delta_{\Bar{x}}(t_j,t_j')$.

\begin{itemize}
    \item  For the first constraint, we simply compute the total weight of $\pi_j$, namely $\sum_{e\in E(\pi_{j})} \Bar{x}(e)$, and compare with  $\delta^{\leq i}_{T}(t_j,t'_j)$. If this weight is strictly larger than $\delta^{\leq i}_{T}(t_j,t'_j)$, we have found a violated constraint.
   \item  For the second constraint,  if $\delta_{\Bar{x}}(t_j,t_j') \ge  \delta^{\leq i}_{T}(t_j,t_j')$, then every directed path $P$ from $a_j$ to $b_j$ satisfies the second family of constraints.  Otherwise, the second constraint corresponding to a shortest path from  $t_j$ to $t_j'$ is a violated constraint. 
\end{itemize}

Since the feasibility LP admits a polynomial-time separation oracle, we can decide its feasibility in polynomial time using the ellipsoid method. 
\end{proof}

\paragraph{The algorithm.~} As established above, we can find a feasible drawing of $\pi_i$ by enumerating all possible candidate trajectories and testing each candidate using \Cref{lm:check-feasibility}. However, note that every time we draw $\pi_i$, the faces in the candidate trajectory of $\pi_i$ are split into two faces. Thus, in the worst case, the number of faces increases exponentially in the number of iterations (which is $|T|^2$). 

To keep the running time $2^{\poly(|T|)}$ as in \Cref{thm:exp}, we need to guarantee that $|V(G_{i})| = \poly(|T|)$. For this, we use the following nest size reduction subroutine of Chen and Tan~\cite{ChenTan2025} to reduce the size of $G_i$ after every step. 

\begin{lemma}[Nest Size Reduction, Section 5~\cite{ChenTan2025}]\label{lem:nest_reduction}
Let $(G_i,T)$ be the nest obtained from $\textsc{EmbedNextPath}(G_{i-1}, t_i,t_i')$. There is an algorithm, denoted by $\textsc{ReduceNest}(G_{i},T)$, that runs in time $\poly(|V(G_i)|)$ and transform $(G_i,T)$ into an instance $(\Tilde{G}_i,T)$, with a different weight function  on the edges of $\Tilde{G}_i$, such that
\begin{enumerate}
    \item $(\Tilde{G}_i,T)$ is also a nest instance realizing the partial quasi-metric $\delta^{\leq i}_T$.
    \item vertices of $\Tilde{G}_{i}$ are created only at proper intersections of paths in $\Tilde{\Pi}_{\leq i}$ (viewed as curves), and no two paths have the same two vertices along the same direction of both paths. 
    \item $|V(\Tilde{G}_i)| = O(|T|^{6})$.
\end{enumerate}
\end{lemma}

The nest reduction algorithm by  Chen and Tan~\cite{ChenTan2025} considers every pair of shortest paths between terminals and reduces the number of intersections between each pair. Here, we can reduce the size of $G_i$ slightly faster, by exploiting the fact that $G_{i-1}$ is already a reduced nest and therefore, we only need to consider pairs of paths involving $\pi_i$. This does not improve the running time of the overall algorithm, though. 

\begin{algorithm}
\caption{$\textsc{EmbedNextPath}(G_{i-1}, t_i,t_i', \Pi_{\leq i-1})$:}
\begin{algorithmic}[1]
\algnotext{EndIf}
\algnotext{EndFor}
\For{each candidate trajectory $\mathcal{F}$ of $\pi_i$ on $G_{i-1}$}
    \State   draw $\pi_{i}$ according to $\mathcal{F}$ into $G_{i-1}$ to obtain $\hat{G}_i$.
    \State $\Pi_{\leq i}\leftarrow \Pi_{\leq i-1}\cup \{\pi_i\}$
    \State $w_i \leftarrow \textsc{AssignWeight}(\hat{G}_i, \delta^{\leq i}_{T},\Pi_{\leq i})$ \Comment{From \Cref{lm:check-feasibility}}
    \If{$w_i\not= \bot$} \Comment{There is a feasible weight assignment}
       \State $G_{i}\leftarrow \textsc{ReduceNest}(\hat{G}_i, T)$ \Comment{From Lemma~\ref{lem:nest_reduction}}
       \State \textbf{return} $(G_i,\pi_i)$
    \EndIf
\EndFor
\Statex \Comment{An exponential time algorithm for embedding the next path $\pi_i$ from $t_i$ to $t_{i}'$ on top of the graph $G_{i-1}$ from the previous iterations.}
\end{algorithmic}
\label{alg:next-path-exp}
\end{algorithm}

The algorithm for the next path embedding is given in \Cref{alg:next-path-exp}.  For running time, observe that by \Cref{lem:nest_reduction}, $|V(G_i)| = \poly(|T|)$ for every $i$. Thus, the number of candidate trajectorys of $\pi_i$ is $2^{O(|V(G_i)|\log |V(G_i)|)} = 2^{\poly(|T|)}$. By \Cref{lm:check-feasibility}, the running time of next path embedding is $2^{\poly(|T|)}\poly(\log W)$, which is also the asymptotic running time of the whole planar embedding algorithm.

\paragraph{Correctness.~} By \Cref{lm:check-feasibility}, $\hat{G}_i$ satisfies both the embedding condition and the weighting condition. Thus, if we use $\hat{G}_i$ as the graph for the next step, by \Cref{lm:iterative-correctness}, we will get a planar embedding of $(T,\delta_T)$. However, a subtlety here is that $G_i$ is obtained by applying a nest reduction to $\hat{G}_i$, which potentially modifies $\hat{G}_i$. To address this issue, we formalize the notion of conforming, which roughly means that $G_i$ (after the reduction) has all the properties of $\hat{G}_i$ and therefore, the correctness is still followed from \Cref{lm:iterative-correctness}.

For clarity, let $\Tilde{G}_i = \textsc{ReduceNest}(\hat{G}_i, T)$ be the graph obtained by applying nest reduction to $\hat{G}_i$.  We say that $\Tilde{G}_i$ \EMPH{conforms} with the first $i$ iterations of the Chen-Tan embedding framework if we can form a sequence of nests $(\Tilde{G}_0,T), (\Tilde{G}_1,T),\ldots, (\Tilde{G}_i,T)$ such that:
\begin{itemize}
    \item $\Tilde{G}_0$ has no edge.
    \item  For every $1\leq j \leq i$, $\Tilde{G}_j$ is an output of $\textsc{EmbedNextPath}(\Tilde{G}_{j-1}, t_j,t_j')$. 
\end{itemize}

Note that the next path embedding could have more than one valid output; the Chen-Tan embedding framework implies that any output satisfying the embedding and weighting conditions is good enough. Clearly, if $\Tilde{G}_i$ conforms with the first $i$ iterations of the Chen-Tan embedding framework, then we can set $G_i \leftarrow \Tilde{G}_i$ for the next iteration. Thus, the correctness of our algorithm follows from \Cref{lm:iterative-correctness} and the following lemma.

\begin{lemma}\label{cor:static_recognition_iter_i} $\Tilde{G}_i$ conforms with the first $i$ iterations of the Chen-Tan embedding framework.
\end{lemma}

\begin{proof}

We construct a sequence of nests $(\Tilde{G}_0,T), (\Tilde{G}_1,T),\ldots,(\Tilde{G}_{i-1},T), (\Tilde{G}_i,T)$ to serve as a witness that $\Tilde{G}_i$ conforms as follows. 

To start, let $\Tilde{G}_0$ be an empty nest. In iteration $j$, we insert the path $\pi_j$ into $\Tilde{G}_{j-1}$ by copying its embedding from $\Tilde{G}_i$. This embedding may contain intersection vertices in $\Tilde{G}_i$ that are between $\pi_j$ and paths of $\Tilde{G}_i$ that have not been inserted in $\Tilde{G}_{j-1}$. In order to construct $\Tilde{G}_{j}$ which satisfies the embedding condition we remove any such intersection vertex $v$, and merge it's two adjacent edges into one. The new edge is assigned weight in $w_j$ equal to the sum of weights of the merged edges in $w_{i}$. 
For intersection vertices that have already been removed once in a previous iteration, we find the edge that was created by merging when the vertex was removed, and we reverse the merge.
The above operations do not change the length of $\pi_j$ or any $\pi_k$ with $1 \leq k < j$, nor does it introduce a shortcut for any of these paths. As such the weighting condition will be satisfied after each iteration and thus we have shown each $\Tilde{G}_{j}$ constructed this way is an output of $\textsc{EmbedNextPath}(\Tilde{G}_{j-1}, t_j,t_j')$. 
\end{proof}

\section{A Polynomial Time Algorithm}\label{sec:polytime}

In this section, we give a polynomial-time algorithm for embedding an Okamura-Seymour quasimetric. We follow the Chen-Tan embedding framework set up in \Cref{sec:Exp-time}. Observe that the exponential running time complexity comes from enumerating all candidate trajectories of
the curve $\pi_{i}$. Our goal here is to avoid this exhaustive search: instead of enumerating all trajectories, we design a polynomial algorithm for finding a feasible trajectory of $\pi_{i}$ (satisfying the two conditions in the framework). Therefore, the overall planar embedding algorithm is polynomial.

\paragraph{Simplifying notation.} Since we only focus on embedding $\pi_i$, we simplify notation by removing the subscripts as much as we can. Specifically:
\begin{itemize}
    \item We use $a$ and $b$ to denote $t_i$ and $t_{i}'$, respectively, and use $\pi$ to denote $\pi_i$.
    \item  We use $G$ to denote $G_{i-1}$. We also slightly abuse notation to use $\delta_T$ in place of $\delta^{\leq i}(T)$.  $(T,\delta_T)$ is a partial Okamura-Seymour quasimetric: some distance values (corresponding to pairs in iterations $i+1$ or after) are $\ast$ (unknown), while other distance values (corresponding to pairs in iterations $i$ or before)  are integers from $\{0,1\ldots, W\}$.  We call a terminal pair with a known distance value a \EMPH{known terminal pair}. Pair $(a,b)$ is a known terminal pair.
    \item We use $\Pi$ to denote $\Pi_{\leq i-1}$, which is the set of edge-disjoint paths in $G$ between known terminal pairs, except for the pair $(a,b)$. 
\end{itemize}
Additionally, let $\mathcal{F}^a$ (equivalently $\mathcal{F}^b$) be the set of faces of $G$ incident to $a$ (equiv. $b$). 
\bigskip

Our goal is to design an algorithm to implement $\textsc{EmbedNextPath}(G, a,b,\Pi)$  in polynomial time. More precisely, to draw $\pi$ on top of $G$ such that the embedding and weighting conditions are satisfied.

\bigskip Recall that a trajectory $\mathcal{F} = \langle F_{0} \in \mathcal{F}^a, F_1,\ldots,F_r\in \mathcal{F}^b \rangle$ of $\pi$ is a sequence of adjacent faces that $\pi$ passes through in the drawing of $\pi$ on $G$. The linear program in \Cref{lm:check-feasibility} for testing the validity of a trajectory and assigning weights to the edges crucially relies on knowing the trajectory of $\pi$.

\begin{algorithm}
\caption{$\textsc{FindValidTrajectory}(G, T, \pi, \Pi)$:  \Comment{An algorithm for finding a valid trajectory of $\pi$ in $G$.}}  
\begin{algorithmic}[1]
\algnotext{EndIf}
\algnotext{EndFor}

\For{each $F^a \in \mathcal{F}^a$}
\State $\mathcal{F}^{p} \leftarrow \langle F^{a}\rangle$
\State $F_{current} \gets F^{a}$
\While{$F_{current} \notin \mathcal{F}^b$}
     \For{each $F_{next} \in \mathcal{N}(F_{current})$ \textbf{and} $F_{next} \notin \mathcal{F}^{p}$}
        \If{$\textsc{CheckPrefixValidity}(G, T, \mathcal{F}^{p}\cup \{F_{next}\}, \pi,\Pi)$}   \Comment{Testing if $\mathcal{F}^{p}\cup \{F_{next}\}$ is a valid prefix.}
            \State $\mathcal{F}^{p}\leftarrow \mathcal{F}^{p}\cup \{F_{next}\}$
            \State $F_{current} \leftarrow F_{next}$
            \State \textbf{break} the for loop in line 4      
        \EndIf
    \EndFor
\EndWhile
\EndFor
\State  return $\mathcal{F}^{p}$
\end{algorithmic}
\label{alg:find-trajectory}
\end{algorithm}

\paragraph{A greedy algorithm.} To get a polynomial time algorithm, as described in \Cref{subsec:ideas}, we will build a valid trajectory $\mathcal{F}$ of $\pi$ greedily, adding one face to $\mathcal{F}$ at a time, starting from an initial guess $F^a \in \mathcal{F}^a$ until we reach a face in $\mathcal{F}^{b}$. For each face $F$ in $G$, we denote by $\mathcal{N}(F)$ the set of neighboring faces of $F$, which are faces sharing an edge in $F$. Note that there are no degree 2 vertices in $G$, except the terminals.

We say that a sequence of faces $\mathcal{F}^{p} = \langle F_0 \in \mathcal{F}^a, \ldots, F_{\ell}\rangle$ is a \EMPH{valid prefix} for $\pi$ if there exists a valid trajectory $\mathcal{F}$ for $\pi$ such that $\mathcal{F}^{p}$ is a prefix of $\mathcal{F}$.  Our key technical contribution is to design a (complicated) linear program \hyperlink{validityLP}{\textsc{PrefixValidityLP}} that tests whether the current (incomplete) trajectory can be extended to a complete, valid trajectory of $\pi$. This test allows us to decide which neighboring face of $F_{\ell}$ could be part of a valid prefix, and therefore, we can grow the prefix by one more face. The pseudocode is given in \Cref{alg:find-trajectory}. 

We note that by \Cref{clm:simple_path}, each face of $G$ appears at most once in the prefix, so in line 5 of \Cref{alg:find-trajectory}, we only consider $F_{next}$ not in $\mathcal{F}^{p}$. The bulk of the technical detail is to check the validity of a prefix in line 6. First, we need to set up some key concepts.

\subsection{The Setup}

\begin{figure}[ht]
    \centering
    \begin{subfigure}[b]{0.32\textwidth}
        \centering
        \includegraphics[page=1, trim=70pt 160pt 70pt 90pt, clip, width=\textwidth]{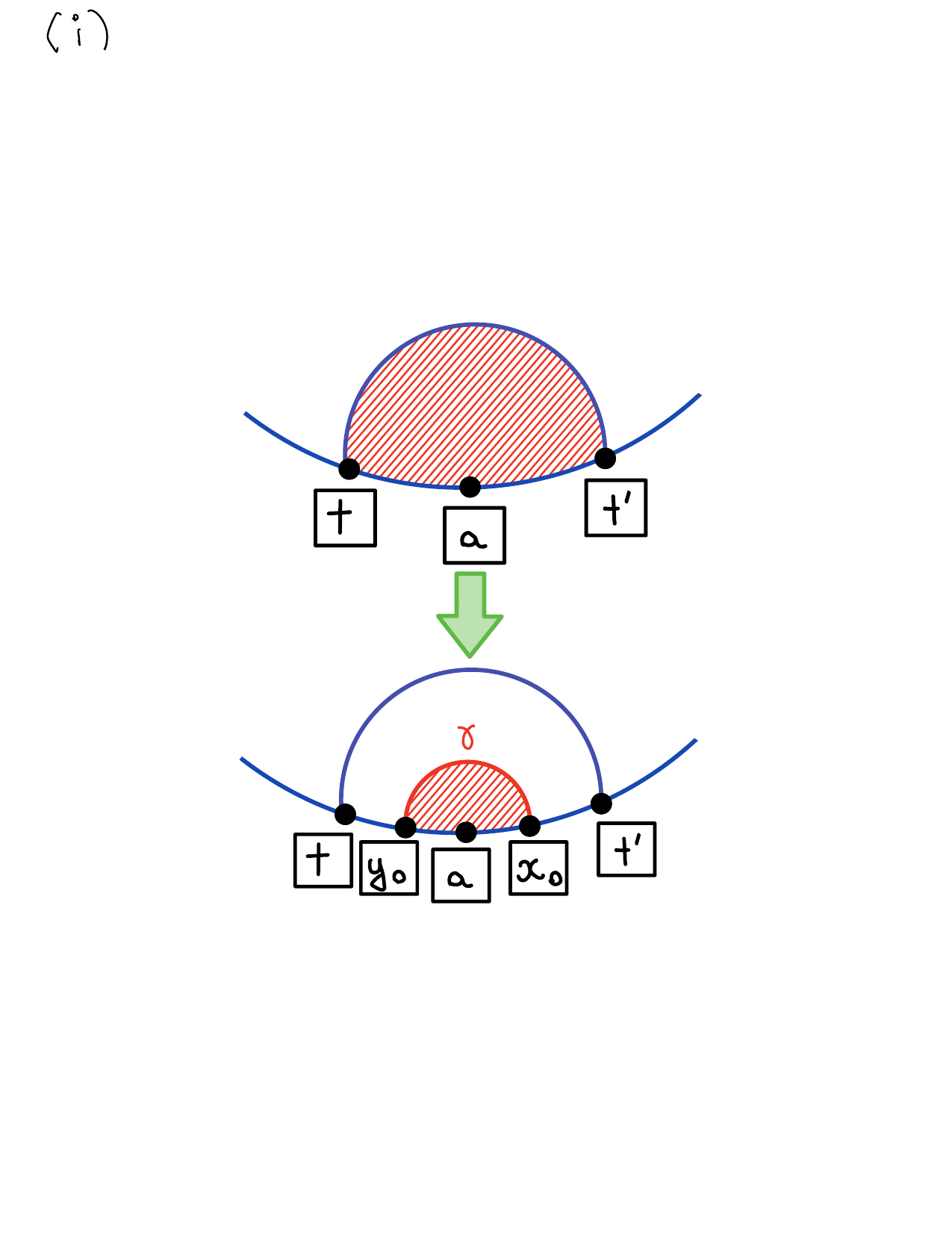}
        \caption{(i)}
        \label{fig:disci}
    \end{subfigure}
    \hfill 
    \begin{subfigure}[b]{0.32\textwidth}
        \centering
        \includegraphics[page=2, trim=70pt 120pt 70pt 90pt, clip, width=\textwidth]{figuresfinal/DiscFinal.pdf}
        \caption{(ii)}
        \label{fig:discii}
    \end{subfigure}
    \hfill
    \begin{subfigure}[b]{0.32\textwidth}
        \centering
        \includegraphics[page=3, trim=70pt 120pt 70pt 90pt, clip, width=\textwidth]{figuresfinal/DiscFinal.pdf}
        \caption{(iii)}
        \label{fig:disciii}
    \end{subfigure}
    \caption{It is convenient to think of $F^a$ (red dashed area) as the \emph{unique} face of $G$ incident to $a$. We get this property by locally modifying the embedding of $G$ at $a$. There are three cases, each requiring their own modification: (i) $|\mathcal{F}^a| = 1$, (ii) $|\mathcal{F}^a| = 2$, and (iii) $|\mathcal{F}^a| > 2$. Note that in case (iii), if $\partial F^a$ shares a non-empty segment with exactly one of $\{\pi_j, \pi_{j+1}, \dots, \pi_{j+k}\}$ then we can indeed redraw as an instance of case (ii).}
    \label{fig:disc}
\end{figure}

First, we give a formal definition of a wall and corridor edges. These concepts were introduced by Chen and Tan~\cite{ChenTan2025}, as mentioned briefly in \Cref{subsec:ideas}. To do this, we will first need to modify $G, \mathcal{F}^p, \text{and } F^a$, to guarantee a useful property of $F^a$.

To further simplify our presentation, it would be convenient if $F^a$ is the \EMPH{unique face} of $G$ incident to $a$. This property confers an important advantage: $a$ is separated from $b$ by the wall, which makes subsequent topological arguments simpler. Chen and Tan~\cite{ChenTan2025} described carving out a tiny disc around $a$ and making that the face $F^a$. For our analysis, we will also need to refer to the endpoints of the disc. We will describe how to carve the disk by locally modifying the embedding of the instance at $a$. There are three case, depending on how many faces of $G$ contain $a$ on their boundaries: (i) $|\mathcal{F}^a| = 1$, (ii) $|\mathcal{F}^a| = 2$, and (iii) $|\mathcal{F}^a| > 2$ (see Figure~\ref{fig:disc}). We denote the boundary of the outer face as $\partial G$ and model it as an undirected cycle where the distance between consecutive vertices in the circular ordering $\sigma$ is $+\infty$. (Think of these cycle edges on $\partial G$ as fake edges.) 

\begin{enumerate}[label=(\roman*)]
  \item \textbf{(Figure~\ref{fig:disci}):} $|\mathcal{F}^a| = 1$, we augment $V(G)$ by inserting auxiliary vertices $x_0, y_0$ into $\partial G$ such that the sequence $(x_0, a, y_0)$ appears in clockwise order. We set the boundary distances $\delta_{\partial G}(x_0, a) = \delta_{\partial G}(a, y_0) = 0$. We introduce a directed arc $\gamma$ in the interior of $G$, connecting $x_0$ to $y_0$. We view $\gamma$ as an undirected edge by adding $(x_0, y_0)$, $(y_0, x_0)$ to $E(G)$ with weight $\delta_{G}(x_0, y_0) = \delta_{G}(y_0, x_0) = 0$. The augmented graph is $G'$. The area delimited by combining $\gamma$ and the boundary segment $\partial G[x_0, y_0]$ contained in $F^a$ forms an infinitesimal semi-disk, and is the unique face of $G'$ incident to $a$, denoted by $F^*$.  Now we replace $F_{a}$ in $\mathcal{F}^p$  by two faces $F^*$ and $F_a \setminus F^*$; the first face in $\mathcal{F}^p$ is now $F^{*}$.
    
  \item \textbf{(Figure~\ref{fig:discii}):} If $|\mathcal{F}^a| = 2$ then there exists a known terminal pair either $(a, t)$ or $(t, a)$, with $t \neq b$. Without loss of generality assume that $(a, t)$ is a known terminal pair inserted at some previous iteration $j$, then there exists a non-empty segment $Q$ of $\partial F^a$ which is shared with $\pi_j$ and $\mathcal{F}^a\setminus\{F^a\}$, and $Q$ has $a$ as one of its endpoints. Let $t'$ be the first terminal that we encounter traveling along $\partial F^a$ by starting at $a$ and traversing $Q$ first. If the remaining piece of $\partial F^a$ continues clockwise along $\partial G$ to $a$, then we relabel $a$ as $y_0$ and augment $V(G)$ by inserting auxiliary vertices $x_0, a$ into $\partial G[t, y_0]\setminus\{t, y_0\}$ such that the sequence $(x_0, a, y_0)$ appears in clockwise order. Else relabel $a$ as $x_0$, and insert vertices $a, y_0$ into $\partial G[x_0, t]\setminus\{x_0, t\}$ such that the sequence $(x_0, a, y_0)$ appears in clockwise order. From this point, we continue analogously to (i) to update $\mathcal{F}^{p}$.

  \item \textbf{(Figure~\ref{fig:disciii}):} When $|\mathcal{F}^a| > 2$, then there exists $> 1$ known terminal pairs involving $a$ whose embeddings were inserted during previous iterations. Denote them $\{\pi_j, \pi_{j+1}, \dots, \pi_{j+k}\}$, with $j+k < i$, so that the ordering of the indices is with respect to the clockwise circular order at $a$ in the embedding of $G$. If $\partial F^a$ shares a non-empty segment with exactly one of $\{\pi_j, \pi_{j+1}, \dots, \pi_{j+k}\}$, then we apply the transformation of (ii). Otherwise $\partial F^a$ shares a non-empty segment with exactly two embeddings of those known terminal pairs, and they appear consecutively in $\{\pi_j, \pi_{j+1}, \dots, \pi_{j+k}\}$, say $\pi_{j'}, $ and $\pi_{j'+1}$. We augment $V(G)$ by splitting $a$ into three distinct vertices $x_0, a, y_0$ such that the sequence $(x_0, a, y_0)$ appears in clockwise order on $\partial G$. For all $\pi_r \in \{\pi_j, \pi_{j+1}, \dots, \pi_{j+k}\}$, if $r \leq {j'}$ then we modify the edge of $\pi_r$ incident to $a$ by replacing $a$ with $x_0$. When $r \geq {j'+1}$, we modify the edge of $\pi_r$ incident to $a$ by replacing $a$ with $y_0$. For modified edges, the embedding of $G$ is perturbed to satisfy the new assignment. From this point, we continue analogously to (i) to update $\mathcal{F}^{p}$.  
\end{enumerate}

\begin{figure}[htbp]
    \centering
    \begin{subfigure}[b]{0.24\textwidth}
        \centering
        \includegraphics[page=1, trim=90pt 160pt 90pt 150pt, clip, width=\textwidth]{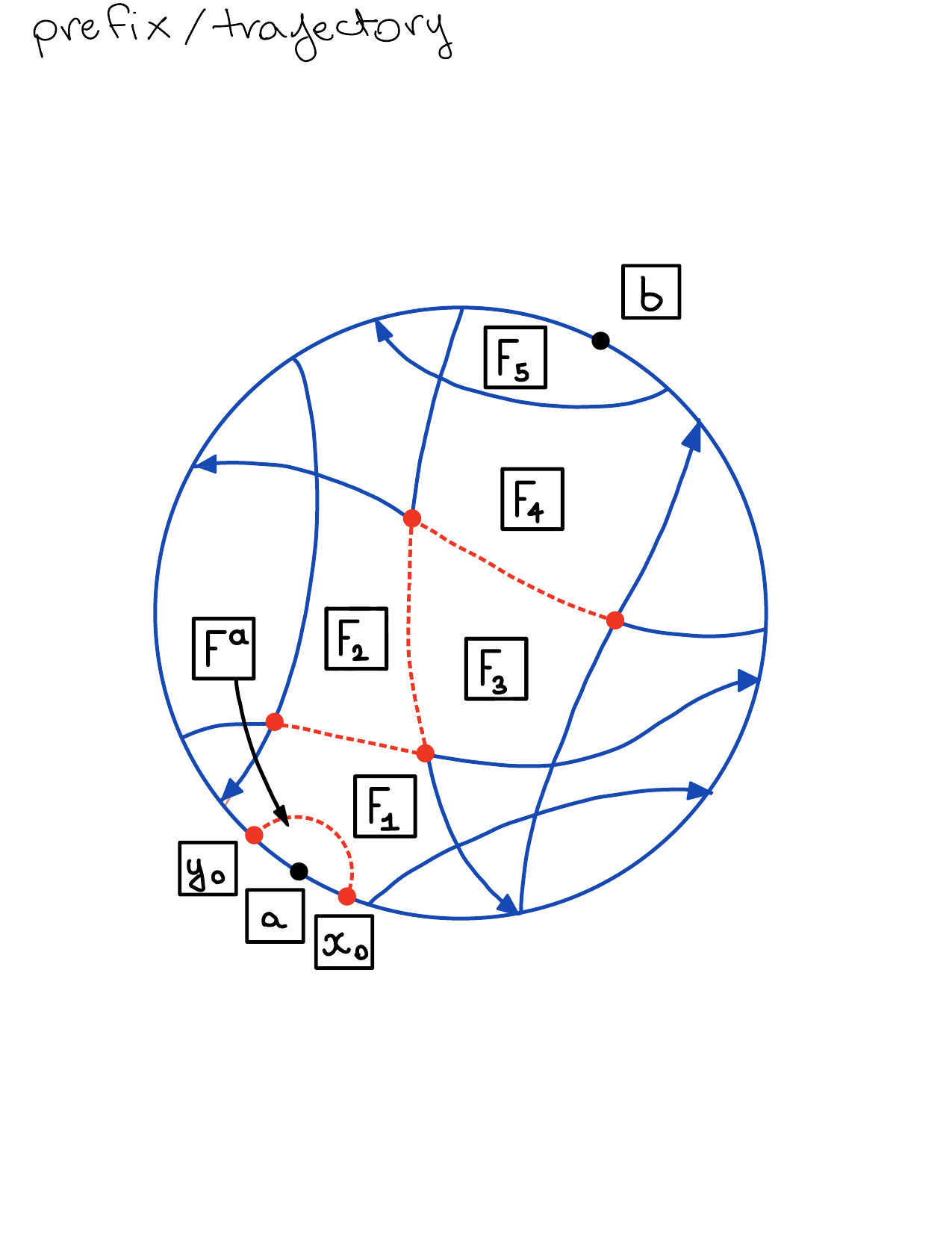}
        \caption{Candidate Prefix}
        \label{fig:validprefix}
    \end{subfigure}
    \hfill 
    \begin{subfigure}[b]{0.24\textwidth}
        \centering
        \includegraphics[page=2, trim=90pt 160pt 90pt 150pt, clip, width=\textwidth]{figuresfinal/LPConstructionFinal.pdf}
        \caption{Wall and Corridor Edges}
        \label{fig:wallcorridor}
    \end{subfigure}
    \hfill
    \begin{subfigure}[b]{0.24\textwidth}
        \centering
        \includegraphics[page=3, trim=90pt 160pt 90pt 110pt, clip, width=\textwidth]{figuresfinal/LPConstructionFinal.pdf}
        \caption{Minimal Subcurve}
        \label{fig:minimalseparator}
    \end{subfigure}
    \begin{subfigure}[b]{0.24\textwidth}
        \centering
        \includegraphics[page=6, trim=90pt 200pt 90pt 130pt, clip, width=\textwidth]{figuresfinal/LPConstructionFinal.pdf}
        \caption{Harder Case}
        \label{fig:hard}
    \end{subfigure}
    \caption{Let $\mathcal{F}^p = (F^a, F_1, F_2, F_3, F_4)$ depicted in a) be the candidate prefix of Definition~\ref{wallandcorridoredges}. In b) we emphasize the corridor edges $e_1, e_2, e_3, e^4$ (dashed red lines) and the wall $W$. The left wall (purple part of the highlighted walk) starts at $y_0$, and travels clockwise around the boundary $\partial F_1$ until an endpoint of $e_2$, then clockwise along $\partial F_2$ until an endpoint of $e_3$, and finally clockwise following $\partial F_3$ until an endpoint of $e_4$. The right wall (orange part of the highlighted walk) starts at $x_0$ then follows $\partial F_1$ in counterclockwise order, continuing analogously to the description of the left wall until $e_4$. The wall is the combined left wall, corridor edge $e_4$, and right wall (complete highlighted walk). c) The graph $G$ along with $\mathcal{F}^p$ as drawn is a special case where the minimal subcurve of the wall separating $a$ and $b$ (blue highlighted walk) is the wall itself. In b) and c) the red dashed region is the interior, and the rest is the exterior (excluding the infinite face). A case where the minimal subcurve separating $a$ and $b$ is not equal to the wall is presented in d) with the wall (complete highlighted walk), and the minimal subcurve separating $a$ and $b$ (blue highlighted walk). The interior is the combined dashed regions, and the exterior is the blank region.}
    \label{fig:wallintex}
\end{figure}
For the remainder of this analysis, we assume $G$ refers to the augmented graph $G'$, and $F^p$ has been modified so that $F_0$ is the unique face of the graph incident to $a$. Once we have found the valid trajectory for $\pi$ we undo the changes by following the above instructions in the reverse order. (We view both $x_0$ and $y_0$ as the same terminal $a$ since their pairwise distances are $0$ in $G'$.) Now, with the new unique face $F_0$ containing $a$, we can formally define wall and corridor edges. 
\begin{definition}[Wall and Corridor Edges.]\label{wallandcorridoredges}
Let $\mathcal{F}^p \;=\; (F_0,F_1,\ldots,F_\ell)$ be a candidate prefix of a (valid) any trajectory for $\pi$. For each $1\le j\le \ell$, let $e_j$ be the edge shared by the two faces $F_{j-1}$ and $F_j$. We call $e_1,\ldots,e_\ell$ the \EMPH{corridor edges} of $\mathcal{F}^p$.
For each $1\le j\le \ell-1$, let $\partial_\mathrm{L}(F_j, e_j, e_{j+1})$ be the minimal segment encountered when walking along $\partial F_j$ from an endpoint of $e_j$ to another endpoint of $e_{j+1}$ in \EMPH{clockwise order}. Similarly, let $\partial_\mathrm{R}(F_j, e_j, e_{j+1})$ be the minimal segment from an endpoint $e_j$ to another $e_{j+1}$ in \EMPH{counterclockwise} order. The \EMPH{left wall} of the prefix $\mathcal{F}^p$ is defined as $\bigcup_{j=1}^{\ell-1} \partial_\mathrm{L}(F_j, e_j, e_{j+1})$,  and the \EMPH{right wall} is defined as $\bigcup_{j=1}^{\ell-1} \partial_\mathrm{R}(F_j, e_j, e_{j+1})$. The union of the \emph{left wall}, the \emph{right wall}, and the \emph{corridor edge $e_l$} is called the \EMPH{wall}, denoted by $\mathcal{W}$, of the prefix $\mathcal{F}^p$. Formally: \[
\mathcal{W}
~=~
\left(\bigcup_{j=1}^{\ell-1} \partial_\mathrm{L}(F_j, e_j, e_{j+1})\right)
~\cup~
\left(\bigcup_{j=1}^{\ell-1} \partial_\mathrm{R}(F_j, e_j, e_{j+1})\right)
~\cup~
e_\ell.
\]
\end{definition}

\begin{observation}\label{obs:wall-separating} The wall $\mathcal{W}$ is a simple curve. Furthermore,  removing $\mathcal{W}$ will separate the disk containing the embedding of  $\bdry G_{p}$ into two regions, one containing $a$ and the other containing $b$.
\end{observation}
\begin{proof}
    $\mathcal{W}$ is simple because $\mathcal{F}^{p}$ does not have repeated faces. The second claim is because $F_0$ is the unique face containing $a$.
\end{proof}

Intuitively, $\mathcal{W}$ separates $G$ into two regions: one containing the prefix $\mathcal{F}^p$ (and the prefix of the drawing of $\pi$ as well), and the other containing possible future extensions of the prefix. We follow this intuition and give the following definitions. 

We draw the prefix of $\pi$ on $\mathcal{F}^{p}$ on $G$ as follows. Let $\pi[0] = a$. Then for  each $j \in [1,\ell]$, we view each edge $e_j$ as a segment on the plane, choose a point $\pi[j]$ on $e_j\setminus \{\text{$e_j$'s endpoints}\}$ and draw a  (directed) curve from $\pi[j-1]$ to $\pi[j]$ strictly inside the face $F_{j}$. The concatenation of the curves from $\pi[0]$ to $\pi[\ell]$ is the prefix drawing of $\pi$. Let \EMPH{$G_{p}$} be the resulting graph. That is $V(G_p) = V(G)\cup \{\pi[j]\}_{1\leq j \leq \ell}$ and $E(G_p) = E(G)\cup \{(\pi[j-1], \pi[j])\}_{1\leq j\leq \ell}$. We still refer to $\mathcal{W}$ as a wall in $G_p$. 

\begin{figure}[htbp]
    \centering
    \begin{subfigure}[b]{0.45\textwidth}
        \centering
        \includegraphics[page=4, trim=70pt 160pt 70pt 110pt, clip, width=.9\textwidth]{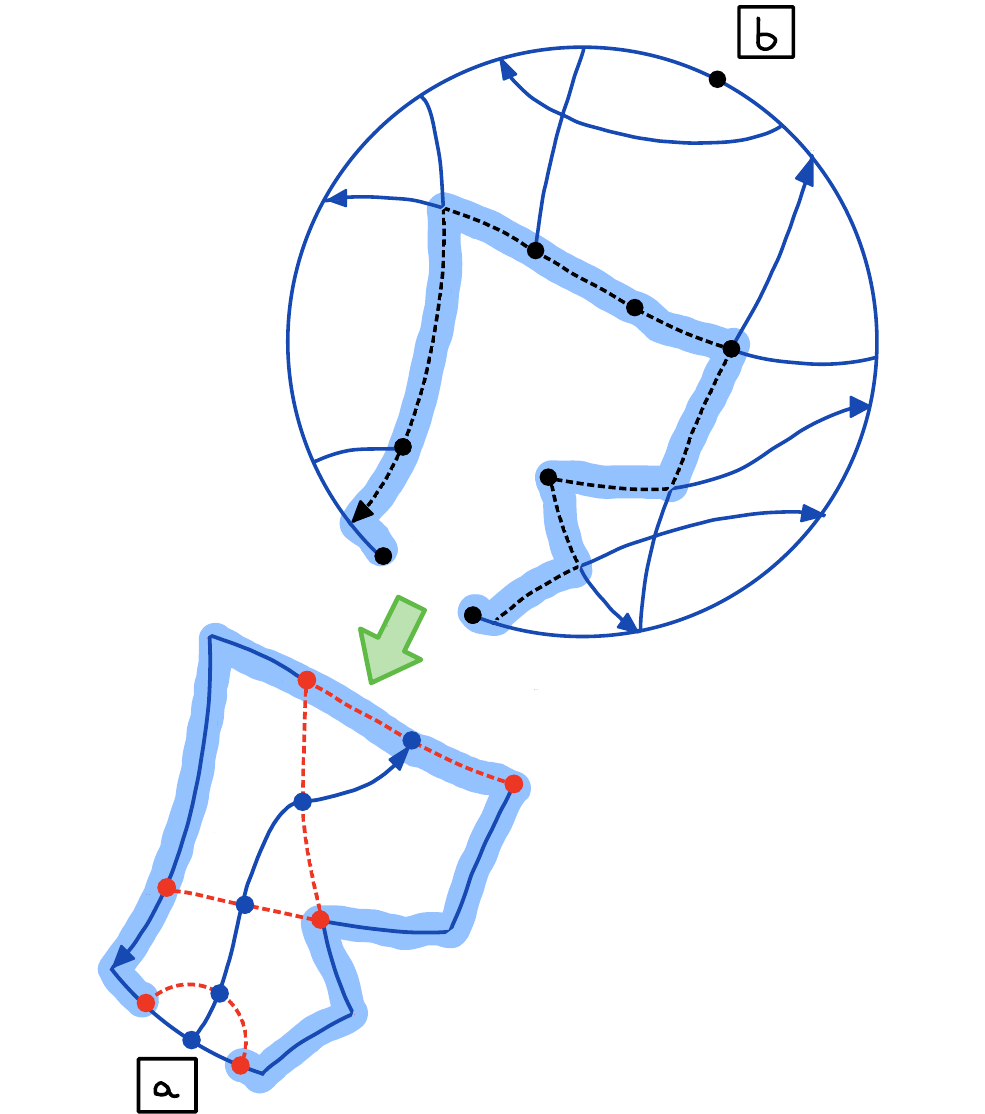}
        \caption{}
        \label{fig:cuttingeasy}
    \end{subfigure}
    \hfill 
    \begin{subfigure}[b]{0.45\textwidth}
        \centering
        \includegraphics[page=7, trim=70pt 160pt 70pt 110pt, clip, width=.9\textwidth]{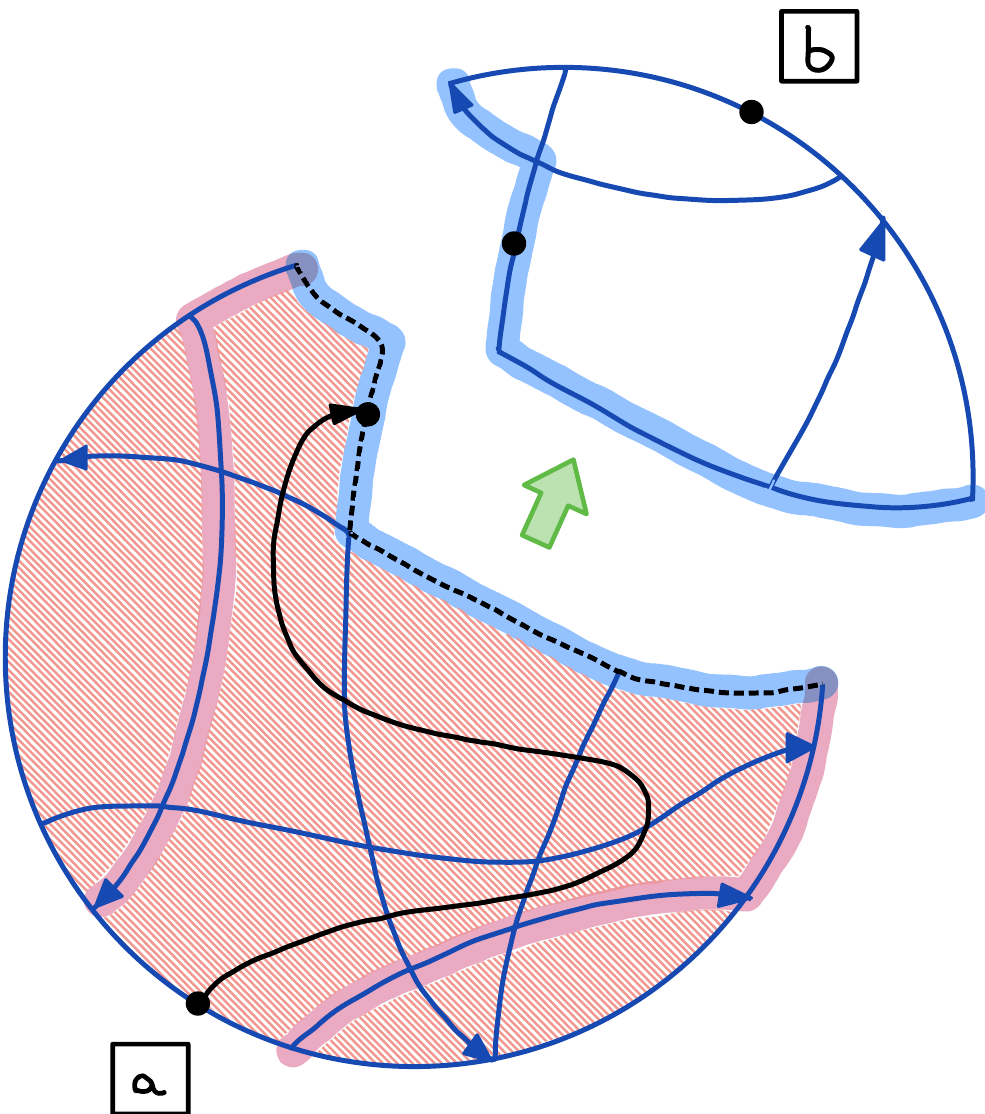}
        \caption{}
        \label{fig:cuttinghard}
    \end{subfigure}
    \hfill
    \caption{An illustration of cutting $G_{p}$ along the wall $\mathcal{W}$ in the example initialized at a) Figure~\ref{fig:validprefix}, b) Figure~\ref{fig:hard}. Note that the edges doubled along $W[x_0, y_0]$ (blue highlighted walks) are removed in the exterior (black dotted lines). The pulled out region is exactly the interior.} 
    \label{fig:cutting}
\end{figure}

\paragraph{Cutting $G_{p}$ along the wall $\mathcal{W}$.}  Let $\bdry G_{p}$ be the infinite face of $G_p$. Let \EMPH{$\mathcal{W}[x_0,y_0]$} be the minimal subcurve of $\mathcal{W}$ such that: (i) both $x_0$  and $y_0$ are on $\bdry G_{p}$ and (ii) removing $\mathcal{W}[x_0,y_0]$ will separate the disk containing the embedding of  $\bdry G_{p}$ into two regions, one containing $a$ and the other containing $b$ (see \cref{fig:minimalseparator}).   Such a subcurve $\mathcal{W}[x_0,y_0]$ exists since $\mathcal{W}$ is a valid candidate by \Cref{obs:wall-separating}. We define the operation of cutting $G_p$ along $\mathcal{W}$, denoted by \EMPH{$G_p \cut \mathcal{W}$}, as follows:
\begin{itemize}
    \item We double edges and vertices along $\mathcal{W}[x_0,y_0]$. This basically splits $G_p$ into two smaller graphs where one graph contains $a$ (and every vertex/edge inside the wall), and the other contains $b$ (and every vertex/edge outside the wall). Note that  $\mathcal{W}[x_0,y_0]$ appears in both graphs (see \Cref{fig:cutting}.)
    \item  The we delete edges (but keep vertices) of  $\mathcal{W}[x_0,y_0]$ in the subgraph containing $b$. 
\end{itemize}

The output of $G_p \cut \mathcal{W}$ is two (edge-disjoint) subgraphs, denoted by \EMPH{$F^{ex}(\mathcal{W})$} and  \EMPH{$F^{in}(\mathcal{W})$}, where  $F^{ex}(\mathcal{W})$ contains $b$ and  $F^{in}(\mathcal{W})$ contains $a$, respectively. We call $F^{ex}(\mathcal{W})$ and  $F^{in}(\mathcal{W})$ the \EMPH{exterior} and \EMPH{interior} induced by $\mathcal{W}$, respectively. We define the set of \EMPH{exterior terminals} as $T^{ex} = \bigl(T \cup V(\mathcal{W})\bigr)\ \cap\ \partial F^{ex}(\mathcal{W})$, which contains all vertices of $T$ and the wall on the infinite face of $F^{ex}(\mathcal{W})$. Similarly, we define the set of \EMPH{interior terminals} as $T^{in} = \bigl(T \cup V(\mathcal{W})\bigr)\ \cap\ \partial F^{in}(\mathcal{W})$. 

\paragraph{Decomposing $G_p$ and $\Pi$.} The idea of cutting $G_p$ along the wall is that it produces two (directed) Okamura-Seymour instance  $(F^{ex}(\mathcal{W}), T^{ex})$ and $(F^{in}(\mathcal{W}), T^{in})$, which can be seen as a decomposition of the directed Okamura-Seymour instance $(G_p,T)$. This decomposition induces a decomposition of shortest paths in $\Pi$ between terminals as we describe next.

Let $\Tilde{\pi}$ be a path in $\Pi$ from terminal $t$ to terminal $t'$.  Let $\langle u_0 = t, u_1,\ldots, u_{r}  = t'\rangle$ be the sequence of vertices on $\Tilde{\pi}$ such that each subpath $\Tilde{\pi}[u_k, u_{k+1}]$ is a maximal subpath that is either entirely contained in   $F^{ex}(\mathcal{W})$  or  in $F^{in}(\mathcal{W})$. Let \EMPH{$\Delta(\Tilde{\pi}) \coloneqq \langle  \Tilde{\pi}[u_0, u_{1}], \ldots, \Tilde{\pi}[u_{r-1}, u_{r}] \rangle$} be a decomposition of $\Tilde{\pi}$ into a sequence of subpaths.  Since  $F^{ex}(\mathcal{W})$ and $F^{in}(\mathcal{W})$ are edge-disjoint and the subpaths $\Tilde{\pi}[u_k, u_{k+1}]$ are maximal, we have:

\begin{observation}\label{obs:boundary_decomposition} If a subpath in  $\Delta(\Tilde{\pi})$ is entirely contained in $F^{ex}(\mathcal{W})$ (or  $F^{in}(\mathcal{W})$), then its endpoints are in $T^{ex}$ (or $T^{in}$, resp.). Furthermore, any two consecutive subpaths of  $\Delta(\Tilde{\pi})$ cannot both be contained in $F^{ex}(\mathcal{W})$ and cannot both be contained in $F^{in}(\mathcal{W})$.
\end{observation}

Now we can decompose the set of paths $\Pi$ into two sets of paths, denoted by \EMPH{$\Pi^{ex}$} and \EMPH{$\Pi^{in}$}, as follows. For every path $\Tilde{\pi} \in \Pi$, and for every subpath $P\in \Delta(\Tilde{\pi})$, if $P$ is entirely contained in $F^{ex}(\mathcal{W})$, we add $P$ to $\Pi^{ex}$; otherwise, we add $P$ to  $\Pi^{in}$.

\begin{claim}
\label{clm:unique_path_pi} For every path $P \in \Pi^{ex}$, there exactly one path $\Tilde{\pi} \in \Pi$ such that $P$ is a subpath of $\Tilde{\pi}$, and therefore, every terminal pair $(u,v)\in T^{ex}\times T^{ex}$ has at most one path from $u$ to $v$ in $\Pi^{ex}$. The same holds for $\Pi^{in}$.
\end{claim}
\begin{proof}
 Suppose, for contradiction, that there exist two distinct paths $P,P'\in \Pi$ such that both contain subpaths with the same ordered endpoints $(u,v)$, namely, a subpath from $u$ to $v$ in $P$ and a subpath from $u$ to $v$ in $P'$. Then $P$ and $P'$ intersect at least at $u$ and $v$, and these two intersections appear in the same order along both paths. This means that $P$ and $P'$ intersect twice in the same direction. However, by the nest reduction algorithm in Lemma~\ref{lem:nest_reduction}, which is applied whenever we find a valid trajectory of $\pi_i$ (as in line 6 of \Cref{alg:next-path-exp}), no two distinct paths in $\Pi$ can intersect twice in the same direction. This contradiction proves the claim.
\end{proof}

Next, we will use the decomposition of $G_p$ and $\Pi$ to formulate a linear program to check the validity of the prefix $\mathcal{F}^p$.

\subsection{Checking Prefix Validity via Linear Programming}

For every pair of terminals $(t,t')\in T^{ex}\times T^{ex}$, let $\mathcal{P}_{t,t'}^{ex}$ be the set of all simple paths from $t$ to $t'$ contained in $F^{ex}(\mathcal{W})$. The set $\mathcal{P}_{t,t'}^{in}$ is defined similarly w.r.t $F^{in}(\mathcal{W})$ and terminal pair $(t,t') \in T^{in}\times T^{in}$. Let $\sigma_{ex}$ denote the order in which the vertices of $T^{ex}$ appear along the boundary of $F^{ex}(\mathcal{W})$.

Recall that checking the validity of $\mathcal{F}^p$ amounts to checking if there exists a way to assign weights to edges of $G_p$ so that in the future we can realize the suffix of $\pi$ from $\pi[\ell]$ to $b$ in $F^{ex}(\mathcal{W})$ without changing the distances of known terminal pairs. Though all edges of  $F^{ex}(\mathcal{W})$ are in $G$ and hence currently have weights (assigned from the previous iteration), we now have to change their weights to account for the insertion of the prefix of $\pi$.  Therefore, the linear program of this section is  more complicated than the linear program in \Cref{lm:check-feasibility}.  

\paragraph{The variables.} There are only two types: edge variables and distance variables. Specifically, each $e \in G_{p}$ has an edge weight variable $x(e)$. Each order pair $(t,t')$ in $T^{{ex}}\times T^{{ex}}$ or $ T^{{in}}\times  T^{{in}}$ has a distance variable $d_{ex}(t, t')$ or $d_{in}(t, t')$, respectively. Note that a single pair $(t,t')$ could appear in both  $T^{{ex}}\times T^{{ex}}$ and $ T^{{in}}\times  T^{{in}}$, and in this case, we have two different variables.

\paragraph{The constraints.} The first two constraints (\Cref{c1} and \Cref{c2}) encode the simple fact that the total weight along any path between every two terminals must be at least their distance. 

Let \EMPH{$\mathcal{S}^{ex}\subseteq T^{ex}\times T^{ex}$} be the set of (ordered) pairs of the endpoints of paths in $\Pi^{ex}$. Similarly, we define \EMPH{$\mathcal{S}^{in}\subseteq T^{in}\times T^{in}$}  from $\Pi^{in}$.   For each pair $(t,t')\in \mathcal{S}^{ex}$, we denote by $\pi^{ex}(t,t')$ the shortest path between them in $\Pi^{ex}$. For every pair $(t,t') \in\mathcal{S}^{in}$, their path is defined analogously, and denoted by $\pi^{in}(t,t')$. While we generally do not know the distance between terminal pairs in $T^{ex}$ and $T^{in}$, we at least know the \emph{shortest paths} of the pairs in $\mathcal{S}^{ex}$ and $\mathcal{S}^{in}$. The next two constraints in our linear program (\Cref{c3} and \Cref{c4}) encode this fact: the total weight along each shortest path must be at most the endpoints' distance.

Recall that the prefix $\mathcal{F}^{p}$ is valid if we can draw the remaining suffix of $\pi$ from $\pi[\ell]$ to $b$ on $F^{ex}(\mathcal{W})$. To guarantee that, we add Monge constraints on $\sigma_{ex}$ (\Cref{c5}) w.r.t $F^{ex}(\mathcal{W})$. Note that we do not need the Monge constraints in  $F^{in}(\mathcal{W})$ since all the edges in the prefix of $\pi$ (from $a$ to $\pi[\ell]$) are already drawn inside  $F^{in}(\mathcal{W})$. 

 Let \EMPH{$\mathcal{S} \subseteq T\times T$} be the set of known terminal pairs. That is, $\mathcal{S}$ contains (i) pairs of the endpoints of paths in $\Pi^{ex}$ and (ii) the pair $(a,b)$ as well. Note that only pairs in $\mathcal{S}$ have non-trivial distance in the partial metric $\delta_T$; the distances of other pairs are $\ast$.  All the constraints so far are applied separately to each graph  $F^{ex}(\mathcal{W})$ and  $F^{in}(\mathcal{W})$. We need a final set of constraints, which are also the most involved, to encode the intuition that when we put $F^{ex}(\mathcal{W})$ and $F^{in}(\mathcal{W})$ together to form $G_p$, the distances between terminal pairs in $\mathcal{S}$ are realized. To describe these constraints, we introduce a supporting \EMPH{directed multigraph} $H$ constructed as follows.

The vertex set $V(H)$ of $H$ is $T^{ex}\cup T^{in}$. For every ordered terminal pair $(t,t')\in T^{ex}\times T^{ex}$ or in $T^{in}\times T^{in}$, we add an directed edge $e = (t,t')$ to $E(H)$. Since the same terminal pair $(t,t')$ could appear in both $ (T^{ex}\times T^{ex})$ and $T^{in}\times T^{in}$, we could have two different directed edges, both from $t$ to $t'$. Thus, $H$ is a directed multigraph. Each edge  $e = (t,t')$ is associated with a variable $d_H(e)$ defined as:
\begin{equation}\label{eq:dHe}
    d_H(e) \coloneqq\
\begin{cases}
d_{\mathrm{ex}}(t,t') & \text{if $e$ is created from $(t,t')\in  T^{ex}\times T^{ex}$},\\[2pt]
d_{\mathrm{in}}(t,t') & \text{if $e$ is created from $(t,t')\in  T^{in}\times T^{in}$}.
\end{cases}
\end{equation}
Note that  $d_H(e)$ is not a new variable, i.e., it is used as a replacement for an existing variable $d_{\mathrm{ex}}(t,t')$ or $d_{\mathrm{in}}(t,t')$, depending on whether $(t,t')\in  T^{ex}\times T^{ex}$ or $(t,t')\in  T^{in}\times T^{in}$. (Notation $d_H(e)$  is used to simplify the description of our linear program.) Also note that $H$ is not necessarily planar. 

For every (ordered) pair of vertices $(t,t')\in \mathcal{S}$, we denote by $\mathcal{P}^H_{t,t'}$ the set of all simple (directed) paths from $t$ to $t'$ in $H$. Recall that $\mathcal{S}$ contains the terminal pairs whose distance we need to realize in $G_p$, including $(a,b)$. If $(t,t')\not= (a,b)$, its shortest path is known, which is a path in $\Pi$, while for $(a,b)$, its shortest path is only partially known (the prefix in $F^{in}(\mathcal{W})$). In both cases, the shortest path between $(t,t')$ can be decomposed into sequence of subpaths,  each with endpoints $u,u'$ such that the ordered terminal pair $(u,u')$ belongs to $\mathcal{S}^{ex}$ or $\mathcal{S}^{in}$ (see \Cref{obs:boundary_decomposition}).  (For pair $(a,b)$, there are only two subpaths, one from $a$ to $\pi[\ell]$, and the other is from $\pi[\ell]$ to $b$, which is unknown.)   By the construction of $H$, each such subpath corresponds to a unique directed edge in $H$, including the subpath  $\pi[\ell]$ to $b$.   Moreover, these directed edges appear consecutively according to the natural order along the shortest path between $t$ and $t'$, and therefore together form a unique directed path from $t$ to $t'$ in $H$. We denote this directed path by $\pi^{H}_{t,t'}$. 

The important point here is that distances between terminals in $H$ represent the distances in $G_p$ obtained by gluing back $F^{ex}(\mathcal{W})$ and  $F^{in}(\mathcal{W})$ together. Using $H$, we can now add two constraints (\Cref{c6} and \Cref{c7}) to realize the distances between terminal pairs in $\mathcal{S}$. The first constraint specifies that the total weight along any path must be at least the shortest distance, and the second constraint guarantees that $\pi^H_{t,t'}$ is the shortest path between its endpoints. 

\begin{tcolorbox}[
 colback=white,           
  colbacktitle=gray,      
  coltitle=white,          
  colframe=gray,          
 title={\textsc{\hypertarget{validityLP}{PrefixValidityLP} (Linear Program For Checking Validity of A Prefix)}}] 
  \begin{footnotesize}
\begin{flalign}
&\sum_{e \in E(P_{t,t'})} x(e) \geq d_{ex}(t,t') \quad &\forall~  (t,t') \in T^{{ex}}\times T^{{ex}},\;\forall \text{ path } P_{t,t'} \in \mathcal{P}_{t,t'}^{ex}\label{c1} \\
&\sum_{e \in E(P_{t,t'})} x(e) \geq d_{in}(t,t') \quad &\forall~ (t,t') \in T^{{in}}\times T^{{in}}, \;\forall \text{ path } P_{t,t'} \in \mathcal{P}_{t,t'}^{in} \label{c2} \\
&\sum_{e \in E(\pi_{t,t'}^{ex})} x(e) \leq d_{ex}(t,t') \quad &\forall ~ (t,t') \in \mathcal{S}^{ex} \label{c3} \\
&\sum_{e \in E(\pi_{t,t'}^{in})} x(e) \leq d_{in}(t,t') \quad &\forall~  (t,t') \in \mathcal{S}^{in} \label{c4} \\
&\phantom{0} d_{ex}(t_1, t_3) + d_{ex}(t_2, t_4) \geq d_{ex}(t_1, t_4) + d_{ex}(t_2, t_3) &\forall ~ \langle t_1, t_2, t_3, t_4 \rangle  \subseq  \sigma_{ex} \label{c5} \\
&\sum_{e \in E(P_{t,t'})} d_H(e) \geq \delta_T(t,t') \quad &\forall~ (t,t')\in \mathcal{S},\; \forall \text{ path } P_{t,t'} \in \mathcal{P}_{t,t'}^{H}\label{c6} \\
&\sum_{e \in E(\pi_{t,t'}^H)} d_H(e) \leq \delta_T(t,t') \quad &\forall~ (t,t')\in \mathcal{S} \label{c7} \\
&\phantom{0} x(e) \geq 0,  \quad d_{in}(t_1,t'_1)  \geq 0, \quad  d_{ex}(t_2,t'_2) \geq 0  &\forall~ e \in G_{p},\; \forall~ (t_1,t'_1)\in  T^{{in}}\times T^{{in}},\; \forall~ (t_2,t'_2)\in \cup  T^{{ex}}\times T^{{ex}}\label{c8} 
\end{flalign}
\end{footnotesize}
\end{tcolorbox}

\paragraph{Solving the LP via the ellipsoid method.} To check if $\mathcal{F}^p$ is a valid prefix, we simply check if \hyperlink{validityLP}{\textsc{PrefixValidityLP}} is feasible. Clearly, the LP has an exponential number of constraints.

A key subroutine to solve  \hyperlink{validityLP}{\textsc{PrefixValidityLP}}  using the ellipsoid method~\cite{GLS81} is a polynomial-time separation oracle: given a candidate assignment to the variables $\{x(e),\; d_{ex}(t,t'),\; d_{in}(u,u') \mid e \in E(G_p), \; t,t' \in T^{ex}, \; u,u' \in T^{in}\}$, either certifies that the assignment satisfies all the constraints, or outputs a violated constraint in time polynomial in the number of variables. Our algorithm for implementing the separation oracle is given below.

\begin{tcolorbox}[
 colback=white,           
  colbacktitle=gray,       
  coltitle=white,          
  colframe=gray,          
 title={\textsc{Separation Oracle}}] 
  \textbf{Input:} an assignment to the variables $\{x(e),\; d_{ex}(t,t'),\; d_{in}(u,u') \mid e \in E(G_p), \; t,t' \in T^{ex}, \; u,u' \in T^{in}\}$\\
 \textbf{Output:} \textsc{Feasible} or a violated constraint.\\
 
 \begin{enumerate}
     \item Constraints in \Cref{c3,c4,c5,c7,c8} can be checked directly in polynomial time, since there are only polynomially many of them.
     \item  Constraints in \Cref{c1,c2} can be checked using shortest path computation. Specifically, for \Cref{c1}, we compute the shortest distance between every pair of terminals in $T^{ex} \times T^{ex}$ in the graph $F^{ex}(\mathcal{W})$ with edge weights $x(e)$. If there exists $\delta_{F^{ex}(\mathcal{W})}(t,t') < d_{ex}(t,t')$ for some pair $(t,t')$, then the constraint  \Cref{c1} corresponds to the shortest between $t$ and $t'$ in  $F^{ex}(\mathcal{W})$ is a violated constraint. Otherwise,  every directed path $P_{t,t'}$ satisfies the constraint in \Cref{c1}. The same holds for \Cref{c6}.
     \item Lastly, checking constraints in \Cref{c6} can also be reduced to shortest path computation. Specifically, we compute the shortest distance between every pair of terminals $(t,t')\in \mathcal{S}$ in $H$, where each edge   $e$ is weighted by $d_H(e)$ from \Cref{eq:dHe}. If $\delta_H(t,t') < \delta_T(t,t')$, then we have found a violated constraint corresponding to the shortest path from $t$ and $t'$ in $H$. Otherwise, all constraints in \Cref{c6} are satisfied.  
 \end{enumerate}
\end{tcolorbox}

\begin{lemma}\label{lm:oracle-polytime} The separation oracle has running time $\poly(|T|)$, and therefore, the running time of the ellipsoid method to decide the feasibility of  \hyperlink{validityLP}{\textsc{PrefixValidityLP}}  is $\poly(|T|,\log W)$.
\end{lemma}
\begin{proof}
Since $|V(G)| = O(|T^6|)$ by \Cref{lem:nest_reduction}, $|V(G_p)| = O(T^6)$. Thus, $|V(H)| = O(|T|^6)$ and $|E(H)| =  O(|T|^{12})$. Since the shortest path between all pairs of terminals can be computed in polynomial time, the separation can be implemented in time $\poly(|T|)$. 

Since the number of variables is $\poly(T)$, and the coefficients of every constraint is either $\{\pm 1\}\cup \{0,\ldots, W\}$ as $\delta_T(t,t') \in \{0,\ldots, W\}$, the ellipsoid algorithm will terminate in $\poly(|T|,\log W)$ iterations~\cite{GLS81}. Thus, the total running time to solve  \hyperlink{validityLP}{\textsc{PrefixValidityLP}}   is  $\poly(|T|,\log W)$.
\end{proof}

\paragraph{The algorithm.} The algorithm for checking if a candidate prefix for $\pi$ is valid is given in \Cref{alg:isvalidprefix}. 

\begin{algorithm}[htbp]
\caption{\textsc{CheckPrefixValidity}$(G, T, \mathcal{F}^{p}, \pi,\Pi)$} \begin{algorithmic}[1]

\State Let $G_{p}$ be obtained from $G$ by drawing the prefix of $\pi$ following  $\mathcal{F}^{p}$.
\State Let $\mathcal{W}$ be the wall of  $\mathcal{F}^{p}$.
\State  $(F^{{ex}}(\mathcal{W}), F^{{in}}(\mathcal{W}))\leftarrow G_p\cut W$ 
\State Let $T^{ex}, T^{in}$ be the terminal sets on the outerface of  $F^{{ex}}(\mathcal{W})$ and $F^{{in}}(\mathcal{W})$, respectively.
\State  Let $\Pi^{ex}, \Pi^{in}$ be paths decomposed from $\Pi$. 
\State Formulate  \hyperlink{validityLP}{\textsc{PrefixValidityLP}} based on $(F^{{ex}}(\mathcal{W}), T^{ex}, \Pi^{ex})$ and $(F^{{in}}(\mathcal{W}), T^{in}, \Pi^{in})$.
\State Check the feasibility of \hyperlink{validityLP}{\textsc{PrefixValidityLP}} using \Cref{lm:oracle-polytime}.
\State \textbf{if} \hyperlink{validityLP}{\textsc{PrefixValidityLP}} is feasible, return \textsc{True}; otherwise return \textsc{False}
\end{algorithmic}
\label{alg:isvalidprefix}
\end{algorithm}

\subsection{Correctness of the Algorithm}

In this section, to simplify the notation, we use  \EMPH{$F^{ex}$} and \EMPH{$F^{in}$} to replace $F^{ex}(\mathcal{W})$ and $F^{in}(\mathcal{W})$, respectively.

\begin{lemma}
    \textsc{CheckPrefixValidity}$(G, T, \mathcal{F}^{p}, \pi,\Pi)$ returns \textsc{True} if and only if $\mathcal{F}^p$ is a valid prefix.
\end{lemma}

\begin{proof}

For the forward direction $(\implies)$ we get an assignment of LP variables $x(\cdot),$ $d_{ex}(\cdot, \cdot)$, and $d_{in}(\cdot, \cdot)$ satisfying all the constraints of \hyperlink{validityLP}{\textsc{PrefixValidityLP}}. We will show that this implies the existence of a valid trajectory $\mathcal{F}$ with $\mathcal{F}^p$ as a prefix.

We note that a feasible solution of the linear program by the ellipsoid method might assign a value larger than $W\cdot n$ to $d_{ex}(t,t')$ or $d_{in}(t,t')$   for some $(t,t')$. This case happens when $t'$ is not reachable from $t$ in the corresponding graph. Therefore, we will interpret these large values as $+\infty$.  

\paragraph{The Exterior as an OS Instance:} 

Consider the exterior, $F^{ex}$. The exterior terminals $T^{ex}$ lie on the infinite face of $F^{ex}$, in the circular order $\sigma_{ex}$. 
Thus, we may view $(F^{ex}, T^{ex})$ as a directed Okamura-Seymour instance.
Furthermore, the set of distance variables $d_{ex}(t, t'),$ $\forall(t, t') \in T^{ex} \times T^{ex}$ forms a quasimetric, and in particular LP constraint (\ref{c5}) is satisfied as:  
\[d_{ex}(t_1, t_3) + d_{ex}(t_2, t_4) \geq d_{ex}(t_1, t_4) + d_{ex}(t_2, t_3) \forall ~ \langle t_1, t_2, t_3, t_4 \rangle  \subseq  \sigma_{ex} \text{ (Monge)},\]
which implies $(T^{ex}, d_{ex})$ is an Okamura-Seymour quasimetric. 

Recall that $\Pi^{ex}$ is the set of maximal subpaths of paths in $\Pi$ which are contained entirely in $F^{ex}$. 
With the following claim, we view $(T^{ex}, d_{ex})$ as an input to the Chen-Tan embedding framework, and $F^{ex}$ is the graph obtained after $|\Pi^{ex}|$ iterations where the known terminal pairs are exactly the pairs of the endpoints of paths in $\Pi^{ex}$.

\begin{claim}
    $F^{ex}$ conforms with the first $|\Pi^{ex}|$ iterations of the Chen-Tan embedding framework.
\end{claim}
\begin{proof}
    Follow the proof of Lemma~\ref{cor:static_recognition_iter_i}, it suffices to show that $F^{ex}$ satisfies the embedding and weighting conditions w.r.t $\Pi^{ex}$. 
    
    The embedding condition is inherited from $G$ because the construction of $\pi$ within $\mathcal{F}^p$ preserves the original embedding of $G$ on the exterior face $F^{ex}$ of $G_p$.
    
    Let $w:E(F^{ex}) \to \mathbb{R}_{\geq 0}$ be the edge-weight function induced from the edge weight variables $d_{ex}(\cdot, \cdot)$, i.e. $w(e) = x(e)$, $\forall e \in E(F^{ex})$. 
    The LP constraints (\ref{c1}), (\ref{c2}), (\ref{c3}), (\ref{c4}), and (\ref{c8}) guarantee that $w$ is an edge-weight function which satisfies the weighting condition of the Chen-Tan framework w.r.t $\Pi^{ex}$.
\end{proof}

Now, we will finish constructing the embedding of $(T^{ex}, d_{ex})$ by inserting paths into $F^{ex}$ for each of the remaining pairs of terminals, including $(\pi[\ell], b)$.
Let $\{(t_j, t'_j)\}_{j=1}^r$ be an enumeration of the pairs $(t, t') \in T^{ex} \times T^{ex}$ that satisfies three conditions: (i) $(t_1, t'_1) = (\pi[\ell], b)$, (ii) $(t_j, t_j')$ are not the endpoints of any path in $\Pi^{ex}$, and (iii) $d_{ex}(t_j, t_j') < +\infty$. 
We iteratively construct a sequence of graphs $\{F^{ex}(j)\}_{j=0}^r$ such that: 
\begin{enumerate}
    \item $F^{ex}(0) = F^{ex}$
    \item For every $1\leq j \leq r$, $F^{ex}(j)$ is an output of $\textsc{EmbedNextPath}(F^{ex}(j-1), t_j,t_j')$ \EMPH{omitting line 6} of Algorithm~\ref{alg:next-path-exp}.
\end{enumerate}
Let \EMPH{$\pi_j$} be the simple curve connecting $t_j$ to $t_j'$ embedded by $\textsc{EmbedNextPath}$ at iteration $j$.

The reason for omitting line 6 of Algorithm~\ref{alg:next-path-exp} is that for every $1 \leq j \leq r$, $F^{ex}(j)$ is exactly the graph $F^{ex}(j-1)$ augmented by the embedding of $\pi_j$. 
More formally, $V(F^{ex}(j-1)) \subseteq V(F^{ex}(j))$, $E(F^{ex}(j-1))  \subseteq E(F^{ex}(j))$.
In particular, $F^{ex}(1)$ is specifically the exterior $F^{ex}$ augmented by the embedding of $\pi_1$, a directed path from $\pi[\ell]$ to $b$. 

\paragraph{Constructing an Edge-Weight Function for $F^{ex}(1)$:} 

First we must find an edge-weight function for $F^{ex}(r)$.
By the weighting condition on the graph $F^{ex}(r)$ there exists an edge-weight function $w^r: E(F^{ex}(r)) \to \mathbb{R}_{\leq 0}$ such that:
\begin{itemize}
    \item (P1) every path $P \in \Pi^{ex}$ with endpoints $(t, t')$ is a directed shortest path from $t$ to $t'$and the sum of the weights of its edges under $w^r$ is equal to $d_{ex}(t, t')$,
    \item (P2) for every $(t, t') \in T^{ex} \times T^{ex}$ with $d_{ex}(t, t') = +\infty$, there is no directed path between $t$ and $t'$ in $F^{ex}(r)$.
    \item (P3) for every $1\leq j \leq r$ the curve $\pi_j$ is a directed shortest path from $t_j$ to $t_j'$ and the sum of the weights of its edges under $w^r$ is equal to $d_{ex}(t_j, t_j')$,
\end{itemize}
In other words, $F^{ex}(r)$ with $w^r$ realizes $(T^{ex}, d_{ex})$. 

\begin{remark}
    Intuitively, in the construction of the LP we have encoded minimum path lengths for pairs $(t, t') \in T^{ex} \times T^{ex}$ in $d_ex(\cdot, \cdot)$ so that, however $\textsc{EmbedNextPath}$ embeds the path $\pi_1$ from $\pi[\ell]$ to $b$ in $F^{ex}(1)$, there exists an edge-weight function such that $\pi_1$ does not shortcut any paths inserted in previous iterations of the embedding process. 
    While $w_r$ perfectly satisfies each of those minimums, the edge-set of $w_r$ is different than $E(F^{ex}(1))$.
\end{remark}

We will iteratively construct an edge-weight function $w^1: E(F^{ex}(1)) \to \mathbb{R}_{\geq 0}$ from $w^r$ which retains the first two above properties, but departs from the third property in its strict sense. Instead of realizing the original distances for all $1 \leq j \leq r$, $w^1$ will satisfy the following relaxed condition for $\pi_j$:
\begin{itemize}
    \item (P3r) $d_{F^{ex}(1)}(t_j, t_j') \geq d_{F^{ex}(r)}(t_j, t_j') = d_{ex}(t_j, t_j')$, where $d_{F^{ex}(r)}$ is the shortest path metric on $F^{ex}(r)$ induced by $w^r$ (likewise for $d_{F^{ex}(1)}$, $F^{ex}(1)$, $w^{1}$).
\end{itemize}
Note that we allow $d_{F^{ex}(1)}(t_j, t_j') = +\infty$. We will say that $w^j$ for $j \in [1, r]$ \EMPH{pseudo-realizes} $(T^{ex}, d_{ex})$ when $w^j$ satisfies the stricter property (P1) $\forall \pi_k$ with $k \leq j$, and the relaxed property (P3r) otherwise.

\begin{claim}\label{clm:pseudorealizes}
    There exists an edge-weight function $w^1$ pseudo-realizing $(T^{ex}, d_{ex})$.
\end{claim}
\begin{proof}
    The proof of Claim~\ref{clm:pseudorealizes} is by induction.

    \bigskip 

    \textbf{Base Case $(w^r)$:} Let $(t_r, v_1), (v_1, v_2), \dots, (v_{k-1}, v_k), (v_k, t_r')$ be the sequence of edges making up the embedding of $\pi_r$ in $F^{ex}(r)$. 
    \begin{enumerate}
        \item Copy $w^{r-1}=w^{r}$.
        \item Remove the edges $(t_r, v_1), (v_1, v_2), \dots, (v_{k-1}, v_k), (v_k, t_r')$, and their associated values from $w^{r-1}$.
        \item For each vertex $v_i$ for $1 \leq i \leq k$, there are now exactly $2$ edges involving $v_i$ in $w^{r-1}$ corresponding to a corridor edge of the wall of $\pi_r$ which was split. Denote them $(u, v), (v, w)$, with $u, w \in V(F^{ex}(r-1))$, and $v \notin V(F^{ex}(r-1))$.
        \item Remove, $(u, v), (v, w)$ from the current graph, and add a directed edge $(u,w)$ with weight $w^{r-1}(u, w) = w^r(u, v) + w^r(v, w)$.
    \end{enumerate}
    In the construction, we do not modify the weights of any edges on paths $\pi_j$ for $1 \leq j \leq r-1$, or paths $P \in \Pi^{ex}$. Thus $w^{r-1}$ is an edge-weight function for $F^{ex}(r-1)$ which pseudo-realizes $(T^{ex}, d_{ex})$.
\end{proof}

\textbf{Inductive Step: } By induction the edge-weight function $w^1: E(F^{ex}(1)) \to \mathbb{R}_{\geq 0}$ obtained from $w^r$ by repeatedly applying the above process in place pseudo-realizes $(T^{ex}, d_{ex})$.

\paragraph{Gluing $F^{ex}(1)$ with the Interior}

We construct the graph $G'$ by performing a vertex and edge identification between the boundary of $F^{in}$ and the corresponding boundary in $F^{ex}(1)$ (see Figure~\ref{fig:gluing}).
Note that $G'$ is exactly $G_p$ where we have extended the prefix of $\pi$ drawn following $\mathcal{F}^p$, from $\pi[\ell]$ to $b$ through the exterior face $F^{ex}$. 

\begin{figure}[htbp]
    \centering
    \begin{subfigure}[b]{0.49\textwidth}
        \centering
        \includegraphics[page=5, trim=90pt 100pt 90pt 110pt, clip, width=0.45\textwidth]{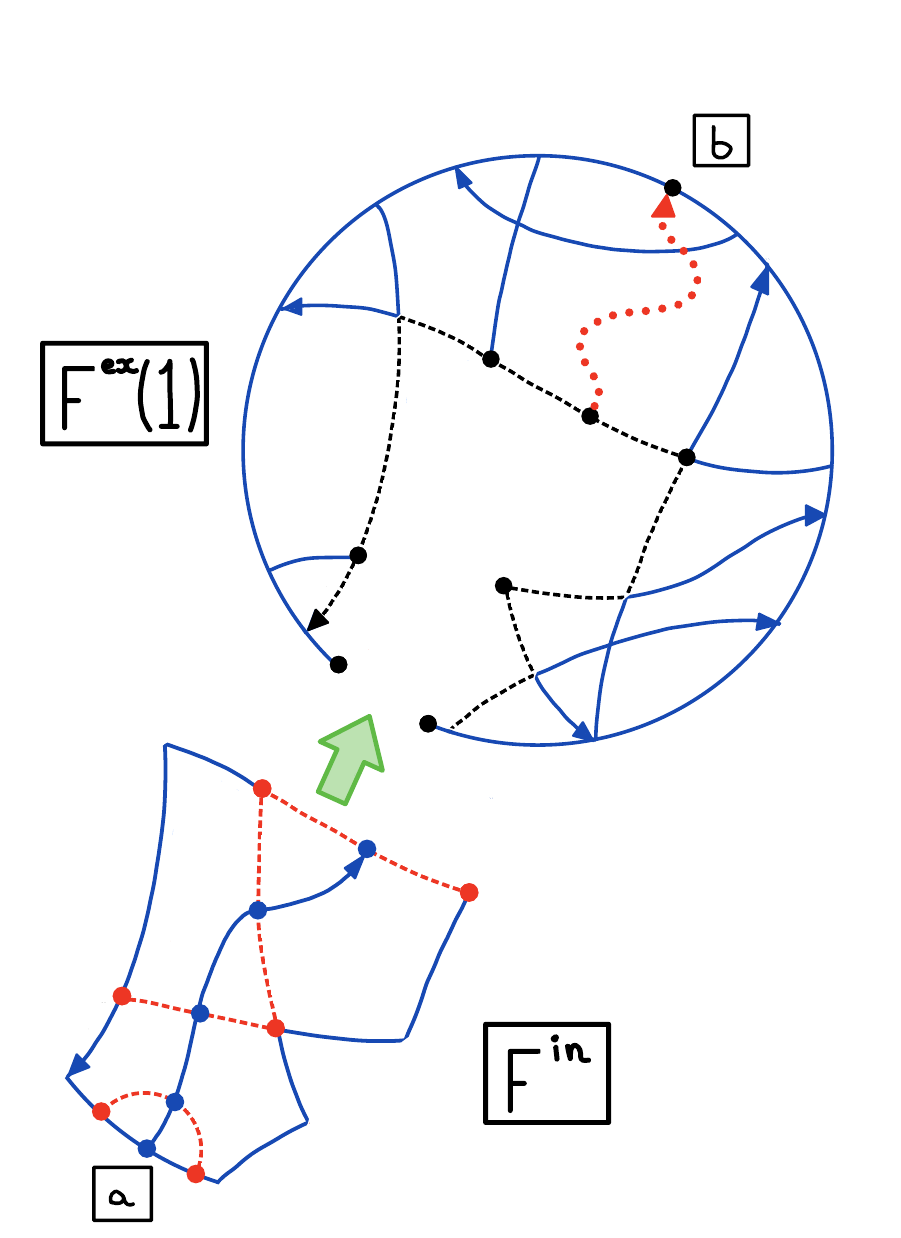}
        \caption{}
        \label{fig:gluingeasy}
    \end{subfigure}
    \begin{subfigure}[b]{0.49\textwidth}
        \centering
        \includegraphics[page=9, trim=50pt 100pt 70pt 110pt, clip, width=0.45\textwidth]{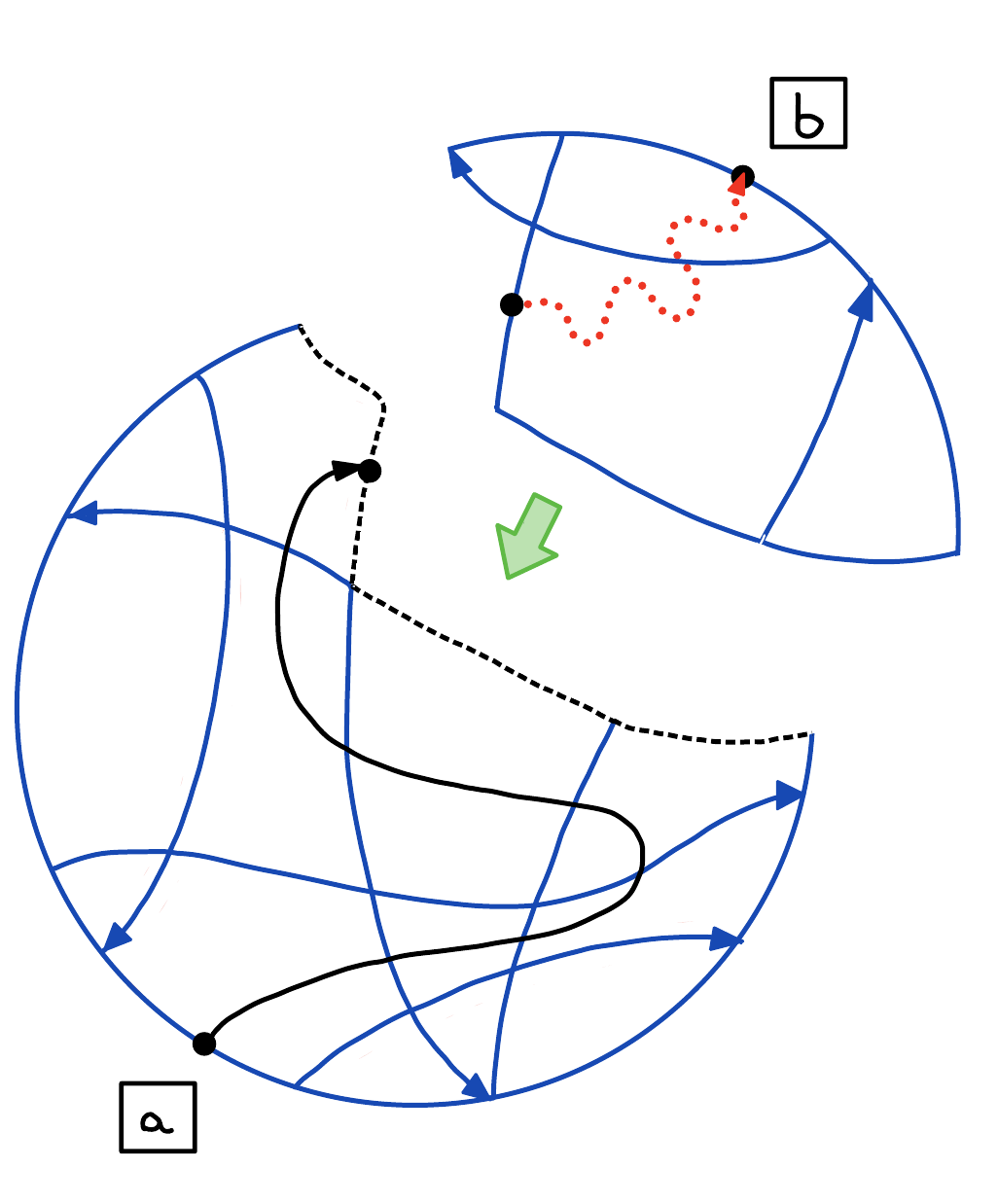}
        \caption{}
        \label{fig:gluinghard}
    \end{subfigure}
    \caption{An illustration of gluing $F^{ex}(1)$ with the interior in the example initialized at a) Figure~\ref{fig:validprefix}, b) Figure~\ref{fig:hard}. We do not know the embedding of the $\pi[\ell]$ to $b$ path in $F^{ex}(1)$, so we indicate it with a wavy dotted arrow.} 
    \label{fig:gluing}
\end{figure}

The following claim completes the proof of the forward direction $(\implies)$.

\begin{claim}\label{clm:Gprime-prop}
    $G'$ satisfies the embedding and weighting conditions w.r.t. $\Pi \cup \pi$.
\end{claim}
\begin{proof}
    The construction of $F^{ex}(1)$ ensures that $G'$ satisfies the embedding condition. Consequently, the remainder of this proof focuses on demonstrating that $G'$ satisfies the weighting condition.

    Let $w^{in}: E(F^{in}) \to \mathbb{R}_{\geq 0}$ be the edge-weight function induced by the setting of edge weight variables on the interior. More precisely: $w^{in}(e) = x(e)$ $\forall e\in E(F^{in})$. 
    We define the edge-weight function $w: E(G') \to \mathbb{R}_{\geq 0}$ as the union of edge-weight functions $w = w^{in} \cup w^1$. 
    Assign the edge weights of $G'$ according to $w$, and call $\delta_{G'}$ the induced shortest path metric. 
    We will show that $w$ satisfies the weighting condition w.r.t. $\Pi \cup \pi$.

    For contradiction, suppose the inverse; that there is some curve $\pi_j \in \Pi \cup \pi$ from $t_j$ to $t_j'$ such that either: (i) $\pi_j$ is not the shortest dipath between $t_j$ and $t_j'$ in $G'$, i.e. $\delta_{G'}(t_j, t_j') < \delta_T(t_j, t_j')$, or (ii) $\pi_j$ does not realize the value of the metric: $\sum_{e \in E(\pi_j)} \delta_{G'}(e) \neq \delta_T(t_j, t_j')$.
    
    Let $H$ be the directed multigraph constructed from $T^{ex}, T^{in}, d_{ex}$, and $d_{in}$ as described for \hyperlink{validityLP}{\textsc{PrefixValidityLP}}.

    \begin{enumerate}[label=(\roman*)]
        \item Let $\gamma$ be any path between $t_j$ and $t'_j$ in $G'$, we will show that $\sum_{e \in E(\gamma)} \delta_{G'}(e) \geq \delta_T(t_j, t'_j)$. If $\sum_{e \in E(\pi_j)} \delta_{G'}(e) = \delta_T(t_j, t'_j)$ then we immediately obtain that there are no shorter paths than $\pi_j$ from $t_j$ and $t'_j$ in $G'$.
        By a slight abuse of notation, $\Delta(\gamma)$ is defined as the decomposition of $\gamma$ into maximal subpaths either contained entirely in $\Pi^{in}$ or $\Pi^{ex}$, and according to Observation~\ref{obs:boundary_decomposition}, each subpath of $\Delta(\gamma)$ has its endpoints in either $T^{ex}$ or $T^{in}$.
        Furthermore, each ordered pair $(t, t') \in T^{ex}$ (equiv. $(t, t') \in T^{in}$) corresponds to an edge in $H$ with edge weight exactly $d_{ex}(t, t')$ (equiv. $d_{in}(t, t')$).

        Let $\Delta^{in}(\gamma) \subseteq \Delta(\gamma)$ be the set of subpaths contained entirely in $\Pi^{in}$, and likewise $\Delta^{ex}(\gamma)$ is the subpaths contained entirely in $\Pi^{out}$. 
        For each subpath $P \in \Delta^{in}(\gamma)$ with endpoints $t_P, t'_P \in T^{in}$, the LP constraints (\ref{c2}), (\ref{c4}), and (\ref{c8}) guarantee that:
        \[\sum_{e \in E(P)} w = \sum_{e \in E(P)} w^{in} = d_{in}(t_P, t'_P) = d_H(e=(t_P, t'_P)).\]
        For subpaths $Q \in \Delta^{ex}(\gamma)$ with endpoints $t_Q, t'_Q \in T^{ex}$, as $w^1$ pseudo-realizes $(T^{ex}, d_{ex})$, the edge-weight function $w$ inherits:
        \[\sum_{e \in E(Q)} w = \sum_{e \in E(Q)} w^{ex} \geq d_{ex}(t_Q, t'_Q) = d_H(e=(t_Q, t'_Q)).\]
        Let $P_{t_j, t'_j} \in \mathcal{P}^H_{t_j, t'_j}$ be the path in $H$ whose edges are exactly the naturally ordered pairs of the endpoints subpaths of $\Delta(\gamma)$.
        From the above and LP constraint (\ref{c6}) we obtain the following contradiction:
        \[\delta_T(t_j, t'_j) > \sum_{e \in E(\gamma)} \delta_{G'}(e) \geq \sum_{e \in E(P_{t_j, t'_j})} d_H(e) \geq \delta_T(t_j, t'_j).\]
        
        \item Following the steps of (i) with $\gamma = \pi_j$ we immediately obtain $\sum_{e \in E(\pi_j)} \delta_{G'}(e) \geq \delta_T(t_j, t_j')$.
        Let $\Delta^{in}(\pi_j) \subseteq \Delta(\pi_j)$ be the set of subpaths contained entirely in $\Pi^{in}$, and likewise $\Delta^{ex}(\pi_j)$ is the subpaths contained entirely in $\Pi^{ex}$. 
        For each subpath $P \in \Delta^{in}(\gamma)$ with endpoints $t_P, t'_P \in T^{in}$, the LP constraints (\ref{c2}), (\ref{c4}), and (\ref{c8}) guarantee that:
        \[\sum_{e \in E(P)} w = \sum_{e \in E(P)} w^{in} = d_{in}(t_P, t'_P) = d_H(e=(t_P, t'_P)).\]
        For subpaths $Q \in \Delta^{ex}(\gamma)$ with endpoints $t_Q, t'_Q \in T^{ex}$, as $w^1$ pseudo-realizes $(T^{ex}, d_{ex})$, the edge-weight function $w$ inherits:
        \[\sum_{e \in E(Q)} w = \sum_{e \in E(Q)} w^{ex} = d_{ex}(t_Q, t'_Q) = d_H(e=(t_Q, t'_Q)).\]
        From the LP constraint (\ref{c7}) along with (i) we obtain the following expression:
        \[\delta_T(t_j, t'_j) \leq \sum_{e \in E(\pi_j)} \delta_{G'}(e) = \sum_{e \in E(P_{t_j, t'_j})} d_H(e) \leq \delta_T(t_j, t'_j).\]        which is a contradiction with $\sum_{e \in E(\pi_j)} \delta_{G'}(e) \neq \delta_T(t_j, t_j')$.
    \end{enumerate}
\Cref{clm:Gprime-prop} now follows.
\end{proof}

For the converse $(\impliedby)$, suppose that $\mathcal{F}^p$ is a valid prefix, we will show that \textsc{CheckPrefixValidity} returns \textsc{True} by constructing a feasible solution to \hyperlink{validityLP}{\textsc{PrefixValidityLP}}.

\paragraph{Setting the LP variables $x(\cdot), d_{ex}(\cdot, \cdot)$, and $d_{in}(\cdot, \cdot)$:}

There exists a valid trajectory $\mathcal{F}$ of $\pi$ such that $\mathcal{F}^p$ is a prefix of $\mathcal{F}$.
Let $G'$ be the graph obtained by drawing the remaining suffix of $\pi$ on top of $G_p$ according to $\mathcal{F}\setminus{F^p}$.
Let $w'$ be an edge-weight function satisfying the weighting condition with $G'$. In particular, $(G', T)$ under $w'$ is a nest realizing the partial quasimetric induced by pairs of the endpoints of paths in $\Pi$ and $(a, b)$. 

Let $\{(\pi[0]=a, \pi[1]), (\pi[1], \pi[2]), \dots, (\pi[\ell-1], \pi[\ell]), (\pi[\ell], \pi[\ell+1]), \dots, (\pi[\ell+k-1], \pi[\ell+k]=b)\}$ be the sequence of directed edges making up the embedding of $\pi$ in $G'$, where $\{(\pi[0]=a, \pi[1]), (\pi[1], \pi[2]), \dots, (\pi[\ell-1], \pi[\ell])\}$ is the prefix of $\pi$ shared with $G_p$.
We obtain an edge-weight function $w^p: E(G_p) \to \mathbb{R}_{\ge 0}$ from $w'$ via the following process. 

\begin{enumerate}
    \item Copy $w^p=w'$.
    \item Remove the edges $\{(\pi[\ell], \pi[\ell+1]), \dots, (\pi[\ell+k-1], \pi[\ell+k]=b)\}$, and their associated values from $w^p$.
    \item For each vertex $v \in \{\pi[\ell+1], \dots, \pi[\ell+k-1]\}$ there are now exactly $2$ edges involving $v$ in $w^p$ (corresponding to split corridor edges of the wall induced by $\mathcal{F}$), denote them $(u, v), (v, w)$. Remove the values associated with $(u, v), (v, w)$ from $w^p$ as $v$ is not a vertex of $G_p$, and add $w^p(u, w) = w'(u, v) + w'(uv, w)$.
\end{enumerate}
By the construction, $w^p$ is an edge-weight function of $G_p$ satisfying:
\begin{itemize}
    \item $\delta_{G_p}(t, t') = \delta_T(t, t')$, $\forall$ pairs $(t, t')$ that are the endpoints of a path in $\Pi$.
    \item $\delta_{G_p}(u, v) \geq \delta_{G'}(u, v)$, $\forall u, v \in V(G_p)$.
\end{itemize}
Where $\delta_{G_p}$ is the shortest path metric induced from $w^p$.

We immediately obtain a setting of edge weight variables $x(e) = w^p(e)$ $\forall e \in E(G_p)$. 
For the exterior distance variables we compute for each ordered pair $(t, t')$ in $T^{ex} \times T^{ex}$, $d_{ex}(t, t') = \delta_{G_p \setminus F^{ex}}(t, t')$ where $\delta_{G_p \setminus F^{in}}$ is the shortest path metric induced by $w^p$ on the subgraph $F^{ex}$ of $G_p$. If there is no path on $F^{ex}$ connecting $t$ and $t'$ then set $d_{ex}(t, t') = \infty$. Analogously, for each ordered pair $(t, t') \in T^{in} \times T^{in}$, we define the interior distance variables as $d_{in} = \delta_{G_p \setminus F^{in}}(t, t')$, where the metric is induced by $w^p$ on the subgraph $F^{in}$. Again, when there is no path on $F^{in}$ connecting $t$ and $t'$ then set $d_{in}(t, t') = \infty$. 

It remains to show that $x(\cdot)$, $d_{ex}(\cdot, \cdot)$, $d_{in}(\cdot, \cdot)$ satisfy the constraints of \hyperlink{validityLP}{\textsc{PrefixValidityLP}}. 

\paragraph{Satisfying the LP constraints:}

The first four constraints of \hyperlink{validityLP}{\textsc{PrefixValidityLP}} are satisfied by construction. Specifically, (\ref{c1}) and (\ref{c2}) follow directly from the definitions of $d_{ex}(\cdot, \cdot)$, $d_{in}(\cdot, \cdot)$, and (\ref{c3}) and (\ref{c4}) from the asserted properties of $w^p$. 

Let $(F'^{ex}, F'^{in})$ be the result obtained from $G'\cut \mathcal{W}$. Then by our initial assumption $(F'^{ex}, T^{ex})$ is a nest realizing quasi-metric $(T^{ex}, d_{ex})$, and $d_{ex}$ satisfies the Monge property w.r.t $\sigma_{ex}$ ((\ref{c5})). 

Let $H$ be the directed multigraph constructed from $T^{ex}, T^{in}, d_{ex}$, and $d_{in}$ as described for \hyperlink{validityLP}{\textsc{PrefixValidityLP}}. Then (\ref{c6}) and (\ref{c7}) are immediate from the asserted properties of $w^p$ derived from $w'$ and the fact that $\mathcal{F}$ is a valid trajectory, and thus $(G', T)$ is a nest realizing $(T, \delta_T)$ with edge-weight function $w'$. 

Finally (\ref{c8}) is trivial, which completes the proof.

\end{proof}

\section{OS Embedding with Minimum Distortion}

In this section, we extend the technique for finding a planar embedding of an Okamura-Seymour quasimetric in \Cref{sec:polytime} to solve the  Okamura-Seymour embedding problem in \Cref{thm:os-embedding} in polynomial time. 

\OSEmbedding*

\begin{proof}

First, we solve the following decision problem.
\begin{quote}
     Given a quasimetric $(T,\delta_T)$,  a permutation $\sigma_T: T \rightarrow T$, and a parameter $\alpha > 0$, decide if there exists an  Okamura-Seymour embedding $(G,T)$ of $(T,\delta_T)$ w.r.t. $\sigma_T$ such that the distortion is at most $\alpha$.
\end{quote}

We solve the decision problem by linear programming. Specifically, for every pair $(t,t')\in T\times T$, we assign a variable $d(t,t') \geq 0$, which represents the distance $\delta_G(t,t')$ in $G$. Since distances in $G$ satisfy the Monge property w.r.t $\sigma_T$, $\{d(t,t')\}_{(t,t')\in T\times T}$ also must satisfy the same property. Thus, we have the following constraints:

\begin{align}
    d(t_i,t_k) + d(t_j,t_\ell) &\ge d(t_i,t_\ell) + d(t_j,t_k)
    \label{eq:min-dist-monge} \quad &\text{$\forall$  $\langle t_i,t_j,t_k,t_\ell\rangle \subseq\sigma_{T}$} \\
   \delta_T(t,t') \leq   & d(t,t') \leq \alpha\cdot   \delta_T(t,t')  \qquad &\forall t,t' \in T\times T
    \label{eq:min-dist-approx}\\
    d(t,t')&\geq 0 \qquad &\forall~ t,t' \in T\times T \label{eq:min-dist-diag}
\end{align}

The second set of constraints in \Cref{eq:min-dist-approx} encodes that the distortion is at most $\alpha$.

Since the linear program only has $O(|T|^4)$ constraints, we can decide if it has a feasible assignment to variables $\{d(t,t')\}_{(t,t')\in T\times T}$  in polynomial time. By \Cref{thm:Chen-Tan}, as long as the assignment of the variables satisfies \Cref{eq:min-dist-monge}, there exists an Okamura-Seymour instance $(G,T)$ realizing $\{d(t,t')\}_{(t,t')\in T\times T}$ as distances between terminals. Thus, the linear program is feasible if and only if the answer to the decision problem is \textsc{Yes}.

Using binary search and solving the decision problem, we can find the minimum distortion $\alpha^*$, as well as the associated distances $\{d^*(t,t')\}_{(t,t')\in T\times T}$ in polynomial time. By \Cref{thm:main}, we can construct an Okamura-Seymour instance  $(G^*,T)$ realizing $\{d^*(t,t')\}_{(t,t')\in T\times T}$ in polynomial time. Then $(G^*,T)$ is  an  Okamura-Seymour embedding $(G,T)$ of $(T,\delta_T)$ w.r.t. $\sigma_T$ of minimum distortion $\alpha^*$.
\end{proof}

\section{Application: Approximate SSSP in CONGEST Model}

In this section, we present an algorithm for computing a \((1+\epsilon)\)-approximate SSSP, as stated in \Cref{thm:distributed}. We also introduce notation that will be used only within this section.

\Distributed*

\subsection{Preliminaries}

We assume familiarity with the CONGEST model. In our high-level algorithm stated in \Cref{subsec:SSSP}, we recursively compute \EMPH{balanced cycle separators} and build a \EMPH{bounded diameter decomposition (BDD)} of the input graph $G$. To obtain an approximate SSSP tree rooted at any vertex $s \in G$, we process the resulting recursion tree from the bottom to the top and compute \EMPH{distance labels} for every vertex in $G$, inspired by Parter's~\cite{Parter20} reachability labeling scheme.

To achieve efficient communication, we adopt the approximate distance oracle by Thorup ~\cite{Thorup04} to broadcast compressed distance information between recursion levels. Such an oracle can only be computed on planar graphs. This becomes possible because our \Cref{thm:os-embedding} allows us to compute an approximate \EMPH{planar distance emulator} in polynomial time at each recursion level. To differentiate it from the distance labeling scheme by Parter ~\cite{Parter20}, we refer to the scheme by Thorup ~\cite{Thorup04} as the \EMPH{distance oracle} throughout this section.

In the following subsection, we introduce the technical tools and concepts mentioned above. 

\subsubsection{Bounded Diameter Decomposition}

\begin{figure}[t]
\centering
\begin{subfigure}[t]{0.49\textwidth}
\centering
\includegraphics[page=1, trim=65pt 225pt 73pt 209pt, clip, width=\linewidth]{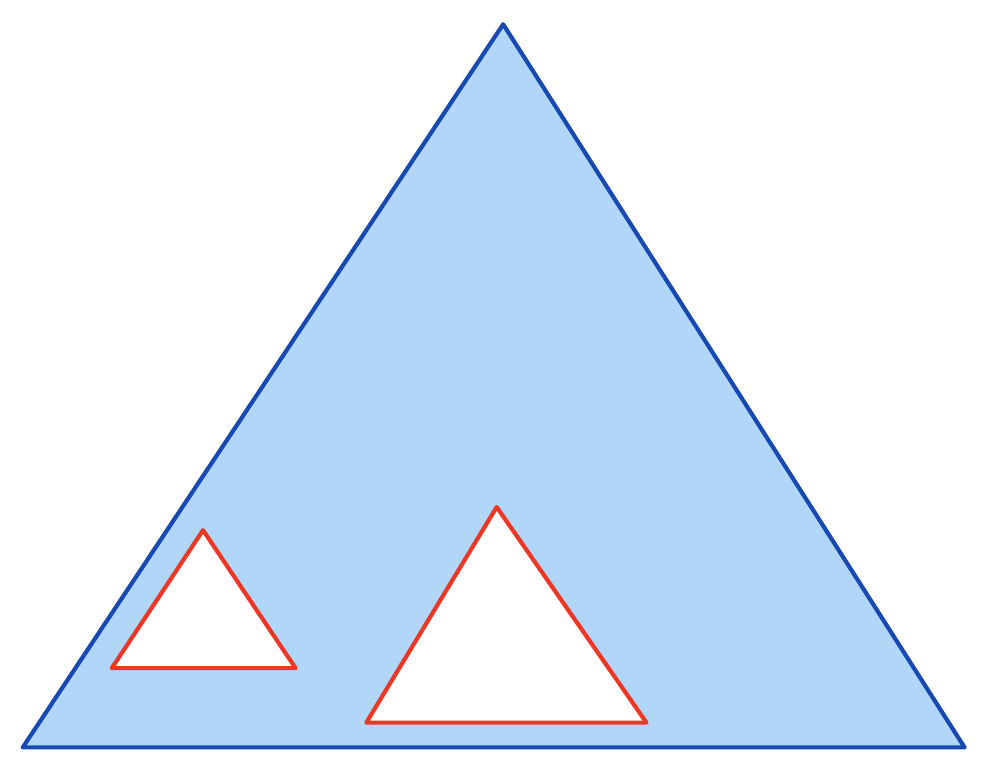}
\caption{Subgraph $G[X]$}
\label{fig:BDD-parent}
\end{subfigure}
\hfill
\begin{subfigure}[t]{0.376\textwidth}
\centering
\includegraphics[page=2, trim=100pt 174pt 102pt 146pt, clip, width=0.88\linewidth]{figuresfinal/SSSP.pdf}
\caption{Child bags of $X$}
\label{fig:BDD-parent-child}
\end{subfigure}

\caption{In this figure we give a simple illustration of the bounded diameter decomposition (BDD). 
In a), $X$ is non-leaf bag in the recursion tree induced by BDD; in b) the boundary of the green region is the balanced cycle separator $S_X$. The removal of $S_X$ partitions $X$ into two child bags, $X_1$ and $X_2$ each containing $S_X$. The drawing illustrates the tree structure. 
}
\label{fig:BDD}

\end{figure}

At a high level (see Fig.~\ref{fig:BDD}), a bounded diameter decomposition (BDD) recursively partitions a planar graph into smaller subgraphs of bounded diameter while incurring low congestion, to facilitate distributed computation. The recursion continues until the resulting subgraphs are sufficiently small to be handled directly using fast communication. Specifically, at each recursive step, we partition the planar graph using a balanced cycle separator.

\begin{remark}
In our algorithm, the sole purpose of the BDD is to partition the graph. It suffices to compute a BDD of the unweighted and undirected version of $G$. The edge weights and directions of $G$ are instead encoded in the distance labels and the distance oracle. We assume that the unweighted hop-diameter of $G$ is $D$. 
\end{remark}

\begin{definition}[Balanced Cycle Separator in ~\cite{Thorup04}]\label{def:cycle_separator}
Let $G=(V,E)$ be a weighted planar digraph, and let $T$ be a rooted spanning tree of the unweighted and undirected version of $G$. A \EMPH{balanced cycle separator} is a cycle $S$ defined by two vertices $y,z \in V$ and their lowest common ancestor $x$ in $T$. Specifically, $S$ consists of the unique path in $T$ from the ancestor $x$ to its descendant $y$, the unique path in $T$ from $x$ to its descendant $z$, and an additional edge $(y,z)$, where $(y,z)$ may be a virtual edge not contained in $E$. Moreover, the subgraph inside $S$ and every connected component of the subgraph outside $S$ each contain at most a constant fraction of the vertices of $G$.
\end{definition}

Here, for any two vertices $y,z \in V(T)$, we define the \EMPH{simple tree path} from $y$ to $z$ as the unique simple path between $y$ and $z$ in $T$. Note that neither vertex is required to be an ancestor of the other. Thus, a \EMPH{balanced cycle separator} can equivalently be viewed as the union of the simple tree path from $y$ to $z$ and an additional edge $(y,z)$, which may be a virtual edge. The following algorithm of Li and Parter computes a balanced cycle separator.

\begin{lemma}[Theorem 4.4 in ~\cite{li2019planar}]\label{lem:cycle_separator_computation}
Given an unweighted and undirected planar graph $G$ with diameter $D$ and a spanning tree $T$ of $G$, there exists an $\widetilde{O}(D)$-round algorithm that computes a balanced cycle separator of $G$.
\end{lemma}

\begin{remark}
The algorithm in \Cref{lem:cycle_separator_computation} only requires $T$ to be an arbitrary spanning tree of $G$. In our algorithm, however, we additionally require $T$ to have depth $O(D)$. We slightly modify the balanced cycle separator computation by first constructing a BFS tree $T$ using layered flooding, which can be completed in $O(D)$ rounds. Since the depth of a BFS tree is at most $D$, the resulting tree satisfies the required depth bound and can then be used in \Cref{lem:cycle_separator_computation}.
\end{remark}

We now define the bounded diameter decomposition of an unweighted and undirected graph $G = (V,E)$. A BDD is represented by a rooted tree $\mathcal{T}$ whose nodes, referred to as \EMPH{bags}, are subsets of $V$. For each bag $X \subseteq V$, we use $G[X]$ to denote the subgraph of $G$ induced by $X$. The \EMPH{root bag} is $V$, corresponding to the original graph $G$. In general, consider a bag $X \subseteq V$, and let $D'$ denote the diameter of $G[X]$. If $X$ is not sufficiently small, we compute a balanced cycle separator $S$ of $G[X]$. Removing the vertices of $S$ from $G[X]$ partitions the remaining graph into connected components $C_1,C_2,\dots,C_k$. For each component $C_i$, where $1 \leq i \leq k$, we form a \EMPH{child bag} $X_i \subseteq X$ consisting of the vertices embedded in the corresponding connected region together with the vertices of $S$ lying on its boundary. $G[X_i]$ is the subgraph corresponding to this child bag. The child bags are constructed such that:
\begin{enumerate}
\item $|X_i| \leq |X|/c$ for some constant $c>1$;
\item the diameter of $G[X_i]$ is bounded by a function of $D'$; and
\item each edge of $G[X]$ belongs to only a small number of the induced subgraphs $G[X_1], G[X_2], \dots, G[X_k]$.
\end{enumerate}
Moreover, every vertex in $X \setminus V(S)$ belongs to exactly one child bag of $X$, whereas separator vertices may belong to multiple child bags. The recursion continues independently on each child bag until the bag is sufficiently small.

This recursive procedure naturally induces the tree $\mathcal{T}$. Suppose that a bag $X$ gives rise to child bags $X_1, X_2, \dots, X_k$. Then the nodes corresponding to $X_1, X_2, \dots, X_k$ are the children of the node corresponding to $X$ in $\mathcal{T}$. Equivalently, we call $X$ the \EMPH{parent bag} of $X_1, X_2, \dots, X_k$. A bag without child bags is called a \EMPH{leaf bag}.

The following lemma provides an efficient algorithm for computing a BDD with the properties listed below. Let $G=(V,E)$ be an unweighted and undirected planar graph with diameter $D$, and let $T$ be a BFS spanning tree of $G$, rooted at an arbitrary vertex $r$, with depth $O(D)$. Similarly, for any $X \subseteq V$, we use $T[X]$ to denote the subgraph of $T$ induced by $X$.

\begin{lemma}[Definition 4.1 and Theorem 4.2 in ~\cite{li2019planar}]\label{lem:computation_of_BDD}
There exists a distributed algorithm that computes a BDD of $G$ in $\widetilde{O}(D)$ rounds. The recursive partition is represented by a rooted tree $\mathcal{T}$ satisfying the following properties:
\begin{enumerate}
\item The height of $\mathcal{T}$ is $O(\log n)$, and its root bag is $V$.
\item Every bag in $\mathcal{T}$ has a unique ID, and every vertex $v \in V$ knows the IDs of all bags containing $v$.
\item For every bag $X$, the induced subgraph $G[X]$ is connected and has diameter $O(D \log n)$.
\item For every bag $X$, $T[X]$ is a spanning tree of $G[X]$ with height $O(D)$.
\item Every leaf bag has size $O(D \log n)$.
\item For any fixed depth $d$, every edge $e \in E$ belongs to at most two subgraphs induced by bags at depth $d$ in $\mathcal{T}$.
\item For every non-leaf bag $X$, let \EMPH{$S_X$} be the balanced cycle separator of $G[X]$ computed by \Cref{lem:cycle_separator_computation}. For any two distinct child bags $X_1$ and $X_2$ of $X$, the following properties hold:
\begin{enumerate}
    \item Every vertex of $S_X$ belongs to at most three child bags of $X$. Any other vertex in $X$ belong to exactly one child bag of $X$.
    \item Every path in $G[X]$ from a vertex in $X_1$ to a vertex in $X_2$ intersects the vertices in $S_X$.
\end{enumerate}
\end{enumerate}
\end{lemma}

We also introduce some notation used later in the description of our algorithm. Let $X$ be an arbitrary bag of $\mathcal{T}$, and let $Anc^*(X)$ denote the set of proper ancestors of $X$ in $\mathcal{T}$. We use \EMPH{$S_X$} to denote either the balanced cycle separator formed by a simple tree path in $T[X]$ together with an additional (possibly virtual) edge, or the set of vertices on this cycle. We define \EMPH{$\mathbb{S}^*_X$} as the collection of balanced cycle separators used to subdivide the proper ancestors of $X$, namely, \EMPH{$\mathbb{S}^*_X$} $:= \{S_{X'} : X' \in Anc^*(X)\}$. We further define \EMPH{$\partial \mathbb{S}^*_X$} $:= \{S \cap G[X] : S \in \mathbb{S}^*_X \text{ and } S \cap G[X] \neq \emptyset\}$. As implied by Lemma 9 in ~\cite{Parter20}, which is stated below, for every $S \in \mathbb{S}^*_X$, the nonempty portion of $S$ contained in $G[X]$ forms a consecutive segment of $S$, or coincides with the entire cycle. We refer to each such portion of a separator as an \EMPH{almost-cycle}. Moreover, we define \EMPH{$\partial \mathbb{S}^+_X$} $:= \partial \mathbb{S}^*_X \cup \{S_X\}$. By a slight abuse of notation, \EMPH{$\mathbb{S}^*_X$}, \EMPH{$\partial \mathbb{S}^*_X$}, and \EMPH{$\partial \mathbb{S}^+_X$} may refer either to the corresponding collections of almost-cycles or to the sets of vertices contained in these almost-cycles, depending on the context. 

Moreover, we use the following property:
\begin{lemma}[Lemma 9 in ~\cite{Parter20}]\label{lem:logn_holes}
The vertices in $\partial \mathbb{S}^*_X$ lie on the boundaries of $O(\log n)$ holes of $G[X]$, and the separator vertices appear consecutively along the boundary of each such hole. Moreover, $\partial \mathbb{S}^*_X$ is of size $O(D\log n)$. Finally, each such hole has either an empty interior or an empty exterior.
\end{lemma}

\subsubsection{Distance Labeling and Emulator}

Distance labels store compact information at each vertex, allowing the distance between two vertices to be calculated or estimated using only their labels, without accessing the entire graph. Both the labeling scheme by Parter ~\cite{Parter20} and the distance oracle by Thorup ~\cite{Thorup04} follow this general pattern. Typically, given a (di)graph $H$ and a vertex $v$, the distance label of $v$ stores the distances between $v$ and a set of \EMPH{hub vertices} of $H$. In the directed setting, it stores both the distances from $v$ to the hub vertices and the distances from the hub vertices to $v$. Fix any pair of vertices $u,v$ in $H$, and suppose that their distance labels contain distance information with respect to a common set of hub vertices $\{s_1,s_2,\dots,s_k\}$. The hub vertices are chosen so that a shortest path from $u$ to $v$ passes through at least one vertex in $\{s_1,s_2,\dots,s_k\}$. Consequently, $\text{dist}_H(u,v)=\min_{s' \in \{s_1,s_2,\dots,s_k\}}\left\{\text{dist}_H(u,s')+\text{dist}_H(s',v)\right\}$. An analogous equality holds for the distance from $v$ to $u$.

In our algorithm, we adapt the reachability labeling scheme of Parter~\cite{Parter20} to encode approximate distance information. Instead of recording reachability relations, our labels store approximate distances to selected vertices on balanced cycle separators, which are our hub vertices.

\begin{definition}[Distance Labels in ~\cite{Parter20}]\label{def:distance_label}
Let $v \in V$, and let $d$ be the minimum depth such that there exists a bag $X$ at depth $d$ in $\mathcal{T}$ with $v \in X$, where either $v \in S_X$ or $X$ is a leaf bag. By the definition of the BDD, for every depth $i \in [0,d]$, there exists a unique bag containing $v$, which we denote by $X_i$. For each $i \in [0,d]$, the distance label \EMPH{$L(X_i,v)$} stores the distances $\text{dist}_{G}(v,s')$ and $\text{dist}_{G}(s',v)$ for every $s' \in \partial \mathbb{S}^+_X$, together with the unique ID of $X_i$. The distance label of $v$, denoted by \EMPH{$L(v)$}, is the concatenation of these labels: $L(v)=L(X_0,v)\circ L(X_1,v)\circ \dots \circ L(X_d,v)$. We refer to $L(X_i,v)$ as the \EMPH{$i$-th layer} of the distance label of $v$.
\end{definition}

In the distance oracle of Thorup~\cite{Thorup04}, the graph is recursively decomposed using separators consisting of a constant number of shortest paths. For each vertex, the oracle selects a small set of portals that serve as the hub vertices. Since we employ this distance oracle as a black box in our algorithm, we avoid the details.

\begin{lemma}[Theorem 3.16 in ~\cite{Thorup04}]\label{lem:distance_oracle}
For any weighted planar digraph $G=(V,E)$, there exists a polynomial-time algorithm that computes a distance oracle capable of answering $(1+\epsilon)$-approximate distance queries for any pair of vertices in near-constant time. Moreover, the information maintained by the oracle can be distributed to each vertex, with each having size $\OO(1/\epsilon)$.
\end{lemma}

A planar distance emulator is a compact planar (di)graph that approximately preserves the pairwise distances among a designated set of vertices, called \EMPH{terminals}. Informally, when we are concerned only with distances among a set of terminals, Thorup's distance oracle can be viewed as a non-planar distance emulator since it provides a compact representation from which the approximate distance between any pair of terminals can be recovered efficiently. In our algorithm, however, we require a directed planar distance emulator. 

\begin{definition}[Approximate Planar Distance Emulator ~\cite{CKT22}]\label{def:planar_distance_emulator}
Let $G=(V,E)$ be a weighted (di)graph, let $T_0 \subseteq V$ be a set of terminals, and let $\epsilon>0$. A weighted planar (di)graph $H$ is a \EMPH{$(1+\epsilon)$-approximate planar distance emulator} of $G$ with respect to $T_0$ if $T_0 \subseteq V(H)$ and, for every pair $u,v \in T_0$, $\text{dist}_G(u,v) \leq \text{dist}_H(u,v) \leq (1+\epsilon)\text{dist}_G(u,v)$.
The graph $H$ need not be a subgraph or a minor of $G$ and may contain additional Steiner vertices.
\end{definition}

\subsection{Algorithm}

\begin{theorem}\label{thm:distributed_distance_label_computation}
For every weighted directed planar graph $G=(V,E,\omega)$ with $|V|=n$, let $D$ denote the unweighted hop-diameter of $G$. There exists an algorithm that computes a $(1+\epsilon)$-approximate distance label of size $\widetilde{O}(D)$ for every vertex in $G$ in $\OO(D/\epsilon)$ rounds.
\end{theorem}

To begin, we fix a planar embedding of $G$ and compute a BFS spanning tree $T$ of the unweighted and undirected version of $G$ in $O(D)$ rounds. We then apply the algorithm in \Cref{lem:computation_of_BDD} to compute a BDD of $G$ in $\OO(D)$ rounds, and let $\mathcal{T}$ denote the resulting recursion tree. In the remainder of this subsection, we show how to compute the distance labels along $\mathcal{T}$.

We first show that \Cref{thm:distributed_distance_label_computation} directly implies \Cref{thm:distributed}. Suppose that the distance labels guaranteed by \Cref{thm:distributed_distance_label_computation} have already been computed. Let $s \in V$ be the source vertex of the DSSSP problem. Since the distance label of $s$ has size $\OO(D)$, it can be broadcast to every vertex in $G$ in $\OO(D)$ rounds.

Consider any vertex $u \in V$. Vertex $u$ now has access to both $L(u)$ and $L(s)$. By Property~(2) of \Cref{lem:computation_of_BDD}, we can identify the minimum-depth bag $X$ in $\mathcal{T}$ such that either $u$ and $s$ belong to two distinct child bags of $X$, or at least one of $u$ and $s$ belongs to the balanced cycle separator $S_X$. According to \Cref{def:distance_label}, the layers $L(X,u)$ and $L(X,s)$ contain distance information with respect to the common set of hub vertices $\partial \mathbb{S}^+_X$. By enumerating all hub vertices in $\partial \mathbb{S}^+_X$ and summing the corresponding distances stored in $L(X,u)$ and $L(X,s)$, we obtain $(1+\epsilon)$-approximations of both the distance from $s$ to $u$ and the distance from $u$ to $s$.

It remains to show how to compute all distance labels in $\OO(D/\epsilon)$ rounds.

\subsubsection{Hub Network}

\begin{figure}[t]
\centering
\begin{subfigure}[t]{0.49\textwidth}
\centering
\includegraphics[page=3, trim=64pt 239pt 73pt 206pt, clip, width=\linewidth]{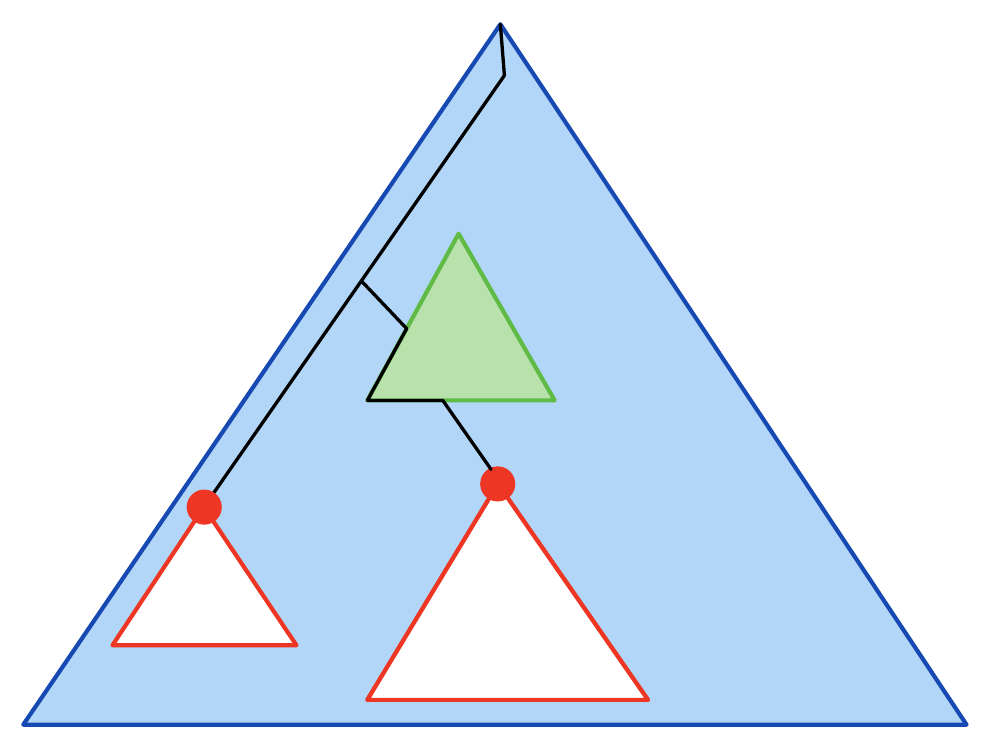}
\caption{Removing the overlapping between $S_X$ and $\mathcal{P}_X$}
\label{fig:hub-network-parent}
\end{subfigure}
\hfill
\begin{subfigure}[t]{0.49\textwidth}
\centering
\includegraphics[page=4, trim=59pt 288pt 78pt 158pt, clip, width=\linewidth]{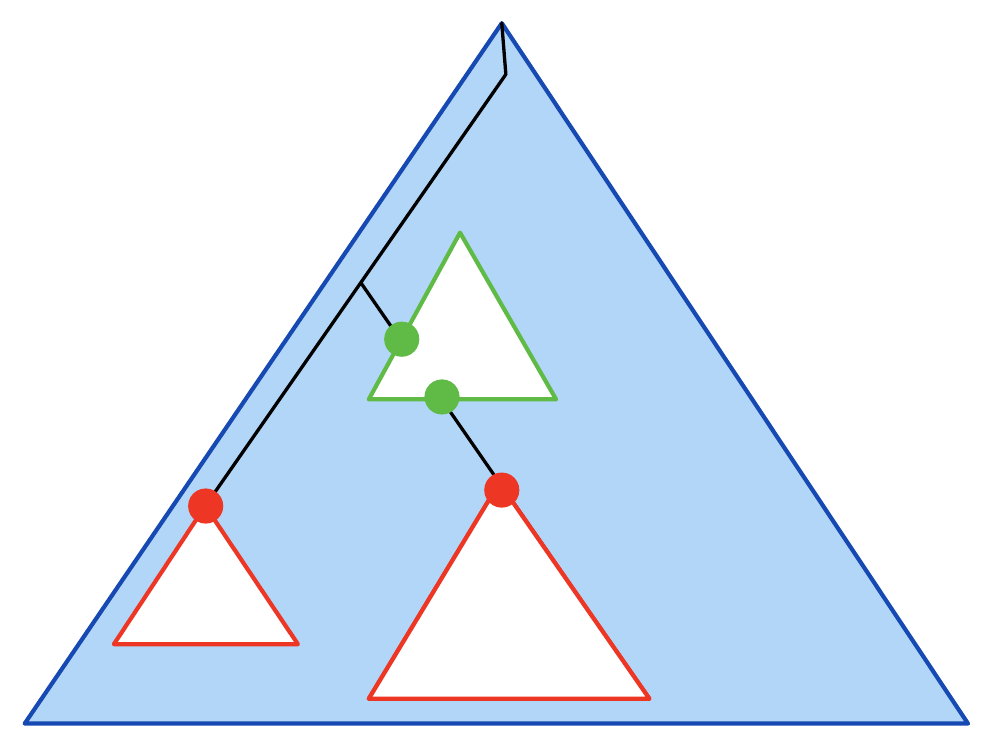}
\caption{Hub Network for the child bag}
\label{fig:hub-network-child}
\end{subfigure}

\caption{This figure illustrates the hub network construction in bag $X$. The boundary of the green region is $S_X$. In a) the black line on the boundary of the green region shows the overlapping between $S_X$ and $\mathcal{P}_X$. We remove this part from the hub network. In b) the paths represented by the black lines form the hub network for child bag $X_1$ (the child bag of $X$ outside $S_X$).}
\label{fig:hub-network}

\end{figure}

We now give an intuitive explanation of a \EMPH{hub network}. Similar to the reachability preserver of Parter~\cite{Parter20}, the key to computing the distance labels in \Cref{def:distance_label} is to obtain approximate pairwise distances among the vertices in $\partial \mathbb{S}^+_X$ for every bag $X$ in $\mathcal{T}$. We compute this information in a bottom-up manner along $\mathcal{T}$. To this end, for each child bag $X_i$ of $X$, we construct a planar distance emulator whose terminal set contains $\partial \mathbb{S}^*_{X_i}$. This allows us to compress the pairwise distance information among the vertices in $\partial \mathbb{S}^*_{X_i}$ into an approximate distance oracle, which can then be broadcast to the $X$. Thus, approximate pairwise distances among the vertices in $\partial \mathbb{S}^+_X$
can be obtained since $\partial \mathbb{S}^+_X \subseteq \bigcup_i \partial \mathbb{S}^*_{X_i}$ by definition.  

\Cref{thm:os-embedding} can only compute a planar distance emulator for an Okamura--Seymour quasimetric, where all terminals must lie on the boundary of a single face. However, by \Cref{lem:logn_holes}, the vertices in $\partial \mathbb{S}^*_X$ lie on the boundaries of $O(\log n)$ holes of $G[X]$. Our \EMPH{hub network} overcomes this obstacle by adding $O(\log n)$ unweighted and undirected paths from $T[X]$ that connect these holes, each of which has length $O(D)$. By cutting lengthwise along the aforementioned paths we create corridors which connect the previously disjoint empty sides of the holes, thereby merging them into a single hole. Consequently, in the resulting graph all vertices in $\partial \mathbb{S}^*_X$ lie on the boundary of one hole.

Formally, we give the following construction:

\begin{definition}[Hub Network]\label{def:hub_network}
For each bag $X$ of $\mathcal{T}$, we construct its corresponding \EMPH{hub network} $\mathcal{P}_X$ recursively in a top-down manner. The hub network $\mathcal{P}_X$ is a collection of simple tree paths in $T[X]$. By a slight abuse of notation, we also use $\mathcal{P}_X$ to denote the set of vertices contained in these paths.

If $X$ is the root bag of $\mathcal{T}$, we define $\mathcal{P}_X:=\emptyset$. Now suppose that $X$ is a non-leaf bag, and let $X_i$ be a child bag of $X$. We construct $\mathcal{P}_{X_i}$ as follows.

\begin{enumerate}
    \item \textbf{Initialization.} Initialize a temporary collection of paths by setting $\mathcal{P}':=\mathcal{P}_X$.

    \item \textbf{Removing overlaps with the current separator.} Recall that $S_X$ consists of a simple tree path $q$ in $T[X]$ together with an additional edge not contained in $T[X]$. For every path $p\in\mathcal{P}'$ that overlaps $q$, we remove the common subpath from $p$. More precisely: without loss of generality suppose that $p$ is a path between $a$ and $d$, with the vertices $a$--$b$--$c$--$d$ appearing in this order along $p$, where $a\neq d$ but consecutive vertices among $a,b,c,d$ may coincide. Furthermore, suppose the subpath between $b$ and $c$, is also a subpath of $q$. We remove $p$ from $\mathcal{P}'$ and add the subpath between $a$ and $b$ if $a\neq b$, and the subpath between $c$ and $d$ if $c\neq d$. (see Fig.~\ref{fig:hub-network} for visualization.)

    \item \textbf{Maintaining connectivity.} If $(\mathcal{P}_X\cup\partial\mathbb{S}^*_X)\cap S_X\neq\emptyset$, no additional path is needed, and we skip this step. Observe that Step~(2) and Step~(3) are mutually exclusive: Step~(2) is performed only when some path in $\mathcal{P}_X$ overlaps the tree-path portion of $S_X$, which implies $(\mathcal{P}_X\cup\partial\mathbb{S}^*_X)\cap S_X\neq\emptyset$, whereas Step~(3) is performed only when this intersection is empty. 
    
    If $(\mathcal{P}_X\cup\partial\mathbb{S}^*_X)\cap S_X = \emptyset$, we add a simple tree path connecting $\mathcal{P}_X\cup\partial\mathbb{S}^*_X$ to $S_X$. Since $\mathcal{P}_X\cup\partial\mathbb{S}^*_X\subseteq X$, $S_X\subseteq X$, and $T[X]$ is connected, there exists a simple tree path $p$ in $T[X]$ with one endpoint in $\mathcal{P}_X\cup\partial\mathbb{S}^*_X$ and the other endpoint in $S_X$, whose internal vertices belong to neither set. We add $p$ to $\mathcal{P}'$.

    \item \textbf{Restricting the hub network to a child bag.} After the preceding steps, every path in $\mathcal{P}'$ intersects $S_X$ only at its endpoints. Consequently, each such path is entirely contained in exactly one child bag of $X$. We define $\mathcal{P}_{X_i}:=\{p\in\mathcal{P}' : p\subseteq T[X_i]\}$.

\end{enumerate}
\end{definition}

We now explain how to compute the hub networks in the distributed setting. We only need to maintain that vertices and tree edges belonging to $\mathcal{P}_X$ know that they are in the hub network of $X$.

\begin{lemma}
    We can compute $\mathcal{P}_X$ for every $X \in \mathcal{T}$ in $\OO(D)$ rounds. 
\end{lemma}

\begin{proof}
We process the bags of $\mathcal{T}$ from top to bottom. Consider a non-leaf bag $X$, and assume that $\mathcal{P}_X$ has already been computed. 

In Step~(2) of \Cref{def:hub_network}, no global broadcast is needed. Every tree edge in $T[X]$ only checks whether it is marked as an edge of $\mathcal{P}_X$ and also lies on the tree path $q$ and removes the overlapping. 

In Step~(3), trivially, we can decide whether $(\mathcal{P}_X\cup\partial\mathbb{S}^*_X)\cap S_X$ is empty in $\OO(D)$ rounds. It it is not empty, we run a multi-source search on $T[X]$ from the vertices in $\mathcal{P}_X\cup\partial\mathbb{S}^*_X$ until reaching $S_X$. Since $T[X]$ has depth $O(D)$, this takes $\OO(D)$ rounds. 

In Step~(4), no further global communication is needed.

All bags at the same depth of $\mathcal{T}$ can be processed in parallel by Property~(6) of \Cref{lem:computation_of_BDD}. Hence, the hub networks for one level can be computed in $\OO(D)$ rounds, and over all $O(\log n)$ levels the total number of rounds is $\OO(D)$.

\end{proof}

\begin{claim} \label{clm:size_of_p}
For every bag $X \in \mathcal{T}$, the size of $\mathcal{P}_X$ is bounded by $\OO(D)$.
\end{claim}

\begin{proof}
By the construction in \Cref{def:hub_network}, at each level of $\mathcal{T}$, we add to $\mathcal{P}_X$ at most one simple path in the tree $T[X]$. Since $T[X]$ has depth $O(D)$, each such path contains $O(D)$ vertices. Moreover, the recursion tree $\mathcal{T}$ has $O(\log n)$ levels. Hence, the total size of $\mathcal{P}_X$ is bounded by $O(D \log n)=\OO(D)$.
\end{proof}

Since $\mathcal{P}_X$ and $\partial \mathbb{S}^*_X$ are of size $\OO(D)$, we let every vertex in $G[X]$ know the information of $\mathcal{P}_X$ and $\partial \mathbb{S}^*_X$ in $\OO(D)$ rounds, for every bag $X \in \mathcal{T}$. Moreover, we have the following important properties in \Cref{clm:disjoint_hub_network} and \Cref{lem:connectivity_of_hole_tree}.

\begin{claim}\label{clm:disjoint_hub_network}
For any bag $X \in \mathcal{T}$ and any simple tree path $p \in \mathcal{P}_X$, the internal vertices of $p$ are disjoint from $\partial \mathbb{S}^*_X$, i.e., $p$ may intersect $\partial \mathbb{S}^*_X$ only at its endpoints. Moreover, the paths in $\mathcal{P}_X$ have pairwise disjoint sets of internal vertices: for any two distinct paths $p,q \in \mathcal{P}_X$, no vertex is an internal vertex of both $p$ and $q$. For brevity, we say that $\mathcal{P}_X$ is \EMPH{internally disjoint} from $\partial \mathbb{S}^*_X$ and that $\mathcal{P}_X$ is \EMPH{pairwise internally disjoint}.
\end{claim}

\begin{proof}
We prove the claim by induction on the depth of $X$ in $\mathcal{T}$.

For the base case, $X$ is the root bag. By \Cref{def:hub_network}, we have $\mathcal{P}_X=\emptyset$, and hence both properties hold trivially.

Now assume that the claim holds for a non-leaf bag $X \in \mathcal{T}$. We prove that it also holds for every child bag $X_i$ of $X$.

Consider the construction of $\mathcal{P}_{X_i}$ in \Cref{def:hub_network}. We first show that, after Steps~(2) and~(3), $\mathcal{P}'$ is internally disjoint from $\partial \mathbb{S}^+_X$ and is pairwise internally disjoint.

Initially, $\mathcal{P}'=\mathcal{P}_X$. By the induction hypothesis, $\mathcal{P}'$ is internally disjoint from $\partial \mathbb{S}^*_X$ and is pairwise internally disjoint. In Step~(2), whenever a path in $\mathcal{P}'$ overlaps the tree-path portion of $S_X$, we remove the overlapping subpath and keep only the remaining subpaths. Thus, after Step~(2), no path in $\mathcal{P}'$ has an internal vertex in $S_X$. Since Step~(2) only replaces paths by their subpaths, it preserves pairwise internal disjointness and also preserves internal disjointness from $\partial \mathbb{S}^*_X$.

It remains to consider Step~(3). If Step~(3) is skipped, the desired properties already hold. Otherwise, Step~(3) adds a simple tree path whose internal vertices belong to neither $\mathcal{P}_X\cup\partial\mathbb{S}^*_X$ nor $S_X$. Since Step~(2) and Step~(3) are mutually exclusive by the construction, we have $\mathcal{P}'=\mathcal{P}_X$ before adding this path. Thus, the newly added path has no internal vertex in the previous paths of $\mathcal{P}'$ and no internal vertex in $\partial \mathbb{S}^+_X$. Hence, after Step~(3), $\mathcal{P}'$ remains pairwise internally disjoint and is internally disjoint from $\partial \mathbb{S}^+_X$.

Finally, in Step~(4), $\mathcal{P}_{X_i}$ is obtained by keeping exactly those paths of $\mathcal{P}'$ that are entirely contained in $T[X_i]$. Hence, pairwise internal disjointness is preserved. Moreover, since $\partial \mathbb{S}^*_{X_i}\subseteq \partial \mathbb{S}^+_X$, $\mathcal{P}_{X_i}$ is internally disjoint from $\partial \mathbb{S}^*_{X_i}$. This proves the claim for $X_i$, completing the induction.

\end{proof}

\begin{lemma}\label{lem:connectivity_of_hole_tree}
    For any bag $X \in \mathcal{T}$, consider the maximal simple tree paths in $T[X]$ contained in the almost-cycles of $\partial \mathbb{S}^*_X$ (A maximal simple tree path is not a subpath of any other simple tree path). The union of these paths together with the simple tree paths in $\mathcal{P}_X$ forms a connected subtree of $T[X]$. For simplicity, we say that $\mathcal{P}_X$ and $\partial \mathbb{S}^*_X$ form a connected subtree of $T[X]$.
\end{lemma}

\begin{proof}
We prove the lemma by induction on the depth of $X$ in $\mathcal{T}$.

It suffices to prove connectivity. Indeed, $\mathcal{P}_X$ and the maximal simple tree paths induced by $\partial \mathbb{S}^*_X$ are subgraphs of $T[X]$, and hence their union is acyclic.

For the base case, $X$ is the root bag. Then $\partial \mathbb{S}^*_X=\emptyset$, and by the construction in \Cref{def:hub_network}, we also have $\mathcal{P}_X=\emptyset$. Thus, the lemma holds trivially.

For the induction step, suppose that the lemma holds for a non-leaf bag $X$. We prove that it also holds for every child bag $X_i$ of $X$.

We first show that, at the end of the construction in \Cref{def:hub_network}, the temporary collection $\mathcal{P}'$ together with $\partial \mathbb{S}^+_X$ forms a connected subtree of $T[X]$. By the induction hypothesis, $\mathcal{P}_X$ and $\partial \mathbb{S}^*_X$ form a connected subtree of $T[X]$. If Step~(2) of \Cref{def:hub_network} removes an overlap, then some path of $\mathcal{P}_X$ intersects the tree-path portion of $S_X$. Hence, $\mathcal{P}_X$ is already connected to $S_X$, and therefore $\mathcal{P}_X$ together with $\partial \mathbb{S}^+_X$ is connected in $T[X]$. Step~(2) only removes common subpaths between paths of $\mathcal{P}_X$ and the tree-path portion of $S_X$; the removed portions are now represented by $S_X$ itself, which is included in $\partial \mathbb{S}^+_X$. Moreover, Step~(3) is skipped whenever Step~(2) is performed. Thus, after Step~(2), the resulting collection $\mathcal{P}'$ together with $\partial \mathbb{S}^+_X$ still forms a connected subtree of $T[X]$.

If Step~(2) does not remove any overlap and $(\mathcal{P}_X\cup\partial \mathbb{S}^*_X)\cap S_X\neq \emptyset$, then $\mathcal{P}'=\mathcal{P}_X$, and $\partial \mathbb{S}^+_X$ is already connected to $\mathcal{P}_X\cup\partial \mathbb{S}^*_X$. Hence, $\mathcal{P}'$ together with $\partial \mathbb{S}^+_X$ forms a connected subtree of $T[X]$.

It remains to consider the case where Step~(2) does not remove any overlap and $(\mathcal{P}_X\cup\partial \mathbb{S}^*_X)\cap S_X=\emptyset$. In this case, Step~(3) adds a simple tree path in $T[X]$ connecting $\mathcal{P}_X\cup\partial \mathbb{S}^*_X$ to $S_X$. Therefore, after adding this path, $\mathcal{P}'$ together with $\partial \mathbb{S}^+_X$ forms a connected subtree of $T[X]$.

Let $T'$ denote the connected subtree of $T[X]$ formed by $\mathcal{P}'$ together with $\partial \mathbb{S}^+_X$. In Step~(4) of \Cref{def:hub_network}, we define $\mathcal{P}_{X_i}:=\{p\in\mathcal{P}' : p\subseteq T[X_i]\}$. Moreover, we have $\partial \mathbb{S}^+_X\cap X_i=\partial \mathbb{S}^*_{X_i}$. Hence, $\mathcal{P}_{X_i}$ together with $\partial \mathbb{S}^*_{X_i}$ is exactly the subgraph $T'\cap T[X_i]$.

It remains to show that $T'\cap T[X_i]$ is connected. Consider any two vertices $u,v\in T'\cap T[X_i]$. Since $T[X]$ is a tree, there is a unique simple path $p_{u,v}$ between $u$ and $v$ in $T[X]$. Since $T[X_i]$ is a connected subtree of $T[X]$ by Property~(4) of \Cref{lem:computation_of_BDD}, the path $p_{u,v}$ is contained in $T[X_i]$. Similarly, since $T'$ is a connected subtree of $T[X]$ and $u,v\in T'$, the same path $p_{u,v}$ is contained in $T'$. Hence, $p_{u,v}$ is contained in $T'\cap T[X_i]$, and thus any two vertices of $T'\cap T[X_i]$ are connected.

Consequently, $\mathcal{P}_{X_i}$ and $\partial \mathbb{S}^*_{X_i}$ form a connected subtree of $T[X_i]$. This completes the induction.

\end{proof}

For each bag $X \in \mathcal{T}$, let \EMPH{$\mathcal{H}^*_X$} denote the collection of $O(\log n)$ holes of $G[X]$ guaranteed by \Cref{lem:logn_holes} for $\partial \mathbb{S}^*_X$. 

We now explain why \Cref{clm:disjoint_hub_network} and \Cref{lem:connectivity_of_hole_tree} imply that, for every bag $X \in \mathcal{T}$, cutting open the simple tree paths in $\mathcal{P}_X$ merges the empty sides of all holes in $\mathcal{H}^*_X$ into a single hole.

We first explain the standard \EMPH{cut-open operation}: when cutting open a
simple path, we duplicate the vertices and edges on the path and split the
incident edges on the two sides according to the cyclic order in the planar
embedding.

\begin{lemma}\label{lem:cut_hub_network_merges_holes}
Fix a bag $X \in \mathcal{T}$. We cut open the paths in $\mathcal{P}_X$ in an
order such that if an endpoint of a path $p$ is an internal vertex of another
path $q$, then $q$ is cut before $p$. Then all holes in $\mathcal{H}^*_X$ are
merged into a single hole. Moreover, the copies of the vertices or the original vertices in $\mathcal{P}_X \cup \partial \mathbb{S}^*_X$ lie on the boundary of this single hole. 
\end{lemma}

\begin{proof}
Such an order exists because, by \Cref{clm:disjoint_hub_network}, the paths in $\mathcal{P}_X$ are pairwise internally disjoint, and by \Cref{lem:connectivity_of_hole_tree}, $\mathcal{P}_X$ form an acyclic connected subgraph of $T[X]$.

Consider the paths in $\mathcal{P}_X$ in the prescribed order. The first path that we cut open has endpoints on the boundaries of two distinct holes in $\mathcal{H}^*_X$. Indeed, if one endpoint were an internal vertex of another path, then that other path would have to be cut first by the choice of the order. Since the internal vertices of this first path are disjoint from $\partial \mathbb{S}^*_X$ by \Cref{clm:disjoint_hub_network}, the path meets
these two holes only at its endpoints. Thus, cutting it open creates a corridor between the two holes and merges them into a single hole.

Now consider any later path $p \in \mathcal{P}_X$ inductively. By the prescribed order, if an endpoint of $p$ is an internal vertex of another path $q$, then $q$ has already been cut open. Thus, after cutting open $q$, the corresponding copy of this endpoint lies on the boundary of a merged hole by the induction assumption. Otherwise, this endpoint of $p$ is a vertex or a copy of a vertex from $\partial \mathbb{S}^*_X$, which is also a boundary vertex of holes in $\mathcal{H}^*_X$ or merged hole by the induction assumption. Thus, when $p$ is cut open, its endpoints lie on the boundaries of two current holes, where current holes may already be the result of merging several holes from $\mathcal{H}^*_X$.

By \Cref{clm:disjoint_hub_network}, the internal vertices of $p$ are disjoint from $\partial \mathbb{S}^*_X$ and the internal vertices of all other paths in $\mathcal{P}_X$. Hence, cutting open $p$ creates a corridor between these two current holes without passing through any other hole in $\mathcal{H}^*_X$ or any previously opened corridor. Therefore, cutting open $p$ merges these two current holes into a single current hole.

Finally, by \Cref{lem:connectivity_of_hole_tree}, $\mathcal{P}_X$ and $\partial \mathbb{S}^*_X$ form a connected subtree of $T[X]$. Thus, after all paths in $\mathcal{P}_X$ have been cut open, the empty sides of all holes in $\mathcal{H}^*_X$ are connected through the opened corridors. Since this structure is a subtree of $T[X]$, these corridors do not enclose any additional face. Consequently, all holes in $\mathcal{H}^*_X$ are merged into a single hole. Since all the paths are cut open and all the holes in $\mathcal{H}^*_X$ are merged, the copies of the vertices or the original vertices in $\mathcal{P}_X \cup \partial \mathbb{S}^*_X$ lie on the boundary of this single hole. 
\end{proof}

We also provide a high-level and more intuitive interpretation of the above proof. Consider the auxiliary graph obtained as follows. For each hole in $\mathcal{H}^*_X$, contract the portion of $\partial \mathbb{S}^*_X$ lying on its boundary into a single node. Since, by \Cref{clm:disjoint_hub_network}, the paths in $\mathcal{P}_X$ do not contain vertices of $\partial \mathbb{S}^*_X$ as internal vertices, no path in $\mathcal{P}_X$ is contracted in this process. Moreover, by \Cref{lem:connectivity_of_hole_tree}, before the contraction, $\mathcal{P}_X$ and $\partial \mathbb{S}^*_X$ form a connected subtree of $T[X]$. Hence, after the contraction, the resulting auxiliary graph remains connected and acyclic, which is a tree. Cutting open the paths in $\mathcal{P}_X$ can then be viewed as cutting open the edges of this contracted tree. Topologically, this opens corridors along the tree and merges the holes in $\mathcal{H}^*_X$ into a single hole (Fig. \ref{fig:SSSP-cutting-gluing} illustrates how cutting merges the holes). The boundary of the merged hole follows the contour walk of the contracted tree, equivalently the DFS tour that records both the forward steps and the backtracking steps. In particular, each cut path appears twice on the boundary, once on each side of the opened corridor.

\subsubsection{Distance Label Computation}

\begin{figure}
\centering
\begin{subfigure}[t]{0.32\textwidth}
\centering
\includegraphics[page=7, trim=50pt 250pt 74pt 158pt, clip, width=\linewidth]{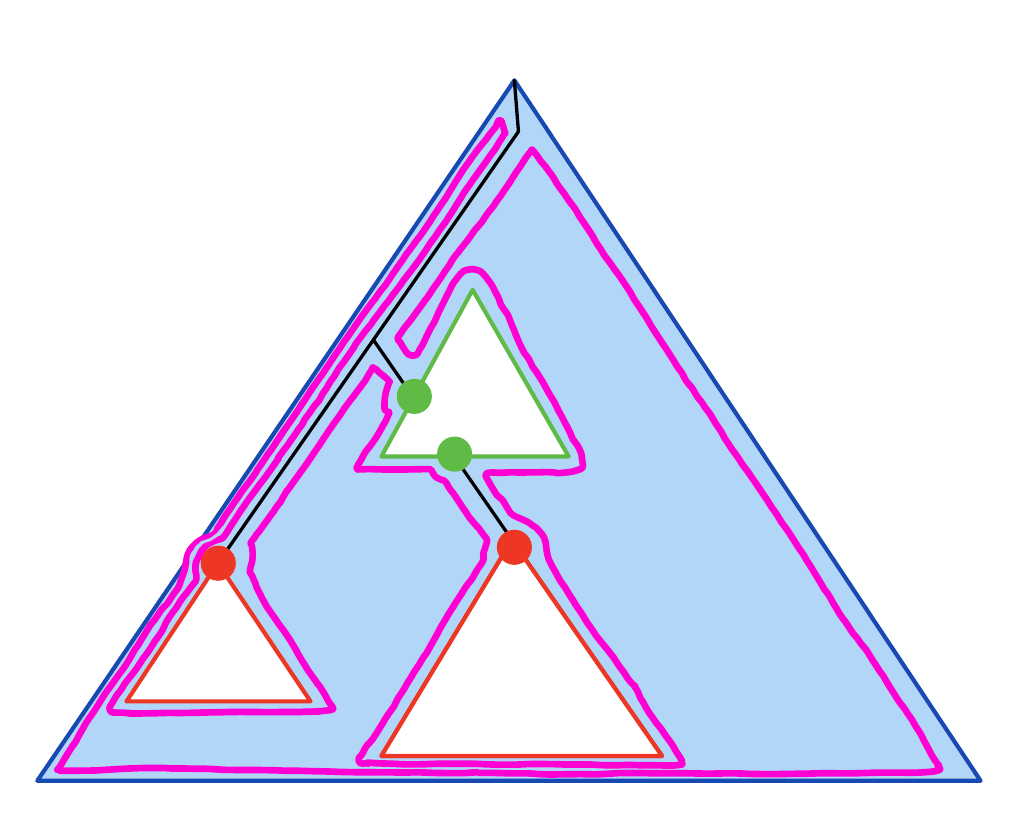}
\caption{DFS tour for the hub network of $X_1$}
\label{fig:SSSP-before-cutting}
\end{subfigure}
\hfill
\begin{subfigure}[t]{0.305\textwidth}
\centering
\includegraphics[page=5, trim=41pt 250pt 74pt 140pt, clip, width=\linewidth]{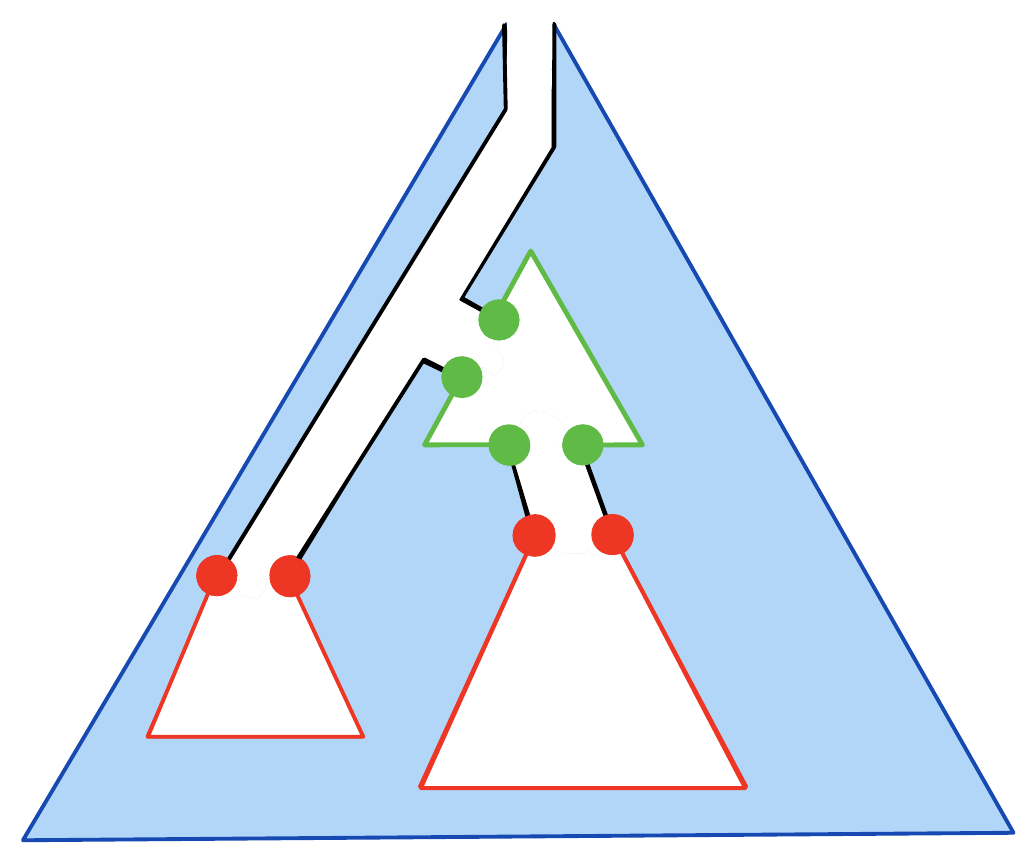}
\caption{Cutting open each path in $\mathcal{P}_{X_1}$}
\label{fig:SSSP-cutting}
\end{subfigure}
\hfill
\begin{subfigure}[t]{0.32\textwidth}
\centering
\includegraphics[page=6, trim=49pt 274pt 87pt 171pt, clip, width=\linewidth]{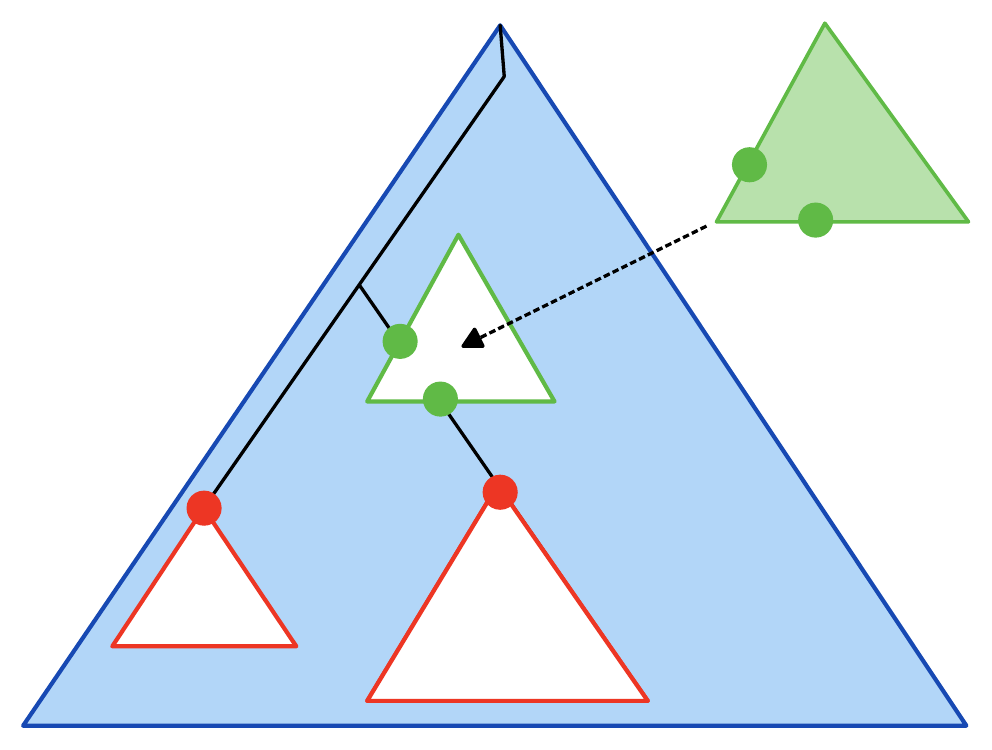}
\caption{Gluing $H_{X_1}'$ and $H_{X_2}'$ via vertices in $S_X$}
\label{fig:SSSP-gluing}
\end{subfigure}

\caption{In this figure we illustrate the cutting and gluing procedures. The DFS tour in (a) along with $\mathcal{P}_{X_1}$ and $\mathcal{H}^*_{X_1}$ corresponds to the boundary of the single hole after cutting open each path in $\mathcal{P}_{X_1}$, as illustrated in (b). In (c), in $K_{X_1}$, we glue each path in $\mathcal{P}_{X_1}$ to obtain $H_{X_1}'$. Then, we glue $H_{X_1}'$ and $H_{X_2}'$
by identifying the copies of the same vertex in $S_X$.}
\label{fig:SSSP-cutting-gluing}
\end{figure}

Now we explain how to compute the distance labels in \Cref{def:distance_label} in a bottom-up manner using the hub network. Recall that the key task is to compute the pairwise distances among the vertices in $\partial \mathbb{S}^+_X$ for every bag $X$ in $\mathcal{T}$.

Consider a non-leaf bag $X \in \mathcal{T}$, and let $X_1,X_2,\dots,X_k$ be its child bags. By the definition of $\partial \mathbb{S}^+_X$, $\partial \mathbb{S}^+_X \subseteq \bigcup_{i=1}^k \partial \mathbb{S}^*_{X_i}$. Thus, for each child bag $X_i$, the distances between pairs of vertices in $\partial \mathbb{S}^+_X$ whose shortest paths are contained in $G[X_i]$ can be obtained directly from the information already computed for $X_i$. It remains to handle the case where a shortest path between two vertices of $\partial \mathbb{S}^+_X$ is not contained in a single child bag of $X$. By Property~(7) of \Cref{lem:computation_of_BDD}, any path in $G[X]$ that goes from one child bag of $X$ to another must intersect the balanced cycle separator $S_X$. Hence, if a shortest path from $u \in \partial \mathbb{S}^+_X$ to $v \in \partial \mathbb{S}^+_X$ passes through multiple child bags of $X$, then it must pass through some vertex of $S_X$. In this case, the distance can be computed via $S_X$ easily since $S_X \subseteq \bigcup_{i=1}^k \partial \mathbb{S}^*_{X_i}$.

The high-level idea is now clear. The main difficulty lies in compressing the pairwise distance information among the vertices in $\partial \mathbb{S}^*_X$ so that it is small enough to be efficiently broadcast to a central vertex in the parent bag at each level of $\mathcal{T}$. By Property~(5) of \Cref{lem:computation_of_BDD}, every leaf bag has small size, and fast communication is possible within a leaf bag. For each leaf bag, we can directly compute an approximate distance oracle for the planar graph induced by that bag.

However, this small-size property does not necessarily hold for a non-leaf bag $X \in \mathcal{T}$. Suppose that, for every child bag $X_i$ of $X$, we have already computed the pairwise distances among the vertices in $\partial \mathbb{S}^*_{X_i}$. Instead of computing an approximate distance oracle directly for $G[X]$, we compute an approximate planar emulator for $G[X]$, with the vertices in $\bigcup_i \partial \mathbb{S}^*_{X_i}$ serving as terminals. We then compute an approximate distance oracle for this planar emulator and broadcast the compressed oracle information to the next parent bag and so on until $X$ is the root. To apply \Cref{thm:os-embedding}, we cut open the hub network $\mathcal{P}_{X_i}$ and apply \Cref{lem:cut_hub_network_merges_holes}, the holes associated with $\partial \mathbb{S}^*_{X_i}$ are merged into a single hole. Hence, the pairwise distances among these terminals form an Okamura--Seymour quasimetric, and \Cref{thm:os-embedding} can be applied to compute an approximate planar emulator for each child bag separately.

Formally, we give the detailed description of the algorithm.

To begin with, for each bag $X \in \mathcal{T}$, let \EMPH{$O_X$} denote the plane graph obtained from $G[X]$ by cutting open all paths in $\mathcal{P}_X$ according to the order specified in \Cref{lem:cut_hub_network_merges_holes}. Let \EMPH{$\mathcal{J}_X$} denote all the vertices and copies of $\mathcal{P}_X \cup \partial \mathbb{S}^*_X$ in $O_X$ after cutting open each path in $\mathcal{P}_X$. Note that $O_X$ is an Okamura-Seymour graph with vertices $\mathcal{J}_X$ lying on its one-face boundary. By \Cref{clm:size_of_p} and \Cref{lem:logn_holes}, $\mathcal{J}_X$ is bounded by $\OO(D)$.

Furthermore, let $d$ be the height of $\mathcal{T}$. We choose $\epsilon'>0$ such that $(1+\epsilon')^{d}=1+\epsilon$. Since $d=O(\log n)$ by Property~(1) of \Cref{lem:computation_of_BDD}, we have $\epsilon'=O(\epsilon/\log n)$.

First, we prove the following lemma by induction.

\begin{lemma}\label{lem:broadcasting_distance_oracle}
    For each bag $X \in \mathcal{T}$, we can compute a distance oracle associated with $X$ such that querying this oracle returns a $(1+\epsilon)$-approximation of the distance in $O_X$ between any two vertices in $\mathcal{J}_X$. Moreover, the size of the oracle information associated with the vertices in $\mathcal{J}_X$ is bounded by $\OO(D/\epsilon)$. The total communication can be completed in $\OO(D/\epsilon)$ rounds.
\end{lemma}

\noindent \textbf{Base Case:}

\begin{enumerate}

\item \textbf{Broadcasting.} By Properties~(5) and~(6) of \Cref{lem:computation_of_BDD}, for every leaf bag $L \in \mathcal{T}$, all edge information inside $G[L]$ can be broadcast to a central vertex $c_L$ in $L$ in $\OO(D)$ rounds. This can be done for all leaf bags in parallel. In particular, $c_L$ learns the direction, weight, and embedding information of every edge in $G[L]$.

\item \textbf{Cutting open the hub network.} The following computation is performed locally at $c_L$. Following \Cref{lem:cut_hub_network_merges_holes}, we cut open every path in $\mathcal{P}_L$. As a result, all holes in $\mathcal{H}^*_L$ are merged into a single hole, whose boundary contains all vertices in $\mathcal{J}_L$. The resulting Okamura--Seymour instance is $O_L$ by definition. 

\item \textbf{Computing the distance oracle.} Applying the algorithm in \Cref{lem:distance_oracle} to $O_L$ yields a distance oracle whose queries return $(1+\epsilon')$-approximate distances in $O_L$ between any two vertices in $\mathcal{J}_L$. This computation takes polynomial time. Since $|\mathcal{J}_L|=\OO(D)$, the oracle information associated with the vertices in $\mathcal{J}_L$ has size $\OO(D/\epsilon)$.
\end{enumerate}

\noindent \textbf{Inductive Step:}
Suppose that every leaf bag in $\mathcal{T}$ has height $0$. Let $X \in \mathcal{T}$ be a non-leaf bag of height $i$. By the induction hypothesis, for each child bag $X_i$ of $X$, we have already computed a distance oracle associated with $X_i$ such that querying this oracle returns a $(1+\epsilon')^i$-approximation of the distance in $O_{X_i}$ between any two vertices in $\mathcal{J}_{X_i}$. Moreover, the size of the oracle information associated with the vertices in $\mathcal{J}_{X_i}$ is bounded by $\OO(D/\epsilon)$.

\begin{enumerate}
    \item \textbf{Broadcasting.} For each child bag $X_i$ of $X$, we broadcast the circular ordering \EMPH{$\sigma_{X_i}$} of the vertices in $\mathcal{J}_{X_i}$, together with the $(1+\epsilon')^i$-approximate oracle information associated with $\mathcal{J}_{X_i}$, from the central vertex $c_{X_i}$ to a central vertex $c_X$ in $X$. Since the size of this information is $\OO(D/\epsilon)$, this step can be completed in $\OO(D/\epsilon)$ rounds.

    \item \textbf{Realizing the Okamura--Seymour quasimetric.} At $c_X$, for each child bag $X_i$, we use the distance oracle associated with $X_i$ to obtain a $(1+\epsilon')^i$-distorted Okamura--Seymour quasimetric on the terminal set $\mathcal{J}_{X_i}$, with respect to the distances in $O_{X_i}$. By \Cref{thm:os-embedding}, we can compute in polynomial time an Okamura--Seymour graph \EMPH{$K_{X_i}$} with vertices in $\mathcal{J}_{X_i}$ lying on the one-face boundary that follow the circular ordering \EMPH{$\sigma_{X_i}$}. Moreover, the pairwise distances among $\mathcal{J}_{X_i}$ in $K_{X_i}$ have distortion at most $(1+\epsilon')^i$ compared with the corresponding distances in $O_{X_i}$. Since $|\mathcal{J}_{X_i}|=\OO(D)$, the size of $K_{X_i}$ is polynomial in $D$.

    \item \textbf{Obtaining an approximate planar emulator.} For each child bag $X_i$, starting from $K_{X_i}$, we locally glue back every path in $\mathcal{P}_{X_i}$ by identifying the copies of the same vertex in $\mathcal{P}_{X_i}$. This gives a planar graph \EMPH{$H'_{X_i}$}. By construction, $H'_{X_i}$ is a $(1+\epsilon')^i$-approximate planar emulator for $G[X_i]$, with the vertices in $\mathcal{P}_{X_i}\cup \partial \mathbb{S}^*_{X_i}$ serving as terminals. Moreover, the planar embedding of these terminals in $H'_{X_i}$ is \EMPH{aligned} with their embedding in $G[X_i]$.

    Hence, we can locally glue together the emulators $H'_{X_i}$ over all child bags $X_i$ of $X$ by identifying the copies of the same vertex in $S_X$ since the child bags only overlap along the separator vertices by Property~(7) of \Cref{lem:computation_of_BDD}. This produces a planar graph \EMPH{$H_X$}, which $H_X$ is a $(1+\epsilon')^i$-approximate planar emulator for $G[X]$. By the construction of the hub network in \Cref{def:hub_network}, the terminal set of $H_X$ is $\bigcup_i \left(\mathcal{P}_{X_i}\cup \partial \mathbb{S}^*_{X_i}\right) = \mathcal{P}_X\cup \partial \mathbb{S}^+_X$. Moreover, since each $K_{X_i}$ has size polynomial in $D$, the size of $H_X$ is also polynomial in $D$. (Fig. \ref{fig:SSSP-cutting-gluing} illustrates this procedure)

    \item \textbf{Cutting open the hub network.} Since $H_X$ is an aligned emulator for $G[X]$ and $\mathcal{P}_{X} \cup \partial \mathbb{S}^*_X \subseteq\mathcal{P}_{X} \cup \partial \mathbb{S}^+_X$ are the terminals, the paths in $\mathcal{P}_X$ can be identified in $H_X$. We locally cut open all paths in $\mathcal{P}_X$ in $H_X$ according to \Cref{lem:cut_hub_network_merges_holes}. Then the vertices in $\mathcal{J}_X$ lie on the boundary of a single hole. We denote the resulting planar graph by \EMPH{$O'_X$}. Since $H_X$ is a $(1+\epsilon')^i$-approximate planar emulator for $G[X]$, the pairwise distances among vertices in $\mathcal{J}_X$ in $O'_X$ have distortion at most $(1+\epsilon')^i$ compared with the corresponding distances in $O_X$.

    \item \textbf{Computing the distance oracle.}
    We locally compute a $(1+\epsilon')$-approximate distance oracle for $O'_X$ using the algorithm in \Cref{lem:distance_oracle}. Since $O'_X$ has size polynomial in $D$ and $|\mathcal{J}_X|=\OO(D)$, the distance oracle information associated with the vertices in $\mathcal{J}_X$ has size $\OO(D/\epsilon)$. Combining this additional $(1+\epsilon')$ distortion with the previous $(1+\epsilon')^i$ distortion, querying the resulting distance oracle gives a $(1+\epsilon')^{i+1}$-approximation of the pairwise distances among vertices in $\mathcal{J}_X$ in $O_X$. This completes the induction. 
\end{enumerate}

By \Cref{lem:broadcasting_distance_oracle}, for each non-leaf bag $X \in \mathcal{T}$, we have computed a $(1+\epsilon)$-approximate planar emulator $H_X$ for $G[X]$, with the vertices in $\mathcal{P}_X\cup \partial \mathbb{S}^+_X$ serving as terminals. Hence, we can obtain approximate pairwise distances among the vertices in $\partial \mathbb{S}^+_X$ within $G[X]$. 

It remains to compute approximate distances between each vertex $v \in X$ and every vertex in $\partial \mathbb{S}^+_X$ within $G[X]$. More precisely, in the directed setting, we need to compute both the distances from $v$ to the vertices in $\partial \mathbb{S}^+_X$ and the distances from these vertices to $v$. We show how to compute this information in the following lemma. This directly completes the computation of the distance labels in \Cref{def:distance_label}.

\begin{lemma}
For every bag $X \in \mathcal{T}$ and every vertex $v \in X$, vertex $v$ can locally compute $(1+\epsilon)$-approximate distances from $v$ to every vertex in $\partial \mathbb{S}^+_X$, as well as $(1+\epsilon)$-approximate distances from every vertex in $\partial \mathbb{S}^+_X$ to $v$.
\end{lemma}

\begin{proof}

We prove this lemma by induction in a bottom-up manner on $\mathcal{T}$. 

\noindent \textbf{Base Case:}

\begin{enumerate}

\item \textbf{Broadcasting.} By Properties~(5) and~(6) of \Cref{lem:computation_of_BDD}, for every leaf bag $L \in \mathcal{T}$, all edge information inside $G[L]$ can be broadcast to every vertex in $L$ in $\OO(D)$ rounds. This can be done for all leaf bags in parallel.

\item \textbf{Computing the approximate distances.} After the broadcasting step, every vertex $v \in L$ knows the entire graph $G[L]$. Hence, $v$ can locally compute $(1+\epsilon')$-approximate distances from $v$ to every vertex in $\partial \mathbb{S}^+_L$, and also $(1+\epsilon')$-approximate distances from every vertex in $\partial \mathbb{S}^+_L$ to $v$.

\end{enumerate}

\noindent \textbf{Inductive Step:}
Let $X \in \mathcal{T}$ be a non-leaf bag of height $i>0$, and let $X_i$ be an arbitrary child bag of $X$. By the induction hypothesis, every vertex $v \in X_i$ has already locally stored $(1+\epsilon')^i$-approximate distances from $v$ to every vertex in $\partial \mathbb{S}^+_{X_i}$, as well as $(1+\epsilon')^i$-approximate distances from every vertex in $\partial \mathbb{S}^+_{X_i}$ to $v$.

\begin{enumerate}

    \item \textbf{Broadcasting.} Recall that the central vertex $v_X$ has locally computed a $(1+\epsilon')^i$-approximate planar emulator $H_X$, with $\mathcal{P}_X\cup \partial \mathbb{S}^+_X$ serving as terminals. We then compute a $(1+\epsilon')$-approximate distance oracle for $H_X$ using \Cref{lem:distance_oracle}, and broadcast only the oracle information associated with the vertices in $\partial \mathbb{S}^+_X$ to every vertex in $X$. Since this information has size $\OO(D/\epsilon)$, the broadcasting can be completed in $\OO(D/\epsilon)$ rounds. Combining the distortions, every vertex in $X$ obtains $(1+\epsilon')^{i+1}$-approximate pairwise distances among the vertices in $\partial \mathbb{S}^+_X$ within $G[X]$.

    \item \textbf{Computing approximate distances.} Now, for every vertex $v \in X$, we compute $(1+\epsilon')^{i+1}$-approximate distances from $v$ to every vertex $u \in \partial \mathbb{S}^+_X$, as well as from every vertex $u \in \partial \mathbb{S}^+_X$ to $v$.

    \begin{enumerate}
        \item \textbf{When $v \in \partial \mathbb{S}^+_X$.} In this case, both $v$ and $u$ belong to $\partial \mathbb{S}^+_X$. Hence, the $(1+\epsilon')^{i+1}$-approximate distances between $v$ and $u$ in both directions can be obtained directly from the distance oracle broadcast in the previous step.

        \item \textbf{When $v \notin \partial \mathbb{S}^+_X$.} Since $v \notin S_X$, by Property~(7) of \Cref{lem:computation_of_BDD}, the vertex $v$ belongs to exactly one child bag $X_i$ of $X$.

        If $u \in X_i$ and the shortest path between $u$ and $v$ in $G[X]$ is contained in $G[X_i]$, then the desired distances have already been computed in the child bag $X_i$ by the induction hypothesis.

        Otherwise, the shortest path between $u$ and $v$ in $G[X]$ goes outside $G[X_i]$. By Property~(7) of \Cref{lem:computation_of_BDD}, such a path must intersect some vertex $t \in S_X$. Thus, we can decompose the distance through a separator vertex $t$. Without the loss of generality, the shortest path between v and u can be divided into a subpath between $v$ and $t$ within $G[X_i]$, and a subpath between $t$ and $u$ within $G[X]$.The first distance is available from the information computed in the child bag $X_i$, since $t \in S_X \subseteq \partial \mathbb{S}^+_{X_i}$; the second distance is available from the distance oracle for $\partial \mathbb{S}^+_X$, since both $t$ and $u$ belong to $\partial \mathbb{S}^+_X$. Combining these estimates gives $(1+\epsilon')^{i+1}$-approximate distances between $v$ and $u$ in both directions.
    \end{enumerate}
\end{enumerate}
This completes the induction.

\end{proof}

\paragraph{AI Disclosure.~} All results except the NP-hardness proof (\Cref{thm:embedding-hardness}) are obtained by the authors without using any AI tool. The writing and presentation are done entirely by the authors: AI is only used in finding typos and minor mistakes. For the NP-hardness proof (\Cref{thm:embedding-hardness}) in \Cref{subsec:hardness},  the main idea, including the source NP-hard problem as well as the reduction, was discovered by ChatGPT Pro 5.5. The authors' main contributions to this specific result are: (1) refuting several wrong proofs by ChatGPT; (2) guiding the proof discovery; (3) verifying and rewriting the correct proof in our own way; and (4) finding a valid range for the distortion parameter $D$ in \Cref{thm:embedding-hardness}. The authors are solely responsible for all scientific content and claims presented in this paper.

\paragraph{Acknowledgments.} The first author would like to thank Merav Parter for discussions on distributed SSSP in planar graphs, which ultimately led to this work.  The last author would like to thank Zihan Tan for patiently answering many questions about~\cite{ChenTan2025}. This work is supported by NSF grant CCF-2517033 and NSF CAREER Award CCF-2237288.

\bibliographystyle{plain}
\bibliography{ref}

\begin{thebibliography}{10}

\bibitem{abboud2018near}
Amir Abboud, Pawel Gawrychowski, Shay Mozes, and Oren Weimann.
\newblock Near-optimal compression for the planar graph metric.
\newblock In {\em Proceedings of the Twenty-Ninth Annual ACM-SIAM Symposium on Discrete Algorithms}, pages 530--549. SIAM, 2018.

\bibitem{ADPW25}
Yaseen Abd-Elhaleem, Michal Dory, Merav Parter, and Oren Weimann.
\newblock Distributed maximum flow in planar graphs.
\newblock In {\em Proceedings of the ACM Symposium on Principles of Distributed Computing}, pages 278--286, 2025.

\bibitem{ADW26}
Yaseen Abd-Elhaleem, Michal Dory, and Oren Weimann.
\newblock A simple distributed deterministic planar separator, 2026.

\bibitem{ABFPT98}
Richa Agarwala, Vineet Bafna, Martin Farach, Mike Paterson, and Mikkel Thorup.
\newblock On the approximability of numerical taxonomy (fitting distances by tree metrics).
\newblock {\em SIAM Journal on Computing}, 28(3):1073--1085, 1998.

\bibitem{Bandelt1990}
Hans-J\"{u}rgen Bandelt.
\newblock Recognition of tree metrics.
\newblock {\em SIAM Journal on Discrete Mathematics}, 3(1):1–6, February 1990.

\bibitem{Batagelj1990}
Vladimir Batagelj, Tomaž Pisanski, and J.~M.~S. Sim\~{o}es Pereira.
\newblock An algorithm for tree-realizability of distance matrices.
\newblock {\em International Journal of Computer Mathematics}, 34(3–4):171–176, 1990.

\bibitem{badoiu2005}
Mihai Bundefineddoiu, Kedar Dhamdhere, Anupam Gupta, Yuri Rabinovich, Harald R\"{a}cke, R.~Ravi, and Anastasios Sidiropoulos.
\newblock Approximation algorithms for low-distortion embeddings into low-dimensional spaces.
\newblock In {\em Proceedings of the 16th Annual ACM-SIAM Symposium on Discrete Algorithms}, SODA '05, page 119–128, 2005.

\bibitem{buneman1971}
Peter Buneman.
\newblock The recovery of trees from measures of dissimilarity.
\newblock In F.~R. Hodson, D.~G. Kendall, and P.~Tautu, editors, {\em Mathematics in the Archaeological and Historical Sciences}, pages 387--395. Edinburgh University Press, Edinburgh, 1971.

\bibitem{badoiuSTOC05}
Mihai Bǎdoiu, Julia Chuzhoy, Piotr Indyk, and Anastasios Sidiropoulos.
\newblock Low-distortion embeddings of general metrics into the line.
\newblock In {\em Proceedings of the 37th annual ACM symposium on Theory of computing}, STOC '05, pages 225--233, 2005.

\bibitem{CCCLPP25}
Hsien-Chih Chang, Vincent Cohen-Addad, Jonathan Conroy, Hung Le, Marcin Pilipczuk, and Micha{\l} Pilipczuk.
\newblock Embedding planar graphs into graphs of treewidth ${O}(\log^3 n)$.
\newblock In {\em Proceedings of the 2025 Annual ACM-SIAM Symposium on Discrete Algorithms (SODA)}, pages 88--123. SIAM, 2025.

\bibitem{CC25}
Hsien-Chih Chang and Jonathan Conroy.
\newblock Distance approximating minors for planar and minor-free graphs.
\newblock In {\em Proceedings of the 66th Annual Symposium on Foundations of Computer Science}, FOCS '25, page 734–754. IEEE, 2025.

\bibitem{CCLMST23a}
Hsien-Chih Chang, Jonathan Conroy, Hung Le, Lazar Milenkovic, Shay Solomon, and Cuong Than.
\newblock Covering planar metrics (and beyond): O(1) trees suffice.
\newblock In {\em 2023 IEEE 64th Annual Symposium on Foundations of Computer Science (FOCS)}, pages 2231--2261, 2023.

\bibitem{CCLMST24}
Hsien-Chih Chang, Jonathan Conroy, Hung Le, Lazar Milenković, Shay Solomon, and Cuong Than.
\newblock {\em Shortcut Partitions in Minor-Free Graphs: Steiner Point Removal, Distance Oracles, Tree Covers, and More}, page 5300–5331.
\newblock SODA '24. 2024.

\bibitem{CKT22}
Hsien-Chih Chang, Robert Krauthgamer, and Zihan Tan.
\newblock Almost-linear $\varepsilon$-emulators for planar graphs.
\newblock In {\em Proceedings of the 54th Annual ACM SIGACT Symposium on Theory of Computing}, pages 1311--1324, 2022.

\bibitem{CO20}
Hsien{-}Chih Chang and Tim Ophelders.
\newblock Planar emulators for monge matrices.
\newblock In {\em Proceedings of the 32nd Canadian Conference on Computational Geometry}, CCCG' 20, pages 141--147, 2020.

\bibitem{CC08}
Manoj Changat and Victor Chepoi.
\newblock Discrete metric spaces and related graph classes.
\newblock In Manoj Changat, Sandi Klav\v{z}ar, Henry~Martyn Mulder, and A.~Vijayakumar, editors, {\em Convexity in Discrete Structures: Joint Proceedings of the International Instructional Workshop on Convexity in Discrete Structures, Thiruvananthapuram, India, 2006 and the International Workshop on Metric and Convex Graph Theory, Barcelona, Spain, 2006}, volume~5 of {\em Ramanujan Mathematical Society Lecture Notes Series}, pages 25--46. Ramanujan Mathematical Society, 2008.
\newblock Link to article \href{https://www.researchgate.net/profile/Manoj-Changat/publication/266940870_Convexity_in_discrete_structures_Joint_proceedings_of_the_international_instructional_workshop_on_convexity_in_discrete_structures_Thiruvananthapuram_Kerala_India_March_22-April_2_2006_and_the_interna/links/55102ca00cf27d62b913bed4/Convexity-in-discrete-structures-Joint-proceedings-of-the-international-instructional-workshop-on-convexity-in-discrete-structures-Thiruvananthapuram-Kerala-India-March-22-April-2-2006-and-the-in.pdf}{here}.

\bibitem{ChenTan2025}
Yu~Chen and Zihan Tan.
\newblock Paths and intersections: Characterization of quasi-metrics in directed okamura-seymour instances.
\newblock In {\em Proceedings of the 36th Annual ACM-SIAM Symposium on Discrete Algorithms (SODA 2025)}, 2025.

\bibitem{CS16}
Victor Chepoi and Morgan Seston.
\newblock An approximation algorithm for $\ell_{\infty}$ fitting robinson structures to distances.
\newblock volume~3, pages 265--276, 2009.

\bibitem{CLPP23}
Vincent Cohen{-}Addad, Hung Le, Marcin Pilipczuk, and Michał Pilipczuk.
\newblock Planar and minor-free metrics embed into metrics of polylogarithmic treewidth with expected multiplicative distortion arbitrarily close to 1.
\newblock In {\em 2023 IEEE 64th Annual Symposium on Foundations of Computer Science}, FOCS ’23, pages 2262--2277, 2023.

\bibitem{filtser2020face}
Arnold Filtser.
\newblock A face cover perspective to $\ell_1$ embeddings of planar graphs.
\newblock In {\em Proceedings of the Fourteenth Annual ACM-SIAM Symposium on Discrete Algorithms}, pages 1945--1954. SIAM, 2020.

\bibitem{FL22}
Arnold Filtser and Hung Le.
\newblock Low treewidth embeddings of planar and minor-free metrics.
\newblock In {\em Proceedings of the 63rd Annual Symposium on Foundations of Computer Science, {FOCS}' 22}, 2022.

\bibitem{FKS19}
E.~Fox{-}Epstein, P.~N. Klein, and A.~Schild.
\newblock Embedding planar graphs into low-treewidth graphs with applications to efficient approximation schemes for metric problems.
\newblock In {\em Proceedings of the 30th Annual ACM-SIAM Symposium on Discrete Algorithms}, SODA `19, page 1069–1088, 2019.

\bibitem{GH16}
Mohsen Ghaffari and Bernhard Haeupler.
\newblock Distributed algorithms for planar networks {I}: Planar embedding.
\newblock In {\em Proceedings of the 2016 ACM Symposium on Principles of Distributed Computing}, pages 29--38, 2016.

\bibitem{GH16b}
Mohsen Ghaffari and Bernhard Haeupler.
\newblock Distributed algorithms for planar networks {II}: {L}ow-congestion shortcuts, {MST}, and min-cut.
\newblock In {\em Proceedings of the twenty-seventh annual ACM-SIAM symposium on Discrete algorithms}, pages 202--219. SIAM, 2016.

\bibitem{GP17}
Mohsen Ghaffari and Merav Parter.
\newblock Near-optimal distributed dfs in planar graphs.
\newblock In {\em 31st International Symposium on Distributed Computing (DISC 2017)}, pages 21--1. Schloss Dagstuhl--Leibniz-Zentrum f{\"u}r Informatik, 2017.

\bibitem{GZ22}
Mohsen Ghaffari and Goran Zuzic.
\newblock Universally-optimal distributed exact min-cut.
\newblock In {\em Proceedings of the 2022 ACM Symposium on Principles of Distributed Computing}, pages 281--291, 2022.

\bibitem{GolumbicKaplanShamir1994}
Martin~Charles Golumbic, Haim Kaplan, and Ron Shamir.
\newblock Algorithms and complexity of sandwich problems in graphs.
\newblock In {\em Graph-Theoretic Concepts in Computer Science}, volume 790 of {\em Lecture Notes in Computer Science}, pages 57--69. Springer, 1994.

\bibitem{GolumbicKaplanShamir1995}
Martin~Charles Golumbic, Haim Kaplan, and Ron Shamir.
\newblock Graph sandwich problems.
\newblock {\em Journal of Algorithms}, 19(3):449--473, 1995.

\bibitem{Gordon1987}
A.~D. Gordon.
\newblock A review of hierarchical classification.
\newblock {\em Journal of the Royal Statistical Society. Series A (General)}, 150(2):119, 1987.

\bibitem{GLS81}
M.~Gr\"{o}tschel, L.~Lovász, and A.~Schrijver.
\newblock The ellipsoid method and its consequences in combinatorial optimization.
\newblock {\em Combinatorica}, 1(2):169–197, June 1981.

\bibitem{gnrs99}
Anupam Gupta, Alistair Sinclair, Ilan Newman, and Yuri Rabinovich.
\newblock Cuts, trees and $\ell_1$-embeddings of graphs.
\newblock In {\em Proceedings of the 40th Annual Symposium on Foundations of Computer Science}, FOCS '99, page 399, USA, 1999. IEEE Computer Society.

\bibitem{HHMM20}
Momoko Hayamizu, Katharina~T. Huber, Vincent Moulton, and Yukihiro Murakami.
\newblock Recognizing and realizing cactus metrics.
\newblock {\em Information Processing Letters}, 157:105916, May 2020.

\bibitem{HT74}
John Hopcroft and Robert Tarjan.
\newblock Efficient planarity testing.
\newblock {\em Journal of the ACM (JACM)}, 21(4):549--568, 1974.

\bibitem{hurkens1988tidy}
C.~A.~J. Hurkens, L{\'a}szl{\'o} Lov{\'a}sz, Alexander Schrijver, and {\'E}va Tardos.
\newblock How to tidy up your set-system?
\newblock In A.~Hajnal, L.~Lov{\'a}sz, and V.~T. S{\'o}s, editors, {\em Combinatorics}, volume~52 of {\em Colloquia Mathematica Societatis J{\'a}nos Bolyai}, pages 309--314. North-Holland, Amsterdam, 1988.
\newblock Proceedings of the 7th Hungarian Colloquium, Eger, Hungary, 1987.

\bibitem{JMR25}
Benjamin Jauregui, Pedro Montealegre, and Ivan Rapaport.
\newblock Deterministic distributed dfs via cycle separators in planar graphs.
\newblock In {\em Proceedings of the ACM Symposium on Principles of Distributed Computing}, pages 268--277, 2025.

\bibitem{karczmarz2025}
Adam Karczmarz and Da~Wei Zheng.
\newblock Subquadratic algorithms in minor-free digraphs:(weighted) distance oracles, decremental reachability, and more.
\newblock In {\em Proceedings of the 2025 Annual ACM-SIAM Symposium on Discrete Algorithms (SODA)}, pages 4338--4351. SIAM, 2025.

\bibitem{KS21}
Ken{-}ichi Kawarabayashi and Anastasios Sidiropoulos.
\newblock Embeddings of planar quasimetrics into directed $\ell_1$ and polylogarithmic approximation for directed sparsest-cut.
\newblock In {\em Proceedings of the 62nd {IEEE} Annual Symposium on Foundations of Computer Science}, FOCS' 21, pages 480--491, 2021.

\bibitem{KPR93}
Philip Klein, Serge~A. Plotkin, and Satish Rao.
\newblock Excluded minors, network decomposition, and multicommodity flow.
\newblock In {\em Proceedings of the Twenty-Fifth Annual ACM Symposium on Theory of Computing}, STOC '93, pages 682--690, 1993.

\bibitem{krauthgamer2019flow}
Robert Krauthgamer, James~R Lee, and Havana Rika.
\newblock Flow-cut gaps and face covers in planar graphs.
\newblock In {\em Proceedings of the Thirtieth Annual ACM-SIAM Symposium on Discrete Algorithms}, pages 525--534. SIAM, 2019.

\bibitem{krauthgamer2013mimicking}
Robert Krauthgamer and Inbal Rika.
\newblock Mimicking networks and succinct representations of terminal cuts.
\newblock In {\em Proceedings of the twenty-fourth annual ACM-SIAM symposium on Discrete algorithms}, pages 1789--1799. SIAM, 2013.

\bibitem{KR08}
Robert Krauthgamer and Tim Roughgarden.
\newblock Metric clustering via consistent labeling.
\newblock In {\em Proceedings of the 19th Annual ACM-SIAM Symposium on Discrete Algorithms}, SODA '08, page 809–818, 2008.

\bibitem{kumar2025approximate}
Nikhil Kumar.
\newblock An approximate generalization of the okamura--seymour theorem.
\newblock {\em SIAM Journal on Computing}, 54(5):FOCS22--159, 2025.

\bibitem{le2023vc}
Hung Le and Christian Wulff{-}Nilsen.
\newblock {VC} set systems in minor-free (di)graphs and applications.
\newblock In {\em Proceedings of the 2024 {ACM-SIAM} Symposium on Discrete Algorithms}, SODA' 24, pages 5332--5360, 2024.

\bibitem{LTZ25}
George~Z. Li, Zihan Tan, and Tianyi Zhang.
\newblock Paths and intersections: Exact emulators for planar graphs.
\newblock In {\em Proceedings of the 66th Annual Symposium on Foundations of Computer Science}, FOCS '25, pages 1903--1926, 2025.

\bibitem{li2019planar}
Jason Li and Merav Parter.
\newblock Planar diameter via metric compression.
\newblock In {\em Proceedings of the 51st Annual ACM SIGACT Symposium on Theory of Computing}, pages 152--163, 2019.

\bibitem{LO93}
Peter~J. Looges and Stephan Olariu.
\newblock Optimal greedy algorithms for indifference graphs.
\newblock {\em Computers \& Mathematics with Applications}, 25(7):15–25, 1993.

\bibitem{matouvsek2002open}
Jiri Matou{\v{s}}ek and Assaf Naor.
\newblock Open problems on embeddings of finite metric spaces.
\newblock In {\em Workshop on discrete metric spaces and their algorithmic applications}, 2002.

\bibitem{memoli2018quasimetric}
Facundo M{\'e}moli, Anastasios Sidiropoulos, and Vijay Sridhar.
\newblock Quasimetric embeddings and their applications.
\newblock {\em Algorithmica}, 80(12):3803--3824, 2018.

\bibitem{OS81}
Haruko Okamura and Paul~D Seymour.
\newblock Multicommodity flows in planar graphs.
\newblock {\em Journal of Combinatorial Theory, Series B}, 31(1):75--81, 1981.

\bibitem{Parter20}
Merav Parter.
\newblock Distributed planar reachability in nearly optimal time.
\newblock In {\em 34th International Symposium on Distributed Computing (DISC 2020)}, pages 38--1. Schloss Dagstuhl--Leibniz-Zentrum f{\"u}r Informatik, 2020.

\bibitem{Roberts1969}
Fred~S. Roberts.
\newblock Indifference graphs.
\newblock In Frank Harary, editor, {\em Proof Techniques in Graph Theory}, pages 139--146. Academic Press, 1969.

\bibitem{RGHZL22}
V{\'a}clav Rozho{\v{n}}, Christoph Grunau, Bernhard Haeupler, Goran Zuzic, and Jason Li.
\newblock Undirected $(1+\epsilon)$-shortest paths via minor-aggregates: near-optimal deterministic parallel and distributed algorithms.
\newblock In {\em Proceedings of the 54th Annual ACM SIGACT Symposium on Theory of Computing}, pages 478--487, 2022.

\bibitem{Thorup04}
Mikkel Thorup.
\newblock Compact oracles for reachability and approximate distances in planar digraphs.
\newblock {\em Journal of the ACM (JACM)}, 51(6):993--1024, 2004.

\bibitem{Zaretskii1965}
K.~A. Zaretsky.
\newblock Constructing trees from the set of distances between pendant vertices.
\newblock {\em Uspekhi Matematicheskikh Nauk}, 20:90--92, 1965.
\newblock In Russian.

\bibitem{ZGYHS22}
Goran Zuzic, Gramoz Goranci, Mingquan Ye, Bernhard Haeupler, and Xiaorui Sun.
\newblock Universally-optimal distributed shortest paths and transshipment via graph-based $\ell_1$-oblivious routing.
\newblock In {\em Proceedings of the 2022 Annual ACM-SIAM Symposium on Discrete Algorithms (SODA)}, pages 2549--2579. SIAM, 2022.

\end{thebibliography}
\appendix

\section{Appendix}

\subsection{Hardness of Bounded-Distortion OS Embedding}
\label{subsec:hardness}
In this section, we prove \Cref{thm:embedding-hardness}, that deciding whether a finite metric admits a bounded-distortion embedding into an Okamura--Seymour metric is NP-hard. We note that in this section, we deal only with metrics instead of quasi-metrics. Since any metric is a quasimetric, our NP-hardness for metric implies NP-hardness for quasi-metrics. 

\OSEmbeddingNPhardness*

We reduce from the unit-interval graph sandwich problem, which was shown to be NP-hard by Golumbic, Kaplan, and Shamir \cite{GolumbicKaplanShamir1994,GolumbicKaplanShamir1995}. Recall that a graph $G = (V, E)$ is a \EMPH{unit-interval graph} if each vertex $v \in V$ can be assigned a unit interval on the real line such that two vertices are adjacent if and only if their intervals overlap.

\begin{problem}[Unit-Interval Graph Sandwich]\label{def:graph_sandwich}
Given a vertex set $V$ and a pair of disjoint edge-sets $ E^+,E^-\subseteq \binom{V}{2}$, does there exist a unit-interval graph $H$ with vertex set $V$ such that $ E^+\subseteq E(H)$  and  $E(H)\cap E^-=\emptyset$?
\end{problem}

\subsection{Preparation}

We will be using the umbrella property of interval graphs first shown by Roberts~\cite{Roberts1969}. Given a graph $H$ with  $n$ vertices, a linear order on $V(H)$, denoted by $\langle v_1,\ldots,v_n \rangle $, is an \EMPH{umbrella ordering} of $H$ if for all \(1\leq i<j<k \leq n\), $v_i v_k\in E(H)$ implies that $v_i v_j\in E(H)$  and $
v_j v_k\in E(H)$. That is, $v_i$ and $v_k$ are adjacent  if and only if they are both adjacent to every $v_j$ between them. Roberts~\cite{Roberts1969} was the first to show the equivalence between unit interval graphs and  umbrella orderings. Looges and Olariu~\cite{LO93} showed that such an umbrella ordering can be computed in linear time. 

\begin{lemma}[Roberts~\cite{Roberts1969}]\label{lem:umbrella}
A graph \(H\) is a unit-interval graph if and only if it has an umbrella ordering. 
\end{lemma}

Another way to represent an interval graph $H$ is to assign each vertex $v$ a point, called a \EMPH{center}, $c_v$ on the real line such that $uv\in E(H)$ if and only  if $|c_u - c_v|\leq 1$. For a non-adjacent pair, it only requires that $|c_u - c_v| > 1$. Here for our reduction, we need (i) the gap $|c_u - c_v| - 1$ is not too small, i.e., $|c_u - c_v| > 1 + 1/\poly(n)$, and (ii) the distance from the first center to the last center on the line is not too big, as formally stated in the following lemma.

\begin{restatable}[Bounded-gap unit-interval representation]{lemma}{boundedgap}
\label{lem:bounded_gap}
Let \(H\) be a unit-interval graph on \(n\) vertices.  Let
$\langle v_1,v_2,\dots,v_n \rangle$ be the ordering of the vertices according to the unit-interval representation of \(H\) on the real line. Let \(R=n+1\). Then \(H\) has a unit-interval representation,
where \(v_i\) corresponds to a center \(c_i\) on the real line, such that:
\begin{align}
c_j-c_i &\le 1
&&\text{if } i<j \text{ and } v_i v_j\in E(H),\label{eq:bg-eq-1}\\
c_j-c_i &\ge 1+\frac{1}{R}
&&\text{if } i<j \text{ and } v_i v_j\notin E(H),\label{eq:bg-eq-2}\\
c_i &\le c_{i+1}
&&\text{for } i=1,\ldots,n-1, \label{eq:bg-eq-3}\\
c_n - c_1 &\le R. \label{eq:bg-eq-4}
\end{align}
\end{restatable}

There is nothing special about $R = n+1$: any polynomial of $n$ at least $n$ would work. The idea of the proof of \Cref{lem:bounded_gap} is rather simple but the actual proof is somewhat tedious and hence is deferred to \Cref{app:proof-bounded_gap}. Note that our reduction does not require computing a representation in \Cref{lem:bounded_gap}; we only require that such a representation exists.

\subsection{The Reduction}

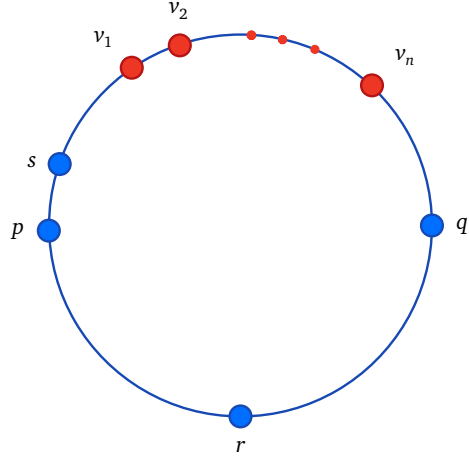
\begin{figure}
\centering
\begin{tikzpicture}[font=\small]
  \node[anchor=south west, inner sep=0pt] (img) at (0,0) {
    \includegraphics[
      page=1,
      trim=116pt 362pt 223pt 174pt,
      clip,
      width=0.33\textwidth
    ]{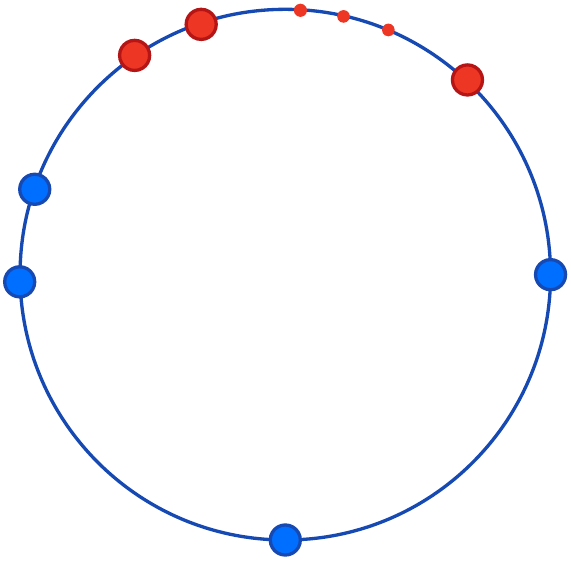}
  };

  \begin{scope}[x={(img.south east)}, y={(img.north west)}]
    \node[anchor=south east] at (0.215,0.930) {$v_1$};
    \node[anchor=south]      at (0.350,1.000) {$v_2$};
    \node[anchor=south west] at (0.850,0.885) {$v_n$};

    \node[anchor=west]  at (1.000,0.510) {$q$};
    \node[anchor=north] at (0.501,0.000) {$r$};
    \node[anchor=east]  at (0.000,0.498) {$p$};
    \node[anchor=east]  at (0.030,0.670) {$s$};
  \end{scope}
\end{tikzpicture}
\caption{The desired order of vertices and anchor points that we force in our construction.}
\label{fig:hardness}
\end{figure}

We now describe the reduction from the unit-interval graph sandwich problem to the  Okamura--Seymour metric embedding problem. For simplicity of the presentation, we fix the distortion $D =\frac75$; any value of $D \in (1, \sqrt2)$ would work.  First, we describe the reduction formally, and then we give the high-level ideas of the reduction.

\paragraph{Reduction:} Let $\mathcal{A}=(V,E^+,E^-)$ be an instance of the unit-interval graph sandwich problem with \(|V|=n\). Let $D = 7/5$ (the distortion) and $R  = n+1$ (the same parameter in \Cref{lem:bounded_gap}). Let  
\begin{equation}\label{eq:eps-UL}
    \epsilon=\frac{D-1}{2R}=\frac{1}{5R},
\qquad
U=1+\epsilon,
\qquad
L=1+\epsilon\left(1+\frac1R\right),
\end{equation}

Intuitively, what we want is that the product of $\eps$ and the maximum interval's length by \Cref{lem:bounded_gap} to be some constant. Moreover, $L$ should be slightly larger than $u$ with respect to the minimum gap by \Cref{lem:bounded_gap}. Next, we construct an instance $\mathcal{B}=(X,d,D = 7/5)$ of the Okamura--Seymour distortion embedding problem. Let $X=V\cup\{p,q,s,r\}$. The additional points $\{p,q,s,r\}$ serve as \EMPH{anchor points} based on which we will extract the linear ordering of $V$ representing the unit-interval graph $H$ from the Okamura--Seymour distortion embedding with distortion $D$.  The distances between pairs of points in $X$ are defined as follows.  First, for every two vertices $u\not= v\in V$ set:
\[
d(u,v)=
\begin{cases}
\frac{U}{D}, & uv\in E^+,\\
L, & uv\in E^-,\\
1, & uv\notin E^+\cup E^-.
\end{cases}
\]
Next, for \(v\in V\), we define its distances to the anchor points as:
\[
d(p,v)=d(q,v)=1,
\qquad
d(s,v)=D,
\qquad
d(r,v)=\frac{11}{10}.
\]
For the anchor points, we define:
\[
d(p,q)=D,
\qquad
d(q,s)=1,
\qquad
d(p,s)=D-1=\frac25,
\qquad
d(p,r)=d(q,r)=\frac D2=\frac{7}{10},
\qquad
d(r,s)=\frac{11}{10}.
\]
This completes the reduction.  Clearly, it is computable in polynomial time.

\paragraph{The high-level ideas:} The hard direction is to show that if there exists an Okamura-Seymour metric $\rho$ within distortion $D$ of $d$, then any cyclic ordering of $X$ satisfying Monge is of the form $\sigma_X = \langle p, s, v_1, \dots, v_n, q, r \rangle$ where $\langle s, v_1, \dots, v_n\rangle$ is the cyclic ordering of vertices restricted to $V \cup \{s\}$ \textbf{and} $v_1, \dots, v_n$ is a valid umbrella ordering characterizing the unit-interval graph $H$ in \Cref{def:graph_sandwich}.   The set of anchor points $\{p,s,q,r\}$ is designed so that if we imagine placing $X$ on the boundary of a disk in the order $\sigma_X$, then the points $p$ and $q$ define a pair of $p$-to-$q$ arcs covering the boundary. The parameters $\{R, \eps, U, L\}$ and distorted metric $d$ are to force every vertex of $V$ to lie on the opposite $p$-to-$q$ arc to the one containing $r$ and for the point $s$ to lie on the same $p$-to-$q$ arc as the vertices of $V$, between $p$ and $V$ (see Fig.~\ref{fig:hardness}). The vertex $p$ is forced by $s$ to be at the same distance from every vertex in
$\{v_1,\ldots,v_n\}$ in any valid $\rho$, so that the Monge property naturally implies the umbrella characterization.

For the other direction (which is easier), $\{p,s,q,r\}$ must satisfy the Monge property in that order. The parameters $\{R, \eps, U, L\}$ and distorted metric $d$ are set so that we may induce a line metric of the linear order on the centers of $V$ in the case of a yes-instance. Line metrics satisfy the Monge property and finally we ensure that after adding the anchor points $\{p,s,q,r\}$ we still have a metric, and it is an Okamura-Seymour metric w.r.t. the order $\sigma_X = \langle p, s, v_1, \dots, v_n, q, r \rangle$. 

By writing down all the distance constraints, we end up with the following system of inequalities: 

\[
M=1+\epsilon R, \;\frac{U}{D} \leq 1<U<L<M \leq D,\; M \leq DL
\]
\[
d(x,v)=d(x,u) \qquad \forall\, x\in\{p,s,q,r\},\ \forall\, v, u\in V,
\]
\[
D\,d(p,v)+D\,d(s,q) = d(p,q)+d(s,v) \qquad \forall\, v\in V,
\]
\[
D\,d(p,s)+D\,d(v,q) < d(p,v)+d(s,q) \qquad \forall\, v\in V,
\]
\[
d(p,q)+d(r,x) > D\,d(p,x)+D\,d(q,r) \qquad \forall\, x\in V\cup\{s\},
\]
\[
d(p,q)+d(r,x) > D\,d(q,x)+D\,d(p,r) \qquad \forall\, x\in V\cup\{s\}.
\]

Moreover, there should exist a metric $\rho$ satisfying
\[
d(a,b)\le \rho(a,b)\le D\,d(a,b)
\qquad \forall\, a,b\in V\cup\{p,s,q,r\},
\]
such that, for any permutation of $V$, the metric $\rho$ always satisfies the Monge property with respect to the cyclic order $p,s,\pi(V),q,r;$ $\rho$ and $d$ always satisfy the triangle inequality as a metric.

\begin{remark} 
If we treat $D$ as a variable, the system of inequalities above gives $1<D<\frac{1+\sqrt{5}}{2}$. We fix a rational $D=\frac75$ in the analysis below so that it is easier to follow the calculation.
\end{remark}

\subsection{From Okamura-Seymour Metric Embedding to Graph Sandwich}

This is the hard direction of the reduction.  Assume there exists an Okamura-Seymour metric
$\rho$ on $X$ such that:
\[
d(x, y)\le \rho(x,y)\le Dd(x,y)
\qquad
\forall x,y\in X,
\]
and $\rho$ satisfies the Monge property with respect to some cyclic order $\sigma_X$ of the elements $X = V \cup \{p,q,s,r\}$. 

Recall that $s, q$ act as the bookends of $V$ in $\sigma_X$. We will show below that our reduction forces $\sigma_X = \langle p, s, v_1, \dots, v_n, q, r \rangle$.
Let \EMPH{$\sigma_V$} be the linear order obtained by fixing an orientation of $\sigma_X$ where we encounter the vertices $V$ after $s$ and before $q$, and restricting to $V$. 
The key observation is that the metric $\rho$ restricted to $V$, denoted by \EMPH{$\rho|_V$}, is \EMPH{Robinson in the ordering $\sigma_V$}~\cite{CS16}:

\begin{definition}[Robinson Metric]\label{def:Robinson}
    A metric $y$ on a finite set of elements $Y$ is \EMPH{Robinson} if there exists an ordering $\sigma_Y$ of $Y$ such that for all $i <_{\sigma_Y} j <_{\sigma_Y} k \in Y$:
    \[y(i, k) \geq \max \{y(i, j), y(j, k)\}.\]
    We say that $Y$ is \EMPH{Robinson in the ordering $\sigma_Y$}.
\end{definition}

\begin{lemma}\label{lm:robinson} $\rho|_V$ is Robinson in the ordering $\sigma_V$. 
\end{lemma}

We delay the proof of \Cref{lm:robinson} to the end of this section. Here we show how to use it to construct a  unit-interval graph $H$.

\paragraph{Constructing $H$.} The vertex set of $H$ is  $V$. For any two vertices $u\not=v \in V$, we add the edge $uv$ to $E(H)$ if and only if $\rho(u,v) \le U$.  We claim that $H$ is a unit-interval graph in the yes instance of \Cref{def:graph_sandwich}.

\begin{claim} \label{lem:rob_implies_yes}
    If $\rho|_V$ is Robinson in the ordering $\sigma_V$, then $H$ is a unit-interval graph certifying that $(V,E^+,E^-)$ is a yes instance of \Cref{def:graph_sandwich}.
\end{claim}

\begin{proof}
    We first show that $H$ is a unit-interval graph by showing that $\sigma_V$ is an  umbrella ordering  of $H$. Let $1\leq i<_{\sigma_V}j<_{\sigma_V}k \leq n$ where $v_i, v_j, v_k \in V$ and $v_i v_k\in E(H)$,  then $\rho(v_i,v_k)\le U$. Since $\rho|_V$ is Robinson in the ordering $\sigma_V$, we have:
    \[
    \rho(v_i,v_j)\le \rho(v_i,v_k)\le U,
     \text{and }
    \rho(v_j,v_k)\le \rho(v_i,v_k)\le U,
    \]
   giving $v_i v_j\in E(H)$ and $v_j v_k\in E(H)$, respectively. Thus, $\sigma_V$ is an  umbrella ordering and hence, $H$ is a unit-interval graph by Lemma~\ref{lem:umbrella}. 
    
    Moreover, if \(uv\in E^+\), by construction \(d(u,v)=U/D\), so $\rho(u,v)\le Dd(u,v)=U.$ Hence, $E^+\subseteq E(H)$. For \(uv\in E^-\), then \(d(u,v)=L>U\), so $\rho(u,v)\ge d(u,v)=L>U$. Hence, $E(H)\cap E^-=\emptyset$. This implies \(H\) is a unit-interval graph certifying that $(V,E^+,E^-)$ is a yes instance of \Cref{def:graph_sandwich}. 
\end{proof}

It remains to prove \Cref{lm:robinson}. First, we prove several simple claims.

\begin{claim}
    Place the elements of $X$ on the boundary of a disk in the plane according to the ordering $\sigma_X$. Then, all vertices in $V \cup \{s\}$ lie on the same \(p\)-\(q\) arc, while \(r\) lies on the other.
\end{claim}

\begin{proof}

For any $x\in V\cup\{s\}$, we have 
\[
\rho(p,q)+\rho(r,x)\ge d(p,q) + d(r,x) = D+ \frac{11}{10}=\frac75+\frac{11}{10}=\frac52.
\]
\[
\rho(p,x)+\rho(q,r)\le D(d(p,x) + d(q,r) ) = D+\frac{D^2}{2}=\frac{119}{50}<\frac52.
\] 
\[
\rho(q,x)+\rho(p,r)\le D(d(q,x) + d(p,r) ) = D+\frac{D^2}{2}=\frac{119}{50}<\frac52.
\]
So we have:
\[
\rho(p,x)+\rho(q,r) < \rho(p,q)+\rho(r,x),
\]
\[
\rho(q,x)+\rho(p,r) < \rho(p,q)+\rho(r,x).
\] 
Hence, among the three perfect matchings on the four points \(p,q,r,x\), the matching $\{pq,rx\}$ has strictly largest total length. By the Monge property, path \(pq\) and path \(rx\) must cross in $\sigma_X$. Equivalently, \(r\) and \(x\) lie on opposite \(p\)-\(q\) arcs. Since this holds for all \(x\in V\cup\{s\}\), all vertices in \(V\cup\{s\}\) lie on the same \(p\)-\(q\) arc, while \(r\) lies on the other.
\end{proof}

Without loss of generality, let the clockwise $p$-to-$q$ arc be the one containing $V\cup\{s\}$.

\begin{claim}\label{clm:pqorder}
    On the clockwise $p$-to-$q$ arc the vertices $V \cup \{p, s, q\}$ appear in the order $p, s, v_1, \dots, v_n, q$.
\end{claim}

\begin{proof}
We prove by contradiction. Suppose otherwise that there exists \(v\in V\) lying before \(s\). Then $\sigma_X$ contains $p,\ v,\ s,\ q$ in order. 
We have:
\[
\rho(p,s)+\rho(v,q) \le D(d(p,s) + d(v,q)) = D(D-1) + D = \frac{49}{25}
\]
Also,
\[
\rho(p,v)+\rho(s,q) \ge d(p,v) + d(s,q) = 1 + 1 = 2. 
\]
This contradicts the Monge property:
\[
\rho(p,s)+\rho(v,q)
\ge
\rho(p,v)+\rho(s,q).
\]
Hence, $\sigma_X$ must have the form $p,\ s,\ v_1,\ldots,v_n,\ q,\ r$ for some ordering \(v_1,\ldots,v_n\) of \(V\).  
\end{proof}

Thus the cyclic order of $X$ on the boundary of the disk is $\sigma_X = \langle p, s, v_1, \dots, v_n, q, r \rangle$. It remains to prove that $\rho|_V$ is Robinson in the induced order $\sigma_V = v_1, \dots, v_n$. 
We need the following guarantee:

\begin{claim}\label{clm:exactD}
    For any $v \in V$, we always have $\rho(p,v) = D$.
\end{claim}

\begin{proof}
For any \(v\in V\), $\sigma_X$ contains $p,\ s,\ v,\ q$ in order. By the Monge property, we have:
\[
\rho(p,v)+\rho(s,q)
\ge
\rho(p,q)+\rho(s,v). 
\]
Since we have:
\[
\begin{aligned}
\rho(p,v)+\rho(s,q) &\le Dd(p, v) + Dd(s, q) = D(1+1) = 2D, \\
\rho(p,q)+\rho(s,v) &\ge d(p, q) + d(s, v)  = D + D = 2D, 
\end{aligned}
\]
Thus, we must have: 
\[
\rho(p,v)+\rho(s,q) = \rho(p,q)+\rho(s,v) = 2D. 
\]
Moreover, for any $v \in V$, $\rho(p, v), \rho(s,q) \leq D$, implying $\rho(p, v) = D$ as promised.
\end{proof}

\begin{proof}[Proof of \Cref{lm:robinson}]
    For any $v_i, v_j, v_k \in V$, such that $i<_{\sigma_V}j<_{\sigma_V}k$, apply the Monge property to the vertices $p,\ v_i,\ v_j,\ v_k$ in that order. It gives
\[
\begin{aligned}
\rho(p,v_j)+\rho(v_i,v_k) &\ge \rho(p,v_i)+\rho(v_j,v_k),\\
\rho(p,v_j)+\rho(v_i,v_k) &\ge \rho(p,v_k)+\rho(v_i,v_j).
\end{aligned}
\]
Since $\rho(p,v_i)=\rho(p,v_j)=\rho(p,v_k)=D$ by~\Cref{clm:exactD}, we obtain
\[
\rho(v_i,v_k)\ge \rho(v_j,v_k),
\qquad
\rho(v_i,v_k)\ge \rho(v_i,v_j).
\]
Thus, $\rho(v_i,v_k)\ge \max\{\rho(v_i,v_j),\rho(v_j,v_k)\}$ as desired.
\end{proof}

\subsection{From Graph Sandwich to Okamura-Seymour Metric Embedding}

This is the easy direction of the reduction, since $\rho|_V$ is induced by a line metric obtained from a valid unit-interval representation, and line metrics naturally satisfy the Monge property with respect to the order of the centers. Also, most of the proof involves straightforward calculation. Given  a yes-instance of the input unit-interval graph sandwich problem, we first construct an embedding metric $\rho$ and then show that: 1) $\rho$ is a valid metric within distortion $D$ of $d$, and 2) $\rho$ is Monge w.r.t some ordering.

\paragraph{The embedding construction.} Let \(H\) be a unit-interval graph on \(V\) such that $E^+\subseteq E(H)$ and $E(H)\cap E^-=\emptyset$. By \Cref{lem:bounded_gap}, there exists a unit-interval representation of $H$ satisfying \Crefrange{eq:bg-eq-1}{eq:bg-eq-4} with $R=n+1$. Order the vertices by the centers of any such bounded-gap unit-interval representation as $v_1,v_2,\dots, v_n$. 

Define:
\[
\rho(u,v) = \textbf{1}[u \neq v] + \epsilon|c_u - c_v| \qquad (u,v \in V),
\]
where $\textbf{1}[u \neq v]$ is an indicator variable.

For the anchor vertices $\{p, s, q, r\}$, set:
\[
\rho(p,v)=\rho(q,v)=\rho(s,v)= \rho(p,q) = \rho(s,q) = D,
\;
\rho(p,s) = D(D-1),\;
\rho(r,v)= \frac{11}{10}.
\]
\[
\rho(s,r)= \rho(r,v) = \frac{11}{10}, \qquad
\rho(p,r)=\rho(q,r)= \frac{D^2}{2}.
\] 

\paragraph{The analysis.} By straightforward calculation (given in \Cref{app:rho-is-metric}), we can show that:

\begin{restatable}{claim}{rhostmetric}\label{rho-is-metric}
    The function $\rho$ is a metric on
    \(X = V \cup \{p,q,s,r\}\).
\end{restatable}

Next, we show that the distortion of $\rho$ is at most $D$,

\begin{claim}\label{clm:good-distortion}
    The metric $\rho$ is within distortion $D$ of $d$.
\end{claim}

\begin{proof}
    We exhaustively verify that $d(x, y)\le \rho(x, y)\le Dd(x, y) \; \forall x, y\in X$. 

    First, we consider pairs involving $V$ only. If $u, v \in V$, and $uv \in E^+$, we have \(d(u,v)=U/D = \frac{1+\epsilon}{D}\) and $1 \leq \rho(u,v) = 1 + \epsilon|c_u - c_v| \leq 1 + \epsilon = Dd(u,v)$ since $|c_u - c_v| \leq 1$ by Lemma~\ref{lem:bounded_gap}. 
    
    Similarly, if $uv \in E^-$, we have \(d(u,v)=L = 1 + \epsilon(1+\frac{1}{R})\) and $1 + \epsilon(1+\frac{1}{R}) \leq 1 + \epsilon|c_u - c_v| \leq 1 + \epsilon R = M < DL = Dd(u,v)$ since we have $1 + \frac{1}{R} \leq |c_u - c_v| \leq R$ by Lemma~\ref{lem:bounded_gap}. 
    
    When \(uv\notin E^+\cup E^-\), we have \(d(u,v)=1\) and $1\le \rho(u,v) =  1 + \epsilon|c_u - c_v|\leq 1 + \epsilon R= M < D$ by Lemma~\ref{lem:bounded_gap}. 

    For the pairs involving anchor points,  let $v \in V$ be any vertex. We have:
    \[
    \rho(p,r)=\rho(q,r) = D\cdot \frac D2=D\,d(p,r)=D\,d(q,r),
    \]
    \[
    \rho(r,s) = \rho(r,v) =d(r,v) = d(r,s),
    \]
    \[
    \rho(p,v)=\rho(q,v) =\rho(s,v) = \rho(p,q) = \rho(s,q) =D=Dd(p,v)=Dd(q,v)= Dd(s,q) = d(p,q) = d(s,v)
    \]
    \[
    \rho(p,s) = D(D-1) = Dd(p,s),
    \]
giving the claim.
\end{proof}

\begin{claim}
    $(\rho, X)$ is an Okamura-Seymour metric with respect to the cyclic order $\sigma_X = \langle p, s, v_1,\dots, v_n,  q, r \rangle$.
\end{claim}

\begin{proof}
    It is sufficient to prove that $\rho$ satisfies the Monge property for any quadruple $(x_1, x_2, x_3, x_4) \subseteq X$ appearing in that order in $\sigma_X$.
    We proceed by cases:

    \paragraph{Case 1: All four vertices are from $V$:} Monge property is satisfied trivially by all quadruples of this type as $\rho|_V$ is a line metric by construction. 

    \paragraph{Case 2: Three vertices are from \(V\) and one vertex is from \(\{p,s,q,r\}\):} Suppose that $v_i,v_j,v_k$ ($\forall i < j< k$) are arbitrary three vertices from $V$ and $x$ is any vertex from $\{p,s,q,r\}$. Observe that the crossing pairs are always going to be $xv_j$ and $v_iv_k$ no matter where $x$ is located, and the non-crossing pairs are always going to be $xv_i$ and $v_jv_k$, or $xv_k$ and $v_iv_j$. Moreover, $\rho(x,v_j) = \rho(x,v_i) = \rho(x,v_k)$ for all possible $x$. Thus, we only need to show that $\rho(v_i,v_k) \ge \rho(v_i,v_j)$ and $\rho(v_i,v_k) \ge \rho(v_j,v_k)$. Since we have $c_i \leq c_j \leq c_k$, we must have $\rho(v_i,v_k) =  1 + \epsilon|c_i - c_k| \geq 1 + \epsilon|c_i - c_j| = \rho(v_i,v_j)$ and $\rho(v_i,v_k) =  1 + \epsilon|c_i - c_k| \geq 1 + \epsilon|c_j - c_k| = \rho(v_j,v_k)$. Hence, the Monge property also holds in this case.

    \paragraph{Case 3: Two vertices are from \(V\) and two vertices are from \(\{p,s,q,r\}\):} Let $u,v$ be any two vertices from $V$ and $x,y$ be any two vertices from $\{p,s,q,r\}$. Without loss of generality, the crossing pair can always be written as $xu$ and $yv$. One of the non-crossing pairs can always be written as $xv$ and $yu$. Observe that any vertex from $\{p,s,q,r\}$ has the same distance to any vertex in $V$ by the construction of $\rho$. Thus, the sum of $xu$ and $yv$ always equals to $xv$ and $yu$, and the Monge property is satisfied trivially. 
    
    It suffices to consider the other non-crossing pair $xy$ and $uv$. We always have $\rho(u,v) = 1 + \epsilon|c_u - c_v| \leq 1 + \epsilon R = M = \frac{6}{5}$ by Lemma~\ref{lem:bounded_gap}. Then, we consider the following two cases. First, if $r \notin \{x,y\}$, observe that we always have $\rho(x,u) = \rho(y,v) = D$, i.e., the sum of the crossing pair equals to $2D$. Since $\rho(p,q) = D, \; \rho(p,s) = D(D-1), \rho(q,s) = D$, the maximum sum of the non-crossing pair equals to $D + M < 2D$, and the Monge property is satisfied correspondingly. Second, if $r \in \{x,y\}$, the sum of the crossing pair $xu$ and $yv$ always equals to $\frac{11}{10} + D$. Since the maximum $\rho(x,y) = \rho(r,s) = \frac{11}{10}$, the maximum sum of the non-crossing pair equals to $\frac{11}{10} + M < \frac{11}{10} + D$, and the Monge property is satisfied in this case.

    \paragraph{Case 4: At most one vertex is from \(V\):} Finally, it remains to check quadruples with at most one vertex in \(V\). It suffices to check
    the five-point order $p,\ s,\ v,\ q,\ r$, where \(v\in V\). The crossing sum is at least both non-crossing sums in each row below:
    
    \[
    \begin{array}{c|c|c|c}
    \text{quadruple}
    &
    \text{crossing}
    &
    \text{non-crossing 1}
    &
    \text{non-crossing 2}
    \\
    \hline
    p,s,v,q
    &
    D+D=\frac{14}{5}
    &
    D(D-1)+D=\frac{49}{25}
    &
    D+D=\frac{14}{5}
    \\[1ex]
    p,s,v,r
    &
    D+\frac{11}{10}=\frac{5}{2}
    &
    D(D-1)+\frac{11}{10}=\frac{83}{50}
    &
    \frac{D^2}{2}+D=\frac{119}{50}
    \\[1ex]
    p,s,q,r
    &
    D+\frac{11}{10}=\frac{5}{2}
    &
    D(D-1)+\frac{D^2}{2}=\frac{77}{50}
    &
    \frac{D^2}{2}+D=\frac{119}{50}
    \\[1ex]
    p,v,q,r
    &
    D+\frac{11}{10}=\frac{5}{2}
    &
    D+\frac{D^2}{2}=\frac{119}{50}
    &
    \frac{D^2}{2}+D=\frac{119}{50}
    \\[1ex]
    s,v,q,r
    &
    D+\frac{11}{10}=\frac{5}{2}
    &
    D+\frac{D^2}{2}=\frac{119}{50}
    &
    \frac{11}{10}+D=\frac{5}{2}
    \end{array}
    \] 
    
    \bigskip
    
    If none of the vertices are from $V$, then the quadruple satisfies Monge by construction. This completes the proof that $(\rho, X)$ is an Okamura-Seymour metric.

\end{proof}

\subsection{Proof of Lemma~\ref{lem:bounded_gap}}\label{app:proof-bounded_gap}

\boundedgap*

\begin{proof}

Let \(H\) be a unit-interval graph. Suppose that $\gamma>0$. We construct the following system. 
\begin{align}
c_j-c_i &\le 1
&&\text{if } i<j \text{ and } v_i v_j\in E(H), \label{eq:bg-edge}\\
c_j-c_i &\ge 1+\gamma
&&\text{if } i<j \text{ and } v_i v_j\notin E(H), \label{eq:bg-nonedge}\\
c_i &\le c_{i+1}
&&\text{for } i=1,\ldots,n-1. \label{eq:bg-order}
\end{align}

The proof is contained in the following series of claims.

\begin{claim}\label{clm:feasibility_gamma_0}
The system \eqref{eq:bg-edge}--\eqref{eq:bg-order} is feasible for some $\gamma_0 >0$.
\end{claim}

\begin{proof}
Since \(H\) is a unit-interval graph, there exists a representation by unit intervals.
Using the centers of this representation, every edge \(v_i v_j\in E(H)\) with \(i<j\)
satisfies $c_j-c_i\le 1$. Every nonedge \(v_i v_j\notin E(H)\) satisfies $c_j-c_i>1$.
There are only finitely many nonedges, so the minimum positive gap over all nonedges is
strictly positive. Hence there exists \(\gamma_0>0\) such that $c_j-c_i\ge 1+\gamma_0$
for every nonedge \(v_i v_j\notin E(H)\) with \(i<j\). Finally, the ordering by centers gives $c_i\le c_{i+1}$. Hence, the system is feasible for this \(\gamma_0\).
\end{proof}

We rewrite the constraints \eqref{eq:bg-edge}--\eqref{eq:bg-order} in the following form:
\[
c_j-c_i\le 1,
\qquad
c_i-c_j\le -1-\gamma,
\qquad
c_i-c_{i+1}\le 0.
\]

We build the associated difference-constraint graph \(G_\gamma\) as follows. The graph has one vertex for each variable \(c_i\). For every constraint $c_b-c_a\le w$, 
we add a directed edge from $a$ to $b$ of weight $w$ in $G_\gamma$. 

\begin{claim}\label{clm:non_negative_cycle}
The system \eqref{eq:bg-edge}--\eqref{eq:bg-order} is feasible for a given \(\gamma\) if and only if the associated difference-constraint graph \(G_\gamma\) has no negative directed cycle.
\end{claim}

\begin{proof}
We first prove the left-to-right direction by contradiction. Suppose the system is feasible for $\gamma$, and there is a directed cycle $i_0\to i_1\to \cdots \to i_k  \to i_0$ of negative total weight in \(G_\gamma\). Since we have an edge from $i_\ell$ to $i_{\ell+1}$ of weight $w(i_\ell,i_{\ell+1})$ in \(G_\gamma\), we have a corresponding constraint $c_{i_{\ell+1}}-c_{i_\ell}\le w(i_\ell,i_{\ell+1})$. Summing all these inequalities over the cycle, the left-hand side telescopes to $0$. Thus, we have $0 \le \sum_{\ell=0}^{k-1} w(i_\ell,i_{\ell+1})$, which contradicts that this is a negative cycle.

We now prove the right-to-left direction. Suppose \(G_\gamma\) has no negative directed cycle. Add a new source vertex \(s_0\) and add a zero-weight edge from $s_0$ to every vertex \(i\) of \(G_\gamma\). Since \(G_\gamma\) has no negative directed cycle and we have only added some non-negative edges, the new graph also has no negative directed cycle, and hence the shortest-path distance $\operatorname{dist}(s_0,i)$ is well-defined for every \(i\), and the triangle inequality holds for this new graph. Thus, we have $\operatorname{dist}(s_0,b)\le\operatorname{dist}(s_0,a)+w(a,b)$, i.e., $\operatorname{dist}(s_0,b) - \operatorname{dist}(s_0,a) \le w(a,b)$. We let $c_i=\operatorname{dist}(s_0,i)$ and this gives $c_b-c_a\le w(a,b)$, which satisfies all the constraints in \eqref{eq:bg-edge}--\eqref{eq:bg-order}. Hence, the system is feasible for $\gamma$.
\end{proof}

\begin{claim}\label{clm:feasibility_gamma_1}
    The system \eqref{eq:bg-edge}--\eqref{eq:bg-order} is feasible for $\gamma_1=\frac{1}{n+1}$.     
\end{claim}

\begin{proof}
Since the system is feasible for some \(\gamma_0\) by claim~\ref{clm:feasibility_gamma_0}, the graph \(G_{\gamma_0}\) has no negative directed cycle by claim~\ref{clm:non_negative_cycle}. Consider any simple directed cycle in \(G_{\gamma_0}\). Its total weight has the form $A-k\gamma_0$, where \(A\in\mathbb{Z}\), and \(k\) is the number of edges induced by constraint \eqref{eq:bg-nonedge} used by the cycle. It suffices to consider every simple directed cycle in \(G_{\gamma_0}\), and thus $0 \le k \le n$. Every possible $A-k\gamma_0$ should be nonnegative, and we distinguish two cases. If \(k=0\), then the cycle weight is simply \(A\), we can safely substitute $\gamma_0$ with $\gamma_1$, and the cycle weight remains non-negative. If \(k>0\), then $A\ge k\gamma_0>0$ since $\gamma_0$ is positive, which implies $A\ge 1$. Since $k \le n$, $A-k\gamma_1 \ge 0$ always holds. 

Hence, $G_{\gamma_1}$ has no negative cycles. This implies system \eqref{eq:bg-edge}--\eqref{eq:bg-order} to be feasible when $\gamma_1 = 1 / (n+1)$ by claim~\ref{clm:non_negative_cycle}. 
\end{proof}

At this point, we slightly modify the original system by adding a new constraint for bounding the span as follows. 

\begin{align}
c_j-c_i &\le 1
&&\text{if } i<j \text{ and } v_i v_j\in E(H), \label{eq:bg-edge_2}\\
c_j-c_i &\ge 1+\gamma
&&\text{if } i<j \text{ and } v_i v_j\notin E(H), \label{eq:bg-nonedge_2}\\
c_i &\le c_{i+1}
&&\text{for } i=1,\ldots,n-1, \label{eq:bg-order_2}\\
c_n - c_1 &\le \frac{1}{\gamma}. \label{eq:bg-span_2}
\end{align}

\begin{claim}\label{clm:feasibility_lemma}
    The system \eqref{eq:bg-edge_2}--\eqref{eq:bg-span_2} is also feasible for $\gamma_1 = \frac{1}{n+1}$. 
\end{claim}
\begin{proof}
Recall that the associated graph $G_{\gamma_1}$ has no negative cycles by Claim~\ref{clm:feasibility_gamma_1}. Add a new source vertex \(s_1\) and add a directed edge from $s_1$ to every variable vertex $i$ of weight $0$. This new graph has no negative directed cycle either, and these shortest-path distances are well-defined.

Suppose that we have a directed edge from $a$ to $b$ induced by \eqref{eq:bg-edge_2}--\eqref{eq:bg-order_2}, we must have $\operatorname{dist}(s_1,b) \le \operatorname{dist}(s_1,a)+w$ by triangle inequality. Thus, by letting $c_v = \operatorname{dist}(s_1,v)$ for every vertex $v$, all the constraints in \eqref{eq:bg-edge_2}--\eqref{eq:bg-order_2} are satisfied.

It remains to verify the new span constraint. For every vertex \(v\), the zero-weight edge from $s_1$ to $v$ gives $c_v=\operatorname{dist}(s_1,i)\le 0$. Since we have no negative cycles, the shortest path from $s_1$ to $v$ is simple and goes through at most $n$ vertices, where each edge has weight at least $-\left(1+\frac{1}{n+1}\right)$. Thus, $c_v\ge -n\left(1+\frac{1}{n+1}\right)$. Hence, all the center vertices lie in the interval $\left[-n\left(1+\frac{1}{n+1}\right),0\right]$. This implies that constraint~\eqref{eq:bg-span_2} is also satisfied since $c_n - c_1 \leq n\left(1+\frac{1}{n+1}\right) < n+1 = \frac{1}{\gamma_1}$. 
    
\end{proof}

Claim~\ref{clm:feasibility_lemma} directly proves Lemma~\ref{lem:bounded_gap}.
\end{proof}

\subsection{Proof of \Cref{rho-is-metric}}\label{app:rho-is-metric}

\rhostmetric*

\begin{proof}
Let \(A=\{p,q,s,r\}\). We set \(\rho(x,x)=0\) for every \(x\in V\cup A\). It is immediate that \(\rho\) is nonnegative, symmetric, and that \(\rho(x,y)=0\) iff \(x=y\). It remains to prove the triangle inequality.

\paragraph{Case 1: Triangles in $V$:} $\rho|_V$ is the sum of a line metric and a discrete metric and thus $\rho$ satisfies the triangle inequality for triangles in $V$.

\paragraph{Case 2: Triangles with one vertex in $\{p,q,s,r\}$:} Let $x \in \{p,q,s,r\}$. Note that for $v \in V$ the value of $\rho(x, v)$ is independent of $v$. Specifically:
\[
    \rho(p,v)=\rho(q,v)=\rho(s,v)=D=\frac75,
    \qquad
    \rho(r,v)=\frac{11}{10}.
\]
For any $u,v\in V$, we have
\[
    \rho(u,v)\le 1+\epsilon R\le D=\frac75
    \le \frac{11}{5}
    \le \rho(u,x)+\rho(x,v).
\]
The other two triangle inequalities are immediate because $\rho(x,u)=\rho(x,v)\le \rho(x,v)+\rho(v,u)$, and symmetrically for \(\rho(x,v)\).

\paragraph{Case 3: Triangles with two vertices in $\{p,q,s,r\}$:} Let $x, y \in \{p,q,s,r\}$. We write $\alpha_x := \rho(x,v)$, and:
\[
    \alpha_p=\alpha_q=\alpha_s=D=\frac75,
    \qquad
    \alpha_r=\frac{11}{10}.
\]
It suffices to check that $|\alpha_a-\alpha_b|\le \rho(a,b)\le \alpha_a+\alpha_b$, which is made clear in the following table:
\[
\begin{array}{c|c|c|c}
\{a,b\} & \rho(a,b) & |\alpha_a-\alpha_b| & \alpha_a+\alpha_b \\ \hline
\{p,q\} & \frac75 & 0 & \frac{14}{5} \\
\{q,s\} & \frac75 & 0 & \frac{14}{5} \\
\{p,s\} & \frac{14}{25} & 0 & \frac{14}{5} \\
\{s,r\} & \frac{11}{10} & \frac{3}{10} & \frac52 \\
\{p,r\} & \frac{49}{50} & \frac{3}{10} & \frac52 \\
\{q,r\} & \frac{49}{50} & \frac{3}{10} & \frac52 .
\end{array}
\]

\paragraph{Case 4: Triangles in $\{p,q,s,r\}$:} This is by construction; we give the table directly:
\[
\begin{array}{c|c|c}
\text{triple} & \text{largest side} & \text{sum of the other two sides} \\ \hline
\{p,q,s\} & \frac75 & \frac75+\frac{14}{25}\\
\{p,q,r\} & \frac75 & \frac{49}{50}+\frac{49}{50}\\
\{p,s,r\} & \frac{11}{10} & \frac{14}{25}+\frac{49}{50}\\
\{q,s,r\} & \frac75 & \frac{11}{10}+\frac{49}{50}.
\end{array}
\]
Thus $\rho$ satisfies the triangle inequality, which completes the proof.
\end{proof}

\end{document}